\documentclass[11pt,reqno]{amsart} 
\usepackage{amsaddr}
\usepackage[utf8]{inputenc}
\usepackage[a4paper, margin=3.8cm]{geometry}

\usepackage{amsmath}
\usepackage{amsfonts}

\usepackage{amssymb}
\usepackage{amsthm}
\usepackage{mathtools}
\usepackage{paralist}
\usepackage{cancel}

\usepackage{natbib}
\usepackage{xcolor}
\usepackage[colorlinks=true,linkcolor=blue,citecolor=blue,pdfborder={0 0 0}]{hyperref}

\usepackage{booktabs}
\usepackage{geometry}

\usepackage{dsfont}
\usepackage{bm}

\usepackage{graphicx}
\usepackage{enumerate}
\usepackage{subcaption}

\theoremstyle{plain}
\newtheorem{rmk}{Remark}
\newtheorem{Lem}{Lemma}
\newtheorem{Theo}{Theorem}
\newtheorem{Prop}{Proposition}

\newcommand{\Mu}{$\mathrm{(M_1)}$}
\newcommand{\Md}{$\mathrm{(M_2)}$}
\newcommand{\Au}{${\mathrm (A_1)}$}
\newcommand{\Ad}{${\mathrm (A_2)}$}
\newcommand{\At}{${\mathrm (A_3)}$}
\newcommand{\Aq}{${\mathrm (A_4)}$}
\newcommand{\brem}{\begin{rmk}}
	\newcommand{\erem}{\end{rmk}}
\newcommand{\indin}{\mathds{1}_{\{Y_i\geq y_n\}}}
\newcommand{\indn}{\mathds{1}_{\{Y\geq y_n\}}}
\newcommand{\ind}{\mathds{1}_{\{Y\geq y\}}}

\newcommand{\cov}{\operatorname{cov}}

\newcommand{\var}{\operatorname{var}}

\newcommand{\vfi}{\varphi}
\newcommand{\E}{\mathbb{E}}
\newcommand{\R}{\mathbb{R}}
\newcommand{\f}{\frac}
\newcommand{\ff}{\frac{1}}

\newcommand{\p}{\mathbb{P}}

\begin{document}


\title[Extreme-PLS with missing data under weak dependence]{Extreme-PLS with missing data under weak dependence}
\author{Stéphane Girard\textsuperscript{1} and Cambyse Pakzad\textsuperscript{2}}
\address{
  \textsuperscript{1}Univ. Grenoble Alpes, Inria, CNRS, Grenoble INP, LJK, 38000 Grenoble, France \\
  \textsuperscript{2}Modal’X, Université Paris Nanterre, 200 Avenue de la République, 92001 Nanterre, France  
}
\email{\textsuperscript{1}stephane.girard@inria.fr, \textsuperscript{2}cambyse.pakzad@parisnanterre.fr}
\maketitle 

\begin{abstract}
This paper develops a theoretical framework for \emph{Extreme Partial Least Squares} (EPLS) dimension reduction in the presence of missing data and weak temporal dependence. Building upon the recent EPLS methodology for modeling extremal dependence between a response variable and high-dimensional covariates, we extend the approach to more realistic data settings where both serial correlation and missing-ness occur. Specifically, we consider a single-index inverse regression model under heavy-tailed conditions and introduce a \emph{Missing-at-Random} (MAR) mechanism acting on the covariates, whose probability depends on the extremeness of the response. The asymptotic behavior of the proposed estimator is established within an $\alpha$-mixing framework, leading to consistency results under regularly varying tails. Extensive Monte-Carlo experiments covering eleven dependence schemes (including ARMA, GARCH, and nonlinear ESTAR processes) demonstrate that the method performs robustly across a wide range of heavy-tailed and dependent scenarios, even when substantial portions of data are missing. A real-world application to environmental data further confirms the method's capacity to recover meaningful tail directions.\\

\end{abstract}
\section{Introduction}

In modern statistical applications, analysts increasingly face interconnected challenges that defy classical assumptions. Indeed, the data collected are often complex, sometimes arising as time series exhibiting serial dependence, with numerous predictors/covariates, which weights on statistical procedures and their computational cost. Furthermore and realistically, datasets often suffer from missing data, particularly in the extreme values of the response variable -- the very observations often most critical to study. Together, these issues build up a compounded sparsity with exacerbated curse of dimensionality by the rarity of extreme events, which in turn may be under-recorded because of the missing-ness mechanism. This convergence undermines standard statistical methods and creates a pressing need for tailored robust techniques.

Addressing the challenge of raw sparsity due to high-dimensions, the field of sufficient dimension reduction (SDR) seeks a low-dimensional projection $P_SX$ of the high-dimensional covariate $X$ preserving the information about the response variable $Y$. More precisely, it is such that $Y$ conditionally to $X$ has same distribution as $Y$ conditionally to $P_SX$, or equivalently $Y$ and $X$ are independent conditionally to $P_SX$ (\citealp{Cook2007}). A classic example is the single-index model, which assumes dim$(S)=1$ and therefore $Y$ only depends on a single linear combination of $X$ with a deterministic vector~$\beta$. Next, one would substitute $X$ by $P_SX$ in further statistical methods as it carries the same amount of information on $Y|X$ without the dimensionality curse.

Two prominent SDR methods are Partial Least Squares (PLS), initiated by \cite{Wold1975}, and Sliced Inverse Regression (SIR), introduced by \cite{Li1991}. PLS finds projections with high variance and high covariance with $Y$, meaning that the information in $X$ that impacts $Y$ is encompassed in $P_SX$, while being the most exhaustive. SIR, in contrast, leverages the simpler inverse regression '$X$ against $Y$', usually at the price of a certain linearity condition and constant variance. This inverse model, where the inference is rather conditioned on the predictors/covariates $X$, somewhat contradicts the philosophy of Fisher, as noted but also encouraged by \cite{Cook2007}.

A nascent research area aims to merge SDR with extreme-value theory 
since nonparametric estimators of extreme conditional features~\citep{Annals} may be impacted both by the scarcity of extremes and the high-dimensional setting.
The challenging task is to capture the information relevant specifically to the response tails.
Pioneering work by \cite{Gardes2018} introduced a notion of tail conditional independence. Later, \cite{xu2022extreme} and  \cite{sabourin} developed estimators for extreme quantiles and tail inverse regression, respectively. However, these approaches largely assume independent data and a fully observed response, conditions that are often violated in practice.

This paper bridges this gap by establishing the theoretical foundations for dimension reduction in the extremes under more realistic conditions: serial dependence and missing data. Our focus is on Extreme-PLS (EPLS, \citet{Bousebata2023}) relying on PLS 
principles for estimating the linear combination $\beta^\top X$ of $X$ that best explains the extreme values of $Y$. The asymptotic properties of the estimated direction $\hat \beta$ are established assuming a single-index inverse regression model and extreme-value type conditions. Here, the asymptotic properties of the estimator are rather derived under the more realistic context of dependent and missing data.
Specifically, a Missing At Random (MAR, \citet{rubin1976inference}) mechanism is assumed on the covariate $X$, the missing-ness probability depending on the extreme-ness of the response variable $Y$. This, for instance, captures scenarios where the failure of a sensor (leading to a missing covariate) is more likely during an extreme weather event.
We refer to \cite{beirlant2023estimation,xu2022handling} for recent works on the estimation of tail features in the case where the missing-ness mechanism is assumed on the response variable $Y$.
Finally, an $\alpha$-mixing framework is adopted both for the observations $(X,Y)$ and the masking process associated with the missing data, in order to account for the temporal dependence common in real-world data like financial or environmental time series.

The remainder of the paper is organized as follows. Some links between  EPLS and the Marginal Expected Shortfall (MES, \citet{Cai}) risk measure are established in Section~\ref{sec-approche}.  The missing data and dependence frameworks are described in Section~\ref{sec-framework} and the EPLS estimator is derived. 
Its asymptotic properties are established in Section~\ref{sec-esti} under tail conditions while its finite sample behavior is illustrated on large scale experiments in Section~\ref{sec-simul} and Section~\ref{sec-reel}. The first 
part of the experiments is conducted on simulated data involving eleven different dependence schemes. The second part illustrates the performance of the method compared to four competitor baselines on environmental data involving measures at nearly 5800 triplets of weather stations. Proofs are postponed to the Appendix.

\section{EPLS and Marginal Expected Shortfall}
\label{sec-approche}

Let us consider the vector $w(y)$ referred to as the EPLS direction in \cite{Bousebata2023}. It is defined as
the unit vector maximizing over $\| w\|=1$ the covariance between $Y$ and the projection $w^\top X$ of $X$ on $w$ given that $Y$ exceeds a large threshold $y>0$: 
 \begin{equation}
 \label{covar}
  w(y)= \arg\max_{\|w\|=1} {\rm cov} (w^\top X,Y \mid Y\geq y),
 \end{equation}
 where $\|\cdot\|$ denotes the Euclidean norm. 
 We also introduce the marginal expected shortfall of $X$ at level $y$ which plays a central role in this study: 
 $$
 \textnormal{MES}_X(y) = \mathbb{E}(X \mid Y > y),
 $$
see \cite{Cai}. Note that $\textnormal{MES}_X(y)\in\mathbb{R}^p$ and
the associated $j$th coordinate is given by
$\textnormal{MES}_X^{(j)}(y)=\textnormal{MES}_{X^{(j)}}(y)$.
This risk measure may alternatively be expressed thanks to the tail-moment defined as $m_X(y) = \mathbb{E}(X \ind)$, whenever it exists. This notation will be kept throughout our work. In fact, denoting $\bar{F}(y) = \mathbb{P}(Y \geq y)$ the survival function of $Y$, one may write $\textnormal{MES}_X(y)=m_X(y)/\bar{F}(y)$. Let us also remark that 
$$
 \textnormal{MES}_Y(y) = \mathbb{E}(Y \mid Y > y)=:\text{ES}_Y(y)
 $$
 is the so-called Expected Shortfall (ES).
The linear optimization problem \eqref{covar} under a quadratic constraint benefits from a closed-form solution  obtained with Lagrange multipliers method and involving the cross-MES between $X$ and $Y$. The next proposition precises this point, it is a simple rewriting of \cite[Proposition~1]{Bousebata2023}.
\begin{Prop}
\label{prop-w}
Suppose that $\mathbb{E}(\|X\|\ind)<\infty$, $\mathbb{E}(|Y|\ind)<\infty$ and $\mathbb{E}(\|XY\|\ind)<\infty$ for all $y>0$. Then, the unique solution of the optimization problem~(\ref{covar}) is given for all $y>0$ by: 
\begin{align}
\label{solution}
w(y) &= \frac{v(y)}{\|v(y)\|}, \quad \text{where} \quad v(y) := \textnormal{MES}_{XY}(y) - \textnormal{MES}_X(y) \textnormal{ES}_Y(y).
\end{align}
\end{Prop} 

In the following, we aim at investigating the behavior of $w(y)$ for large thresholds $y$ without resorting  neither to a linear conditional
expectation assumption as in~\cite{xu2022extreme} nor to a conditional independence assumption as in~\cite{Gardes2018,Saracco}. In contrast, additional assumptions on the  distribution tails are introduced below.

\section{Statistical framework}
\label{sec-framework}

 Here and throughout, $\odot$ and $/$ are respectively the componentwise product and ratio.
For a generic random vector $Z\in \R^p$, let us recall that $Z^{(j)}$ denotes the $j$th coordinate of $Z$ and
the $j$th coordinate of $m_{Z}$ is then given by $m^{(j)}_Z=m_{Z^{(j)}}$, $j\in\{1,\dots,p\}$. 

\paragraph{\bf Model and tail conditions.} Let us consider the single-index inverse regression model introduced in \cite{Bousebata2023}:     
\begin{description}
    \item [\Mu~] $X= g(Y) \beta  + \varepsilon$ 
where $X$ and $\varepsilon$ are $p-$dimensional random vectors, $g:{\mathbb R}\to {\mathbb R}$ is an unknown link function and $\beta \in {\mathbb R}^p$ is the unknown direction of interest.
\end{description}
Model \Mu~is an inverse regression model since the covariates $X=(X^{(1)},\dots,X^{(p)})^\top$ are written as functions of the response variable $Y$.
A Bayesian version of EPLS is introduced in~\cite{STCO7} basing on \Mu~while similar models were used to establish the theoretical properties of SIR, see for instance~\cite{BGG2009,Cook2007}.

Three assumptions on the link function $g$ and the distribution tail of $Y$ and $\varepsilon$ are considered. They rely on the notion of regularly-varying functions.
Recall that $\varphi\in RV_\theta$ with $\theta\in\mathbb{R}$ if and only if $\varphi$ is positive and
    $$
    \lim_{y\to\infty} \frac{\varphi(ty)}{\varphi(y)}=  t^{\theta},
    $$
    for all $t>0$.
        This property is denoted for short by $\varphi\in RV_{\theta}$.
We refer to~\cite{Bing1989} for a detailed account on regular variations.
\begin{description} 
    \item [\Au~] $Y$ is a real random variable with density function $f\in RV_{-1/\gamma-1}$ with $\gamma\in (0,1)$.

    \item [\Ad~] $g\in RV_{\kappa}$ with $\kappa> 0$. 
    
    \item [\At~] There exists $q\in (1,+\infty)$ such that $\mathbb{E}(\|\varepsilon\|^q)<\infty$.

\end{description} 
Assumption \Au~implies that the distribution of $Y$ is in the Fr\'echet maximum domain of attraction with positive tail-index $\gamma$,
see \cite[Theorem~1.5.8]{Bing1989} and  \cite[Theorem~1.2.1]{Haan2007}. This domain of attraction consists of heavy-tailed distributions, such as Pareto, Student~$t$ or Burr distributions, see~\cite{beigoesegteu2004} for other examples.
The larger $\gamma$ is, the heavier the tail.
The restriction to $\gamma<1$ ensures that the first-order moment { $\mathbb{E}(|Y|\ind)$ exists for all $y>0$}.
Assumption \Ad~states that $g$  ultimately behaves like a power function.
Assumption \At~can be interpreted as an assumption on the tails of $\|\varepsilon\|$. It is satisfied, for instance, by distributions with exponential-like tails such as Gaussian, Gamma or Weibull distributions.

\paragraph{\bf Missig data framework.} We consider the situation where some coordinates of the covariate $X$ may not be observed. A random mask 
$\Lambda\in \{0,1\}^p$ is applied on the covariate, such that only the $\Lambda\odot X$ part of $X$ is available for estimation. The missing at random \citep{rubin1976inference} mechanism is adopted here,
{\it i.e.} $\Lambda$ is dependent of $Y$ but not of $X$:
\begin{description}
    \item [\Aq~] $\Lambda$ is a $\{0,1\}^p$ random vector such that, for all $j\in\{1,\dots,p\}$, $\Lambda^{(j)}\mid Y$ are independent Bernoulli random variables with success parameter ${\lambda_j}(Y)=c_j \lambda(Y)$ with $\lambda\in RV_{\tau}$, $\tau<0$ and $c_j>0$.
\end{description}
Note that the rate of non-missing data may depend on the coordinate $j$ through the multiplicative constant $c_j>0$ and decreases in the distribution tail
as a power function with exponent $\tau<0$. We refer to \cite{mayer2019r} for more information on missing data.
Let $(X_i, Y_i, \Lambda_i, \varepsilon_i)_{1\leq i\leq n}$ be a random sequence sampled from a random vector $(X,Y,\Lambda,\varepsilon)$ under model \Mu. Some vectors $X_i$ may not be actually observed because of the missing-ness and our true observations set is thus $\left(Y_i,\Lambda_i \odot X_i \right)_{1\le i \le n}$.

\paragraph{\bf Serial dependence.} We adopt the $\alpha$-mixing framework for its generality, see \cite[Proposition~1 and diagram page~20]{Doukhan1994}. For any two positive integers $a\le b \le +\infty$, let $\mathcal{F}^{b}_{a}$ be the $\sigma$-algebra generated by $\{ (X_i,Y_i,\Lambda_i,\varepsilon_i)\}_{a\le i \le b}$. The sequence $\{ (X_i,Y_i,\Lambda_i,\varepsilon_i)\}_{1\le i \le n}$ is said $\alpha$-mixing if $\alpha(n)\to 0$ as $n\to \infty$ where
 \begin{align*}
  \alpha(n) &:=  \sup_{k\ge 1}\sup_{A\in \mathcal{F}_1^k}\sup_{B\in \mathcal{F}_{k+n}^\infty} |\mathbb{P}(A\cap B) - \mathbb{P}(A)\mathbb{P}(B)|.
\end{align*}
To setup the standard 'small-block/large-block' argument (see \cite{DreesRootzen2010}), the following assumptions are introduced:
\begin{description}
    \item [\Md~] $(X_i, Y_i, \Lambda_i, \varepsilon_i)_{1\leq i\leq n}$ is a strictly-stationary $\alpha$-mixing time serie, sampled from the random vector $(X,Y,\Lambda,\varepsilon)$ under model \Mu, and there exist two sequences of integers $1\ll \ell_n\ll r_n \ll n$ and $ n\alpha(\ell_n)\ll r_n$,  such that 
\begin{align}\label{hyp:mixing}
   &\quad \sum_{j\ge 1 } \alpha^{\frac{\delta}{2+\delta}}(j)<\infty \mbox{ for some } \delta>0 .
\end{align} 
\end{description}
While conditions $1\ll \ell_n\ll r_n \ll n$ and $ n\alpha(\ell_n)\ll r_n$ are at the basis of the 'small-block/large-block' argument, assumption~\eqref{hyp:mixing} is standard in the literature of central limit theorems for strongly mixing processes, see \citet[Theorem~1.7]{Ibragimov62}.

\paragraph{\bf Inference.} Let $y_n\to\infty$ as the sample size $n$ tends to infinity.
 The EPLS with missing data~\eqref{solution} is estimated by the random direction $\hat \beta_\Lambda(y_n) := \hat v_{\Lambda}(y_n)/\| \hat v_{\Lambda}(y_n) \| \in \R^p$, based on 
\begin{align}
    \label{esti}
    \hat v_{\Lambda}(y_n) &:= \ff{\hat{m}_{\Lambda}(y_n)}\left(\hat{\bar{F}}(y_n) \hat{m}_{Y \Lambda \odot X}(y_n) - \hat{m}_{Y}(y_n)\hat{m}_{\Lambda \odot X}(y_n)\right),
\end{align} 
the ratio being understood component-wise with, for any random vector $Z\in \R^p$,
\begin{align*}
& \hat m_{Z}(y_n) := \frac{1}{n}\sum_{i=1}^n Z_i  \indin ,\
\end{align*}
and where $\hat{\bar F}:= \hat{m}_1$ is the empirical survival function.

\section{Asymptotic properties of the estimator}
\label{sec-esti}

We first provide a tool establishing the joint asymptotic properties of  $\hat m_{\Lambda \odot X}(y_n)$, 
$\hat m_{\Lambda}(y_n)$, $\hat m_{Y\Lambda \odot X}(y_n)$ and $\hat m_{Y}(y_n)$ when $y_n\to\infty$. This latter condition ensures that the rate of convergence, which  is driven by the number of effective observations $n(\lambda\bar F)^{\theta}(y_n)$ used in the estimators, tends to infinity as the sample size increases.

From now on, we denote by $d$ the number of non-zero components of $\beta$, supposed positive $d\ge 1$, namely $d:=|\{ j \in\{1,\dots, p\} :\beta^{(j)}\ne 0 \}|$. Without loss of generality and for simplicity, one may assume the support of $\beta$ to be exactly $\{1,\ldots,d\}$.

\begin{Prop}
\label{prop-loi-jointe-alpha-mixing}
Suppose \Md, \Au, \Ad, \At~and \Aq~hold.  Assume also that $$1<\gamma\min\left\{2\kappa+2q+\tau;q\kappa+2\tau;q\kappa/2+\tau\right\},$$
$$\gamma\max\Big\{\frac{q}{q-1}\kappa+\tau;\frac{q}{q-1}(\kappa+1)+2\tau;\frac{q}{q-2-\delta}+\tau;(2+\delta)\kappa+\tau+1; 2\kappa+\tau+2; \kappa+1;2+\delta\Big \}<1.$$
 Introduce $\Xi_n$ the ${\mathbb R}^{{3d+2}}$-random vector defined as:
$$
\left\{  \Big( \frac{\hat m_{\Lambda^{(j)}X^{(j)}}(y_n)}{m_{\Lambda^{(j)}X^{(j)}}(y_n)} -1\Big)_{1\leq j \leq d}, \Big( \frac{\hat m_{\Lambda^{(j)}}(y_n)}{m_{\lambda^{(j)}}(y_n)} -1\Big)_{1\leq j \leq d}, \Big(  \frac{\hat m_{Y\Lambda^{(j)}X^{(j)}}(y_n)}{m_{Y\Lambda^{(j)}X^{(j)}}(y_n)}-1\Big)_{1\leq j \leq d} , \Big(\frac{\hat{m}_{Y^\iota}(y_n)}{m_{Y^\iota}(y_n)}-1\Big)_{\iota\in \{0,1\}}   \right \}.
$$ 
For any $\theta>1+\frac{\delta}{2+\delta} \in (1,2)$ such that $n(\lambda\bar F)^{\theta}(y_n)\gg 1$ holds, the sequence $
 n^{1/2}(\lambda\bar{F})^{\theta/2}(y_n) \, \Xi_n  $ converges in probability to zero in $\mathbb{R}^{{3d+2}}$. When $\theta=1+\frac{\delta}{2+\delta}$, if furthermore 
$$\gamma\max\left\{(2+\delta)(\kappa+1)+\tau;\frac{q(2+\delta)}{q-(2+\delta)}+\tau\right\} <1,$$ then the sequence $
 n^{1/2}(\lambda\bar{F})^{\theta/2}(y_n) \, \Xi_n  $ converges in distribution, at least up to a subsequence, to a centered Gaussian random vector in $\mathbb{R}^{{3d+2}}$.
\end{Prop}

Proposition~\ref{prop-loi-jointe-alpha-mixing} holds in the weaker $\psi,\phi,\rho$-mixing frameworks, see \cite{Doukhan1994}, with a better rate that corresponds to $\delta=0$. Yet in any case, we found ourselves limited to only provide a non-explicit bound on the limiting covariance.
Our main result establishes that $\hat v_{\Lambda}(y_n)$ and $\beta$ are asymptotically aligned in $\R^p$, with respect to the convergence in probability and rate $n^{1/2}(\lambda\bar{F})^{\theta/2}(y_n)$. 
\begin{Theo}
\label{theo-princ}
Under the assumptions of Proposition~\ref{prop-loi-jointe-alpha-mixing}, for any $\theta>1+\frac{\delta}{2+\delta}$, it holds that:
\begin{align*}
\min\left\{n^{1/2}(\lambda\bar{F})^{\theta/2}(y_n)  , g(y_n)\bar{F}^{1/q}(y_n)\right\}\left(\hat{\beta}_{\Lambda}(y_n) -  \beta \right)\xrightarrow[n\to +\infty]{\mathbb{P}} 0,\quad \text{in $\mathbb{R}^{p}$.}
 \end{align*}
 \noindent When $\theta=1+\frac{\delta}{2+\delta}$, if moreover $n^{1/2}(\lambda\bar{F})^{\theta/2}(y_n)\ll g(y_n)\bar{F}^{1/q}(y_n)$ and $n g^{-2}(y_n)(\lambda \bar F)(y_n)^{\theta-2/q}\ll 1$, then, the following convergence holds, at least up to a subsequence,
 \begin{align*}
n^{1/2}(\lambda\bar{F})^{\theta/2}(y_n) \left(\hat{\beta}_{\Lambda}(y_n) -  \beta \right)\xrightarrow[n\to +\infty]{d} \beta\odot G,\quad \text{in $\mathbb{R}^{p}$,}
 \end{align*} 
 where $G$ is some centered Gaussian vector in $\R^p$. 
\end{Theo}

Note that, when the noise admits moments at any order, we may consider $q=+\infty$, thus boiling the conditions down to $\gamma\max\{(2+\delta)\kappa+\tau+1; 2\kappa+\tau+2; \kappa+1;2+\delta \}<1$.
Regarding the threshold regime associated with a finite value of $q$, let us consider $y_n:=n^c$ with $c\in (0,1)$. The property $n(\lambda\bar F)^\theta(y_n)\gg 1$ is satisfied for a generic $\theta>0$ whenever $c< \frac{\gamma}{\theta(1-\gamma\tau)}$. 
It is then checked through numerical computations that the set of all parameters $(q,\delta,\gamma,\kappa,\tau)$ satisfying the conditions needed for Theorem~\ref{theo-princ} is non-empty.   

\section{Validation on simulated data}
\label{sec-simul}

The large scale experiments are presented in Section~\ref{sub-expe}, some implementation details are provided in Section~\ref{sub-seuil} and the obtained results are discussed in Section~\ref{sub-results}.

\subsection{Experimental design}
\label{sub-expe}
The performance of our method is assessed by a Monte Carlo simulation experiment with $N=500$ independent repetitions.
The experimental design involves four setups for drawing a sample $(X_i,Y_i)_{1\le i \le n}$ with $n=500$ and $p=101$, see Paragraphes~\ref{par1}--\ref{par4}.
In any case, the deterministic $\beta$'s coordinates are set to $\beta^{(j)}=\sqrt{2}\sin(2\pi j / p)$ for $1\le j \le p $ and $g:y\in \R^+ \mapsto y^\kappa$ with $\kappa=0.5$ fixed for shortness. 
To balance the contributions of signal and noise, we propose to rescale the conditional noise variance by multiplying $\varepsilon$ with $g({{y}})/10$. Once both the response variable $Y$ and the noise $\varepsilon$ are simulated, and the deterministic functions $\beta$ and $g$ are chosen, the covariate sampling $\{X_1,\dots,X_n\}$ is readily derived from \Mu. The mask dynamic framework is also common to all setups and is described in Paragraph~\ref{par-mask}. Finally, some remarks on the tail heaviness and the strict stationarity of the processes are provided in Paragraph~\ref{par-tail} and Paragraph~\ref{par-station} respectively.

\subsubsection{First setup: univariate ARMA-resp./multivariate GARCH-noise}
\label{par1}
The output and noise random variables are distributed as follows:
\begin{itemize}
    \item $(Y_i)$ is a ARMA$(1,1)$-time serie, \textit{i.e.} $Y_i = \phi_{\rm resp} Y_{i-1} + \theta_{\rm resp}  \eta_{i-1} + \eta_i$, where the innovations $\eta_i$ are i.i.d.\,BurrXII-distributed with parameters $-\rho$ and $1/\gamma$, \emph{i.e.}, with survival c.d.f. $\bar{F}(y)=(1+y^{-\rho})^{1/\gamma}$ with ${y}\ge 0$, $\gamma \in (0,1)$ being the tail-index and $ \rho <0$ being the second-order parameter. Here, $\gamma\in\{ 0.1, 0.4, 0.7\}$. Throughout the whole section, $\rho=-1$ is fixed. Since ARMA$(1,1)$ is a linear process, the response $(Y_i)$ is also heavy-tailed with same tail-index $-1/\gamma$, akin to the innovations $(\eta_i)$. The ARMA$(1,1)$ parameters  $(\phi_{\rm resp},\theta_{\rm resp})$ are picked among $\{(0.8,-0.3),(0.99,-0.98) \}$. The standard case $(\phi_{\rm resp},\theta_{\rm resp}) = (0.8, -0.3)$ induces a stable and moderate mean-reverting process. In contrast, $(\phi_{\rm resp},\theta_{\rm resp}) = (0.99, -0.98)$ creates a pathological dynamic where a single shock/innovation can trigger a lengthy period of volatility and extreme values, making events highly persistent.
\item $(\varepsilon_i)$ conditionally to $Y=y$ is a GARCH$(1,1)$-multivariate time serie, \emph{i.e.}, $\varepsilon_i =\sigma_i \eta'_i$ and
    $\sigma_i= \omega_{\rm noise} + \alpha_{\rm noise}\varepsilon^2_{i-1} + \beta_{\rm noise}\eta'^2_{i-1}$ where the innovations $(\eta'_i)$ are i.i.d. centered Gaussian vectors in $\R^p$ with Toeplitz correlation matrix, \emph{i.e.}, ${\rm Corr}(\varepsilon_i,\varepsilon_j \mid Y=y)=\rho_c^{|i-j|}$ with $\rho_c=0.8$. Three sets of parameters are considered for the GARCH$(1,1)$ model of the noise: $(\omega_{\rm noise},\alpha_{\rm noise},\beta_{\rm noise})$ $\in$ $\{(0.05,0.1,0.85),(0.05,0.05,0.94)\}$. These configurations are chosen to capture a range of volatility dynamics. The standard case $(\omega_{\rm noise},\alpha_{\rm noise},\beta_{\rm noise}) = (0.05, 0.1, 0.85)$ yields persistent but ultimately transitory volatility clustering, meaning that a period of high variance is likely to be followed by another, but the memory of this volatility decays over time. The parameters $(\omega_{\rm noise},\alpha_{\rm noise},\beta_{\rm noise}) = (0.05, 0.05, 0.94)$ define a process with extreme heteroskedasticity. Here, $\alpha_{\rm noise} + \beta_{\rm noise} = 0.99$ is close to one, and the process almost behaves like an Integrated GARCH (IGARCH), see \citep{IGARCH,Nelson90}, where volatility shocks become permanent. 
\end{itemize}

\subsubsection{Second setup: univariate GARCH-resp./multivariate ARMA-noise}
Next, we move onto the second experimental setup by reversing the serial dependence structure of response and noise in the previous setup. 
\begin{itemize}
    \item The response $(Y_i)$ is an univariate GARCH$(1,1)$-process with i.i.d.~BurrXII$(-\rho,1/\gamma)$-distributed innovations. Concerning its parameters, we opt here for $(\omega_{\rm resp},\alpha_{\rm resp},\beta_{\rm resp})=(\omega_{\rm noise},\alpha_{\rm noise},\beta_{\rm noise})$, meaning that they are the same as when the noise was GARCH$(1,1)$, while the innovations are to be understood as univariate i.i.d. BurrXII$(-\rho,1/\gamma)$-distributed. 
    \item  The multivariate ARMA$(1,1)$-structure on the noise has parameters corresponding to the ARMA-response case, \emph{i.e.,} $(\phi_{\rm noise},\theta_{\rm noise})=(\phi_{\rm resp},\theta_{\rm resp})$. The innovations are similar to the GARCH-noise case, namely i.i.d. centered Gaussian vector in $\R^p$ with Toeplitz correlation matrix and $\rho_c=0.8$.
\end{itemize}

\subsubsection{Third setup: ESTAR-resp./GARCH-noise}
We consider the Exponential Smooth Transition Autoregressive (ESTAR) model \citep{Dijk2002} for the response variable to model nonlinear dynamics. Here, the ESTAR(1) process is defined as:$$Y_i = \phi_{\rm low} Y_{i-1} + \phi_{\rm high} Y_{i-1} \left[1 - \exp\left(- Y_{i-1}^2\right)\right] + \eta_i,$$
where $\eta_i \sim \text{BurrXII}(-\rho, 1/\gamma)$ are i.i.d. innovations and $(\phi_{\rm low}, \phi_{\rm high}) = (0.2, 0.95)$. Endowed with these parameters, the dynamics smoothly shift from a quickly mean-reverting or low-persistence regime ($\phi_{\rm low} = 0.2$) near equilibrium to a regime ($\phi_{\rm high} = 0.95$) where large deviations tend to persist during extreme events. Concerning the noise, the process follows the multivariate GARCH$(1,1)$ specification from the first setup.

\subsubsection{Fourth setup: Independence} 
\label{par4}
The independent case is considered with the following combination of parameters: $$(\phi_{\rm resp},\theta_{\rm resp})=(0,0),\qquad (\omega_{\rm noise},\alpha_{\rm noise},\beta_{\rm noise})=(1,0,0),\qquad \alpha_{\rm bar}=0.$$
In this configuration, extremes occur sporadically and independently. The volatility is constant/clusterless, namely, an extreme value  provides no information about the likelihood of another extreme value at a future time.

\subsubsection{Missing data mechanism}
\label{par-mask}
The mask $(\Lambda_i)$ is a multivariate Binary Autoregressive (or BAR) time serie, with marginal correction such that \Aq~is satisfied with $c_j=1$ and $\lambda(y)=y^\tau$, $\tau \in \{-0.1,-0.5,-0.9\}$. Precisely, the BAR dynamics with cross-sectional independence writes as, for any $1\le j \le p$ and $2\le i \le n$,
$$
\Lambda_1^{(j)} \sim {\rm Bernoulli}(p_1),\quad \pi_{i,j} = \alpha_{\rm bar} \Lambda_{i-1}^{(j)}+(1-\alpha_{\rm bar})p_i, \quad \mbox{and }\Lambda_i^{(j)} \sim {\rm Bernoulli}(\pi_{i,j}). 
$$
The parameters of the BAR process are $\alpha_{\rm bar} \in [0,1)$ which tunes the dependence and $p_i \in [0,1]$ the target Bernoulli parameter for $\Lambda_i^{(j)}$. In our case, we fix $\alpha_{\rm bar}=0.5$ and $p_i=\lambda(Y_i)$ which corresponds to $c_j=1$ in \Aq.  Nonetheless, each $\Lambda_i^{(j)}$ though Bernoulli-distributed has its mean different from $p_i$. To remedy this hindrance and fit into our theoretical framework, we use a correction mechanism based on a rejection sampling technique. It starts with an initialization $\Lambda_1^{(j)} \sim {\rm Bernoulli}(p_1)$. Next, for any $i\ge 2$, we compute $\pi_{i,j}$ and independently sample $\tilde{\Lambda}_i^{(j)} \sim {\rm Bernoulli}(\pi_{i,j})$ as candidate value and an uniform $U$. Then, one applies the asymmetric correction:
\begin{equation*}
\Lambda_{i}^{(j)} = 
\begin{cases} 
1 & \text{if } \pi_{i,j} < p_{i} \text{ and } U < \gamma^{+}_{i,j} \text{ and } \tilde{\Lambda}_{i}^{(j)} = 0, \\
0 & \text{if } \pi_{i,j} > p_{i} \text{ and } U < \gamma^{-}_{i,j} \text{ and } \tilde{\Lambda}_{i}^{(j)} = 1, \\
\tilde{\Lambda}_i^{(j)} & \text{otherwise.}  
\end{cases}
\end{equation*}
    where $U\sim {\rm Uniform}(0,1)$ is independent of any source of randomness and  the correction probabilities are $
        \gamma_{i,j}^+ = \frac{p_{i} - \pi_{i,j}}{1 - \pi_{i,j}}$, $  \gamma_{i,j}^- = \frac{\pi_{i,j} - p_{i}}{\pi_{i,j}}$. The idea is that if $\pi_{i,j} < p_{i}$ (under-generation of success), we flip the candidate value from $0$ to $1$ with probability $\gamma^+_{i,j}$ and keep it $1$ if it is already the candidate value. We symmetrically proceed in the case of success over-generation. This ensures marginal consistency, {\it i.e.} $\mathbb{E}(\Lambda_i^{(j)})=p_i$ for any $1\le i\le n$ and $1\le j \le p$. Indeed, one may condition with respect to $\pi_{i,j}$ and do a case analysis. Write $\mathbb{E}(\Lambda_i^{(j)})=\mathbb{E}(\Lambda_i^{(j)} \mid \pi_{i,j})$. On the event $A = \{ \pi_{i,j}< p_i\}$, we have by independence,
\begin{align*}
\mathbb{E}(\Lambda_i^{(j)} \mid \pi_{i,j})
&=  \mathbb{P}(\tilde{\Lambda}_i^{(j)} = 1 \mid \pi_{i,j}) + \mathbb{P}(\tilde{\Lambda}_i^{(j)} = 0,\ U_i^{(j)} < \gamma_{i,j}^+ \mid \pi_{i,j})
= \pi_{i,j}+ (1 - \pi_{i,j}) \cdot \gamma_{i,j}^+ = p_i.
\end{align*}
On the event $B = \{\pi_{i,j}> p_i\}$, similarly $\mathbb{E}(\Lambda_i^{(j)} \mid \pi_{i,j})
= \mathbb{P}(\tilde{\Lambda}_i^{(j)} = 1,\ U_i^{(j)} \geq \gamma_{i,j}^- \mid \pi_{i,j}) = p_i$. On the event $C = \{ \pi_{i,j} = p_i\}$, one has $\gamma_{i,j}^+ = \gamma_{i,j}^- = 0$ which implies $\Lambda_i^{(j)} = \tilde{\Lambda}_i^{(j)}$ and thus, 
$\mathbb{E}(\Lambda_i^{(j)} \mid \pi_{i,j}) = \mathbb{E}(\tilde{\Lambda}_i^{(j)} \mid \pi_{i,j}) = \pi_{i,j} = p_i$. The claim follows by taking the expectation, since $\{A, B, C\}$ is a partition of~$\Omega$.

\subsubsection{Tail heaviness} \label{par-tail}
The GARCH$(1,1)$-structure tends to increase the tail heaviness, even with Gaussian innovations. Thus, we consider in the second setup that $\gamma\in\{ 0.1, 0.4, 0.5\}$ with $0.5$ instead of $0.7$ in the first setup. The case $\gamma=0.5$ is somewhat the limiting possibility if one wants to reach $\tau \le -0.5$. In fact, \cite[Remark~3.4]{BASRAK200295} provides a way to compute the resulting tail-index $\gamma_g$ (with $g$ as GARCH) of the time series, by solving $\E((\alpha_{\rm resp} \tilde{\eta}_1^2+\beta_{\rm resp})^{1/(2\gamma_g)})=1$ where $\tilde{\eta}_1$ is a standardized BurrXII$(-\rho,1/\gamma)$-random variable. Some Monte-Carlo computations yield, when $\gamma=0.1$, in the standard GARCH$(1,1)$-case a new tail-index $\gamma_g \approx 0.28$ and in the IGARCH-like case, $\gamma_g \approx 0.31$. When $\gamma = 0.4$, it respectively provides $\gamma_g \approx 0.43$ and $\gamma_g \approx 0.45$. In any case, the Monte-Carlo method fails when $\gamma=0.5$.

\subsubsection{Strict-stationarity and $\alpha$-mixing} \label{par-station}
 Denote momentarily $x \in \{ {\rm resp}, {\rm noise}\}$. From \cite{Mokkadem88}, the ARMA$(1,1)$ process with i.i.d. finite mean innovations (Gaussian or heavy-tailed) and $|\phi_x|<1$ is strictly-stationary and $\alpha$-mixing. For the GARCH$(1,1)$-process to be strictly-stationary, it is enough to verify by Monte-Carlo computations that $\E(\log (\alpha_x Z +\beta_x))<0$ when $Z\sim Y$ or $Z\sim \mathcal{N}(0,1)$ (noise case). Thus, by \cite{CarrascoChen2002} and since $\alpha_{x} + \beta_{x} < 1$, it is also $\alpha$-mixing. Nonetheless, the log-moment condition is not numerically true when $\gamma=0.5$ (and even less when $\gamma=0.7$), in the case of GARCH$(1,1)$-response, indicating that the process fails to be strictly-stationary. Still, we decided to include this particularly pathological case as our method performs well. At last, the mask $(\Lambda_i)$ is a finite-state, irreducible, aperiodic Markov chains hence is $\alpha$-mixing, see \cite[Theorem~3.2]{Bradley2005}.

\subsection{Selection of the threshold}
\label{sub-seuil}

The choice of the diverging threshold is a recurrent problem in extreme-value statistics, see \cite{scarrott2012review} for a review. In the considered setup, a natural method is to select the value $y_n$ which maximizes the empirical covariance between $Y$ and $\hat \beta_\Lambda(y_n)^\top X$ using a grid search on $(0,\infty)$. This optimization problem can be
simplified by remarking that the considered empirical covariance is a piecewise 
constant function of $y_n$ with $n$ jumps at $Y_1,\dots,Y_n$.
It is thus sufficient to consider a discrete grid; to this end let us introduce $\bar{r}(k)$ defined as
\begin{align*}
\frac{1}{k} \sum_{i=1}^k Y_{n-i+1,n} \, \hat \beta_{\Lambda}(Y_{n-k+1,n})^\top X_{(n-i+1,n)} - \frac{1}{k} \sum_{i=1}^k Y_{n-i+1,n}\, \frac{1}{k} \sum_{i=1}^k  \hat \beta_{\Lambda}(Y_{n-k+1,n})^\top X_{(n-i+1,n)} ,
\end{align*}
and where $X_{(n-i+1,n)}$ denotes the concomitant of $Y_{n-i+1,n}$, \emph{i.e.}, the random variable $X_s$ with $s\in \{1,\ldots,n\}$ being the unique index such that $Y_s=Y_{n-i+1,n}$.
In practice, we let $\hat k := {\rm argmax}_{5\le k\le n/5} \bar{r}(k)$
in order to prevent instabilities of
estimates built on too few data points. 

\subsection{Results}
\label{sub-results}

The performance of our method is assessed through the comparison of the two graphs $(j/p, \hat \beta^{(j)}_\Lambda(Y_{n-\hat k +1,n}))$ and $(j/p, \beta_j)$ indexed by $1\le j \le p$. Figures~\ref{fig:iid_iid}--\ref{fig:estar_patho_GARCH} display the results
on each of the 11 considered dependence frameworks. Each of them includes six or nine panels
associated with different set of parameters $(\gamma,\tau,\kappa)$.

In all $9\times9+2\times 6=93$ situations, the orange curve represents the graph associated with $\beta$, while the blue one is associated with the averaged value of $\hat \beta(Y_{n-\hat k+1,n})$ over the $N=500$ Monte-Carlo replications, where $\hat k$ is selected following the above described procedure. The light blue area corresponds to the confidence region comprising the top $5-95\%$ values of the Monte Carlo replications.

Figure~\ref{fig:iid_iid} displays the no-serial dependence case. Adversely, Figure~\ref{fig:st_ARMA_st_GARCH} corresponds to a standard ARMA$(1,1)$-response $(\phi_{\rm resp},\theta_{\rm resp})=(0.8,-0.3)$ together with a standard GARCH$(1,1)$-noise $(\omega_{\rm noise},\alpha_{\rm noise},\beta_{\rm noise})=(0.05,0.1,0.85)$. The results for the IGARCH-like noise are presented in Figure~\ref{fig:st_ARMA_patho_GARCH}. Similarly, Figure~\ref{fig:patho_ARMA_st_GARCH} is for the pathological ARMA$(1,1)$-response with $(\phi_{\rm resp},\theta_{\rm resp})= (0.99,-0.98)$ with standard GARCH$(1,1)$-noise, and Figure~\ref{fig:patho_ARMA_patho_GARCH} when instead the GARCH$(1,1)$-noise is close to IGARCH, \emph{i.e.,} $(\omega_{\rm noise},\alpha_{\rm noise},\beta_{\rm noise})=(0.05,0.05,0.94)$.
Next, the second experimental setup is gathered in Figures~\ref{fig:st_GARCH_st_ARMA}--\ref{fig:patho_GARCH_patho_ARMA} with mirrored sets of parameters.
This is synthesized in Table~\ref{tab-fig}.
\begin{table}[h]
    \centering
\begin{tabular}{|c|ccccc|}
 Response / Noise  & Indep. & Stand. ARMA & Stand. GARCH&  Patho. ARMA  & Patho. GARCH \\
 \hline
 Indep. &Figure~\ref{fig:iid_iid}  & & & & \\
  Stand. ARMA   & & &  Figure \ref{fig:st_ARMA_st_GARCH} && Figure \ref{fig:st_ARMA_patho_GARCH} \\ 
  Stand. GARCH  & &  Figure \ref{fig:st_GARCH_st_ARMA}  &  & Figure \ref{fig:st_GARCH_patho_ARMA}&  \\ 
   Patho. ARMA & &  &Figure \ref{fig:patho_ARMA_st_GARCH}&  &  Figure \ref{fig:patho_ARMA_patho_GARCH}  \\
  Patho. GARCH & & Figure \ref{fig:patho_GARCH_st_ARMA} &  & Figure \ref{fig:patho_GARCH_patho_ARMA} &  \\
  ESTAR & & & Figure \ref{fig:estar_st_GARCH}& & Figure \ref{fig:estar_patho_GARCH}\\
\end{tabular}
\caption{Organization of the results on simulated data.}
\label{tab-fig}
    \end{table}
Finally, Figures~\ref{fig:estar_st_GARCH} and Figrure~\ref{fig:estar_patho_GARCH} display the case of an ESTAR$(1)$-response with $\phi_{\rm low}=0.2$ and $\phi_{\rm high}=0.95$ with both GARCH$(1,1)$-noise models. The first two rows are devoted to a standard noise with $(\omega_{\rm resp},\alpha_{\rm resp},\beta_{\rm resp})=(0.05,0.1,0.85)$ for $\gamma \in \{0.1,0.5\}$ while the last two rows concern the pathological noise with $(\omega_{\rm resp},\alpha_{\rm resp},\beta_{\rm resp})=(0.05,0.05,0.94)$. 

In all considered situations, the EPLS method provides reliable estimates in terms of bias and variance. It appears to be robust with respect to heavy tails, serial dependence, and missing data, by showing good performance in pathological and nonlinear settings.
Unsurprisingly, the estimation procedure becomes more instable as $\tau<0$ since it implies lighter tails for $\lambda$ and thus more masked observations which reduces the available data sample.
On the other hand, a larger tail-index $\gamma$ yields heavier tails for the response and produces extremely large values more frequently, leading to higher variance in tail estimates. 

\begin{figure}[p]
    \centering
    
    \begin{subfigure}{0.325\textwidth}
        \includegraphics[scale=0.35]{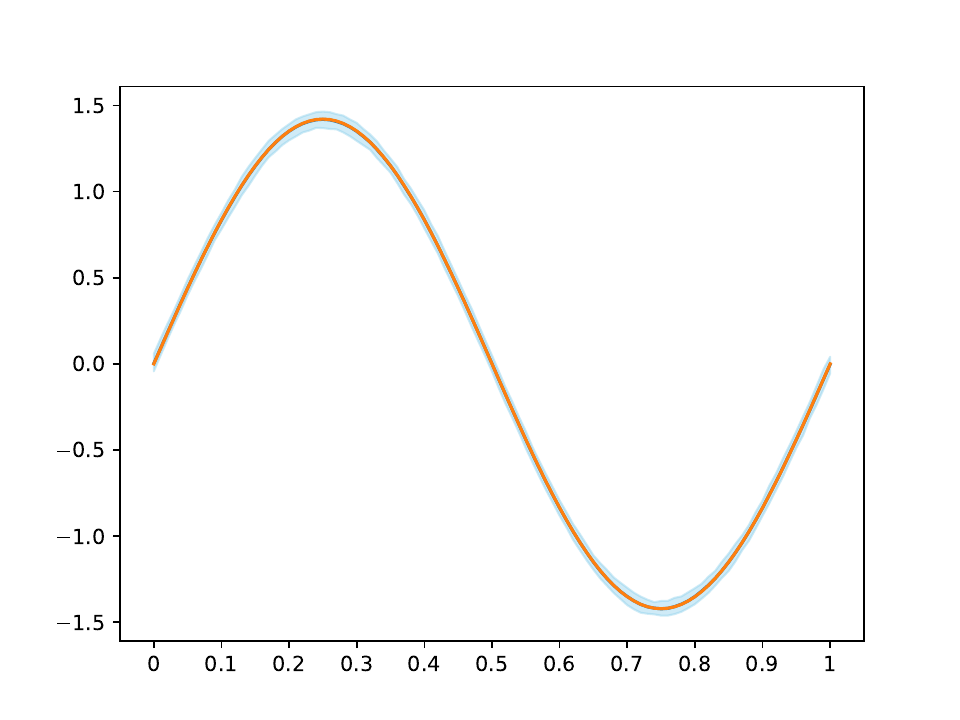}
        \caption{$\gamma=0.1$, $\tau = -0.1$, $\kappa = 0.5$.}
    \end{subfigure}
    \hfill
    \begin{subfigure}{0.325\textwidth}
        \includegraphics[scale=0.35]{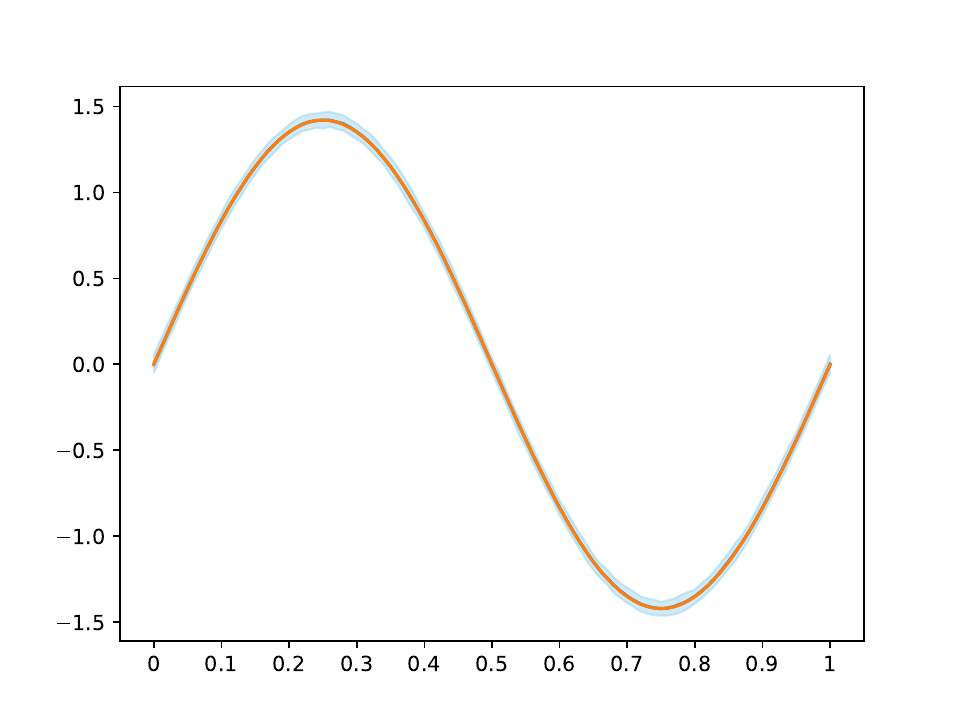}
        \caption{$\gamma=0.1$, $\tau = -0.5$, $\kappa = 0.5$.}
    \end{subfigure}
    \hfill
    \begin{subfigure}{0.325\textwidth}
        \includegraphics[scale=0.35]{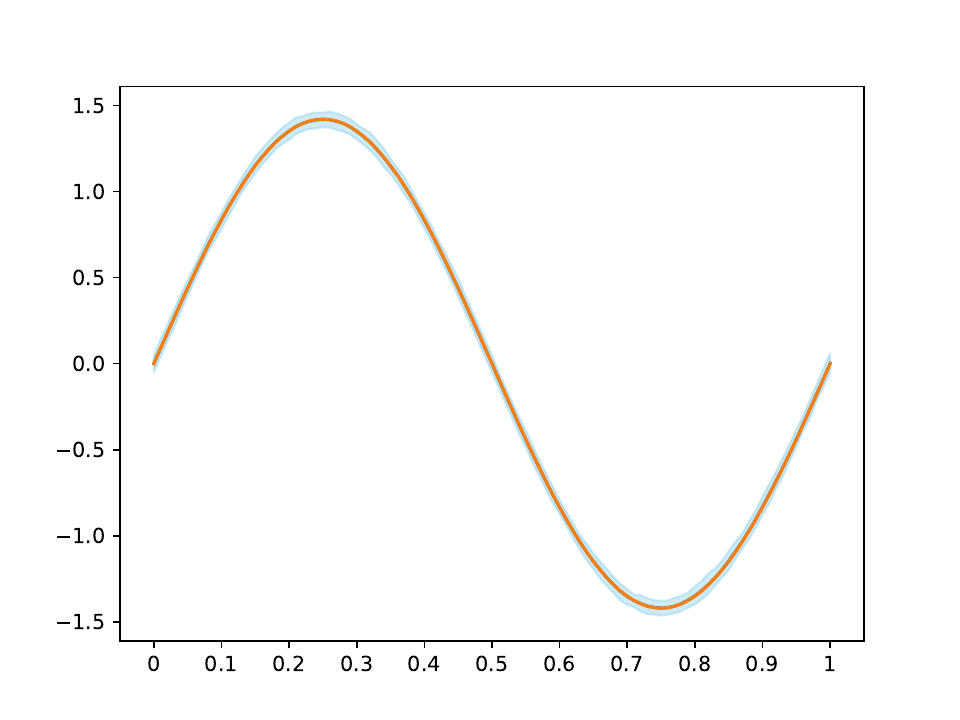}
        \caption{$\gamma=0.1$, $\tau = -1$, $\kappa = 0.5$.}
    \end{subfigure}

    \medskip 
    \begin{subfigure}{0.325\textwidth}
        \includegraphics[scale=0.35]{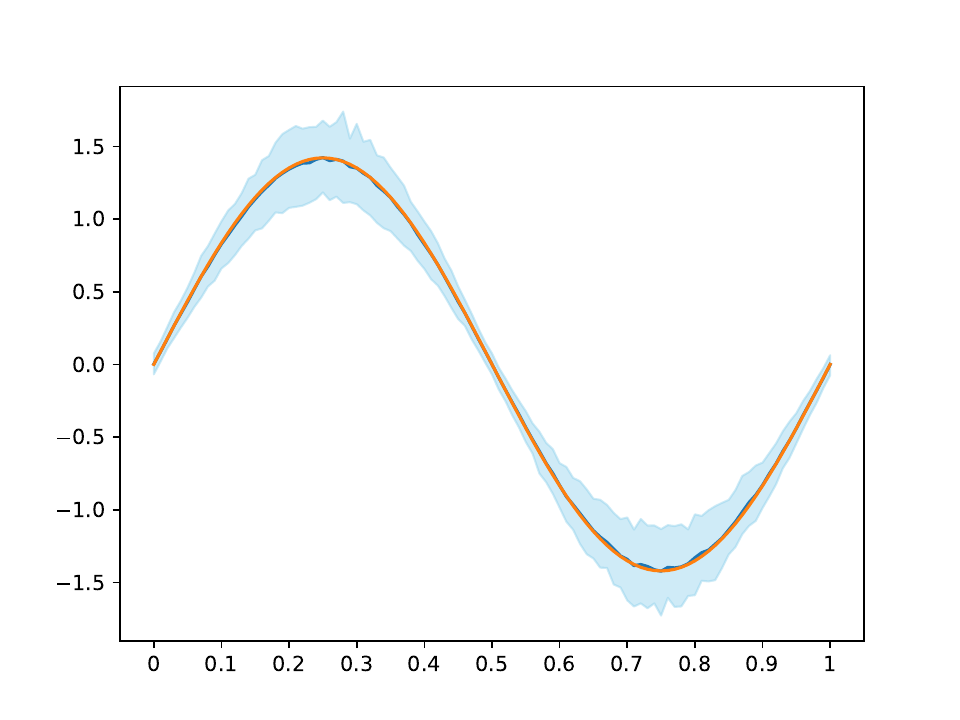}
        \caption{$\gamma=0.5$, $\tau = -0.1$, $\kappa = 0.5$.}
    \end{subfigure}
    \hfill
    \begin{subfigure}{0.325\textwidth}
        \includegraphics[scale=0.35]{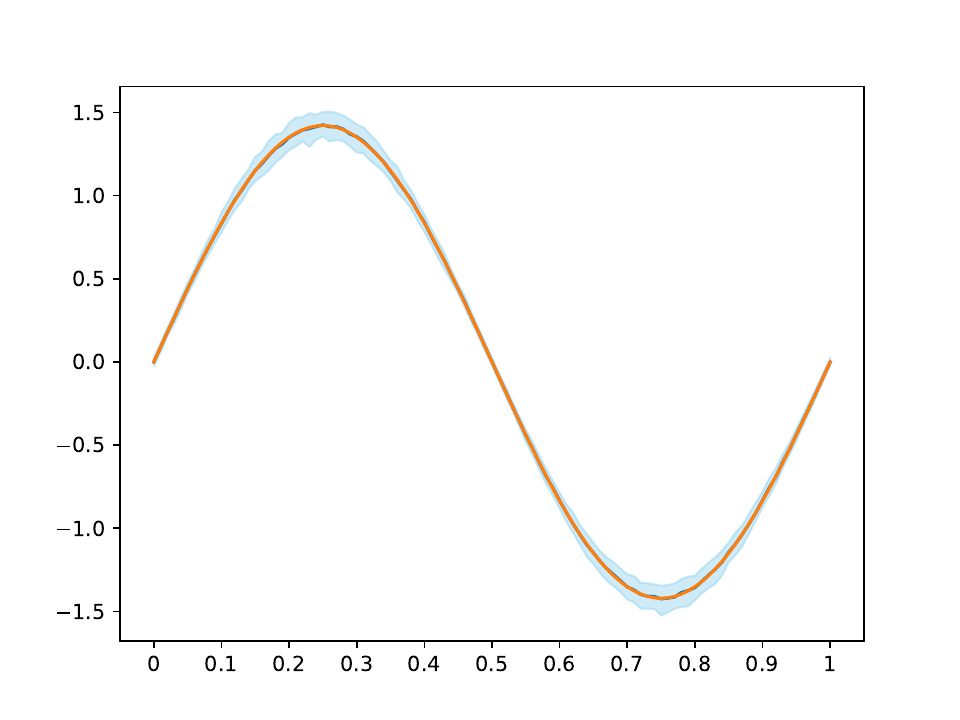}
        \caption{$\gamma=0.5$, $\tau = -0.5$, $\kappa = 0.5$.}
    \end{subfigure}
    \hfill
    \begin{subfigure}{0.325\textwidth}
        \includegraphics[scale=0.35]{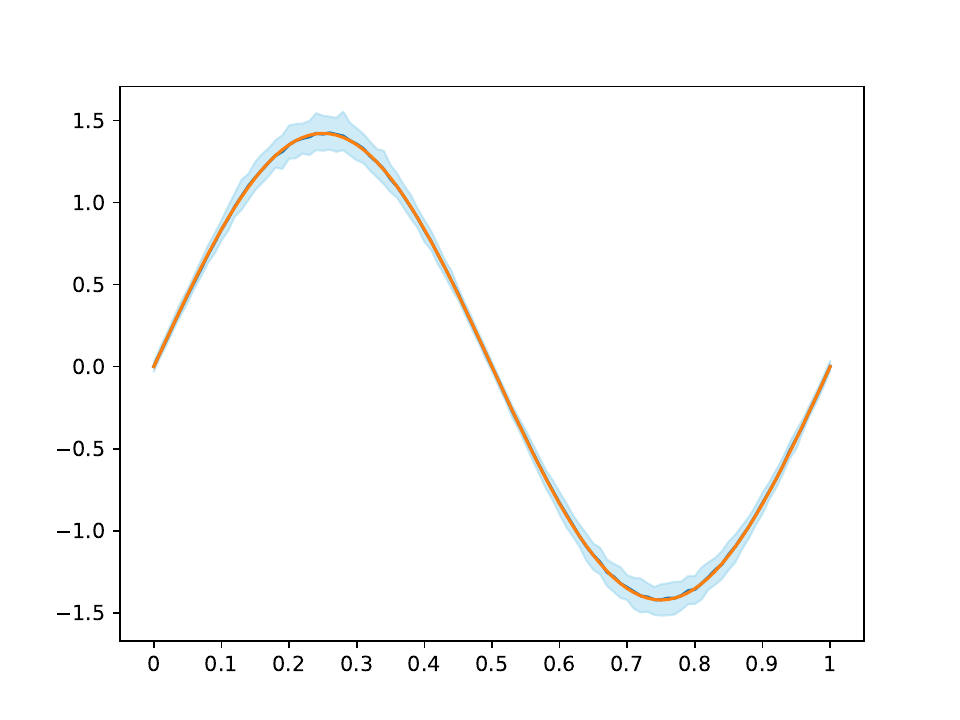}
        \caption{$\gamma=0.5$, $\tau = -1$, $\kappa = 0.5$.}
    \end{subfigure}

    \medskip 

    \begin{subfigure}{0.325\textwidth}
        \includegraphics[scale=0.35]{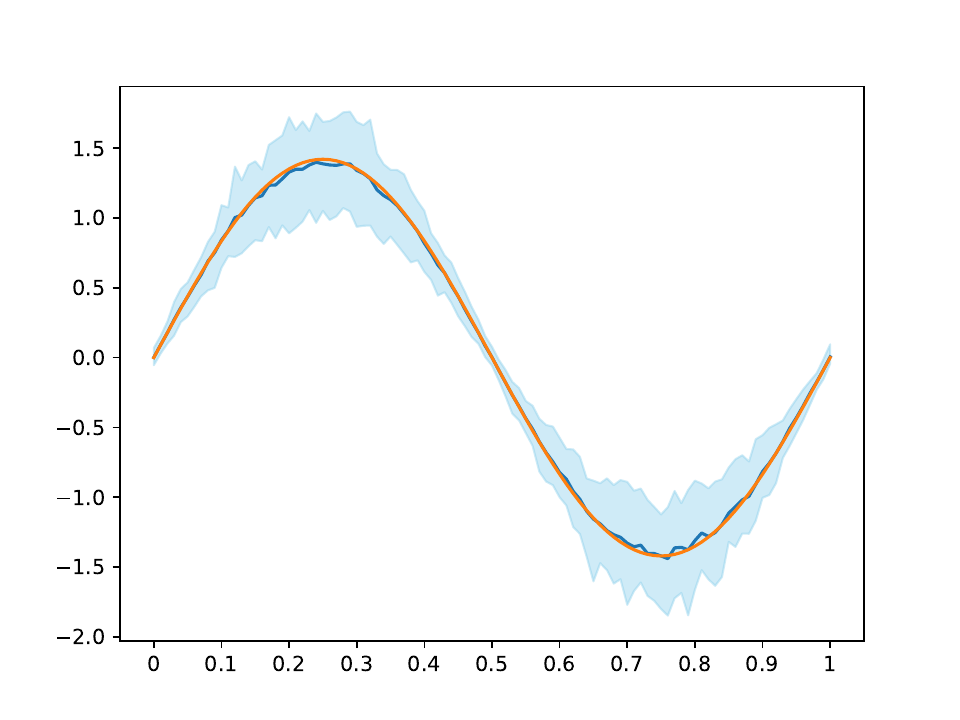}
        \caption{$\gamma=0.9$, $\tau = -0.1$, $\kappa = 0.5$.}
    \end{subfigure}
    \hfill
    \begin{subfigure}{0.325\textwidth}
        \includegraphics[scale=0.35]{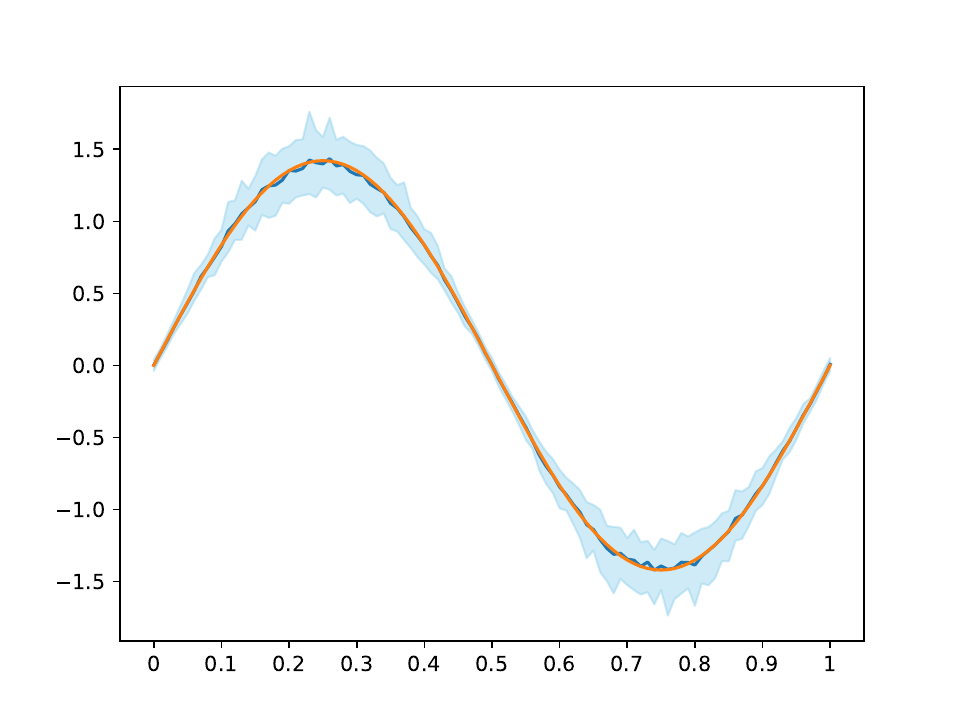}
        \caption{$\gamma=0.9$, $\tau = -0.5$, $\kappa = 0.5$.}
    \end{subfigure}
    \hfill
    \begin{subfigure}{0.325\textwidth}
        \includegraphics[scale=0.35]{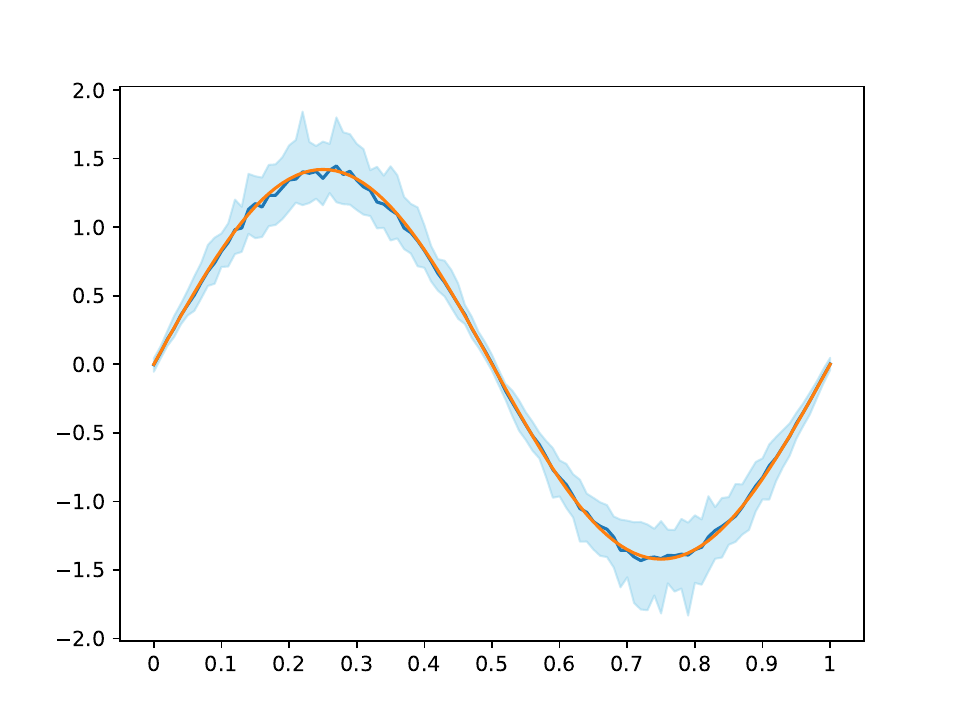}
        \caption{$\gamma=0.9$, $\tau = -1$, $\kappa = 0.5$.}
    \end{subfigure}
      
\caption{Simulation results on the inverse model without serial dependence (i.i.d. case), \emph{i.e.,} $(\phi_{\rm resp},\theta_{\rm resp})=(0,0)$ and $(\omega_{\rm noise},\alpha_{\rm noise},\beta_{\rm noise})=(1,0,0)$. }
    \label{fig:iid_iid}
\end{figure}

\begin{figure}[p]
    \centering
    
    \begin{subfigure}{0.325\textwidth}
        \includegraphics[scale=0.35]{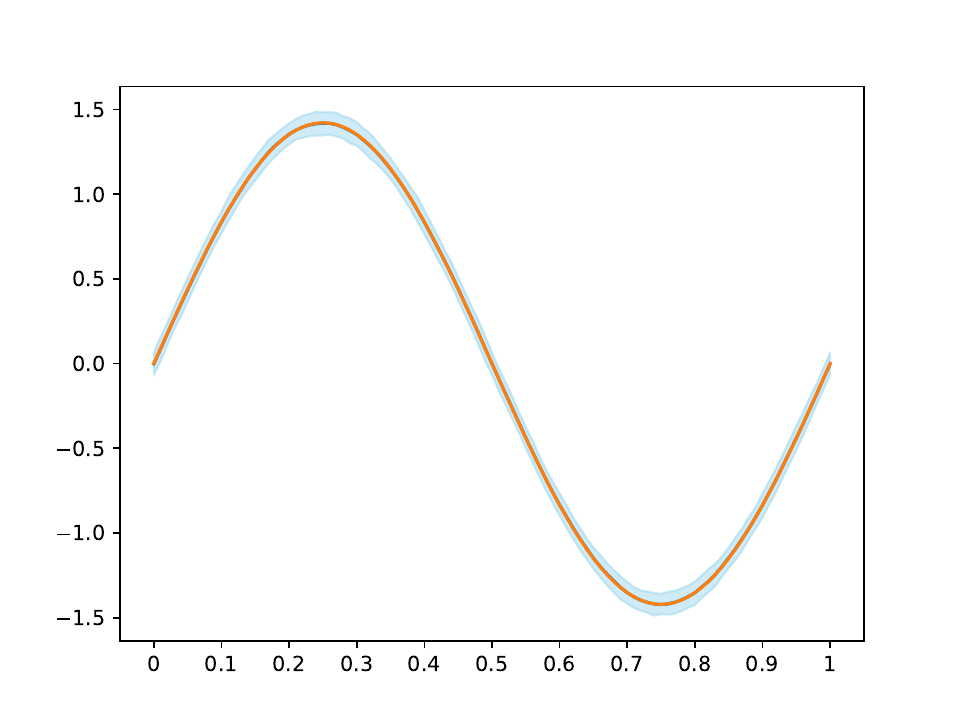}
        \caption{$\gamma=0.1$, $\tau = -0.1$, $\kappa=0.5$.}
    \end{subfigure}
    \hfill
    \begin{subfigure}{0.325\textwidth}
        \includegraphics[scale=0.35]{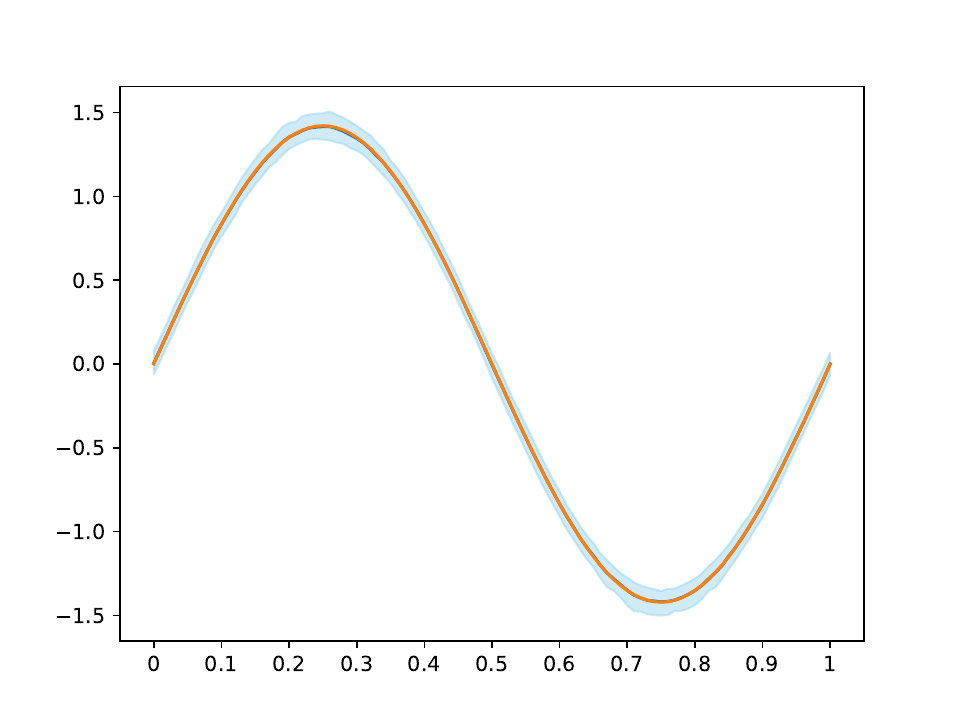}
        \caption{$\gamma=0.1$, $\tau = -0.5$, $\kappa=0.5$.}
    \end{subfigure}
    \hfill
    \begin{subfigure}{0.325\textwidth}
        \includegraphics[scale=0.35]{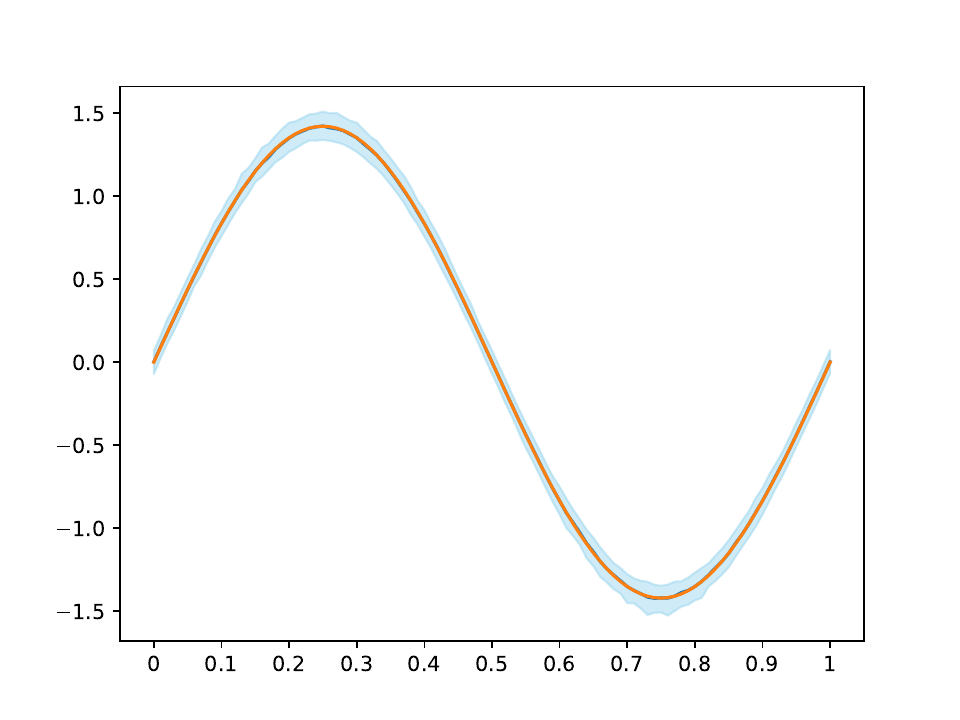}
        \caption{$\gamma=0.1$, $\tau = -0.9$, $\kappa=0.5$.}
    \end{subfigure}

    \medskip 
    \begin{subfigure}{0.325\textwidth}
        \includegraphics[scale=0.35]{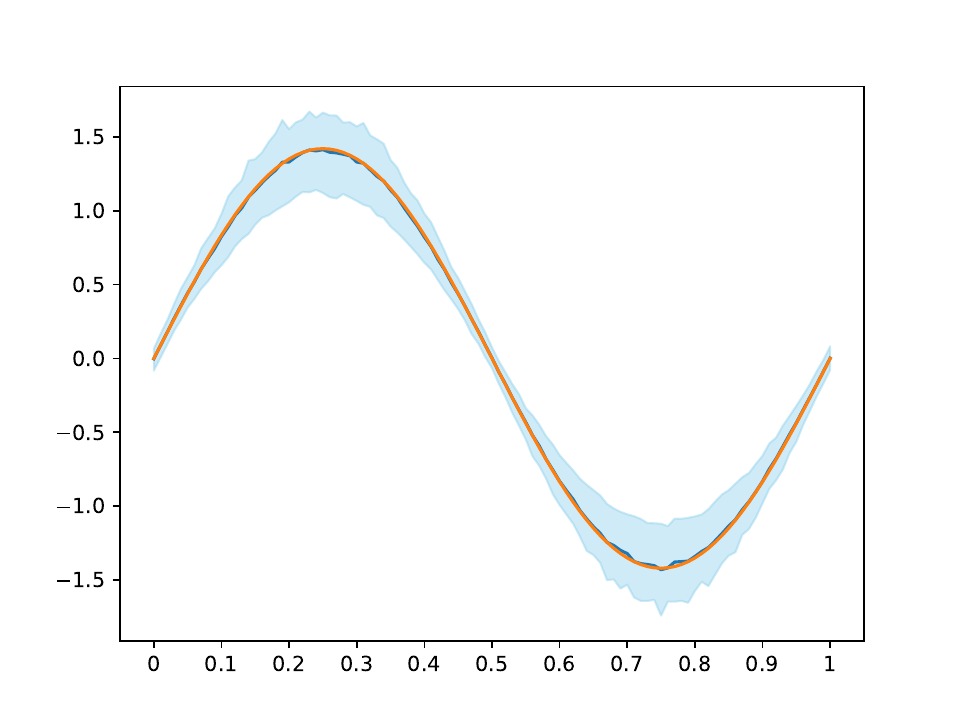}
        \caption{$\gamma=0.4$, $\tau = -0.1$, $\kappa=0.5$.}
    \end{subfigure}
    \hfill
    \begin{subfigure}{0.325\textwidth}
        \includegraphics[scale=0.35]{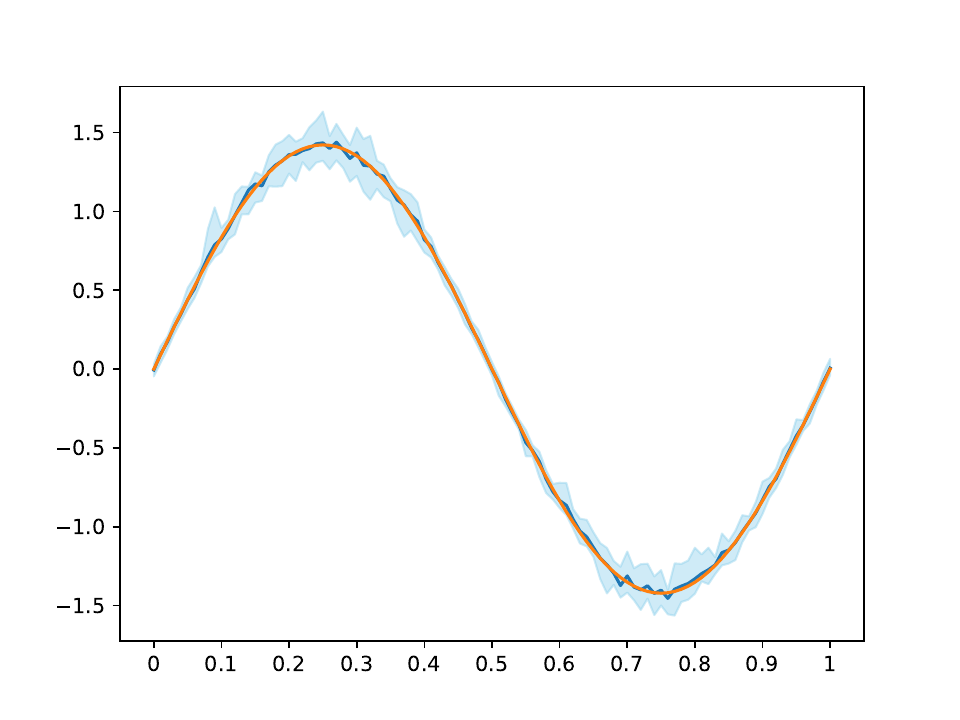}
        \caption{$\gamma=0.4$, $\tau = -0.5$, $\kappa=0.5$.}
    \end{subfigure}
    \hfill
    \begin{subfigure}{0.325\textwidth}
        \includegraphics[scale=0.35]{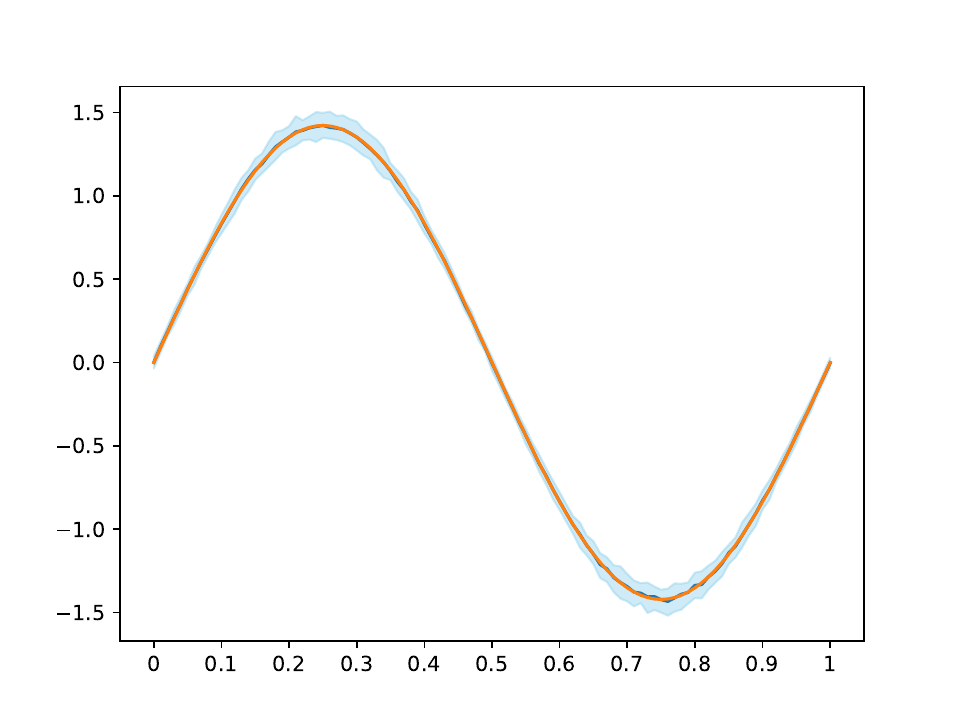}
        \caption{$\gamma=0.4$, $\tau = -0.9$, $\kappa=0.5$.}
    \end{subfigure}

    \medskip 

    \begin{subfigure}{0.325\textwidth}
        \includegraphics[scale=0.35]{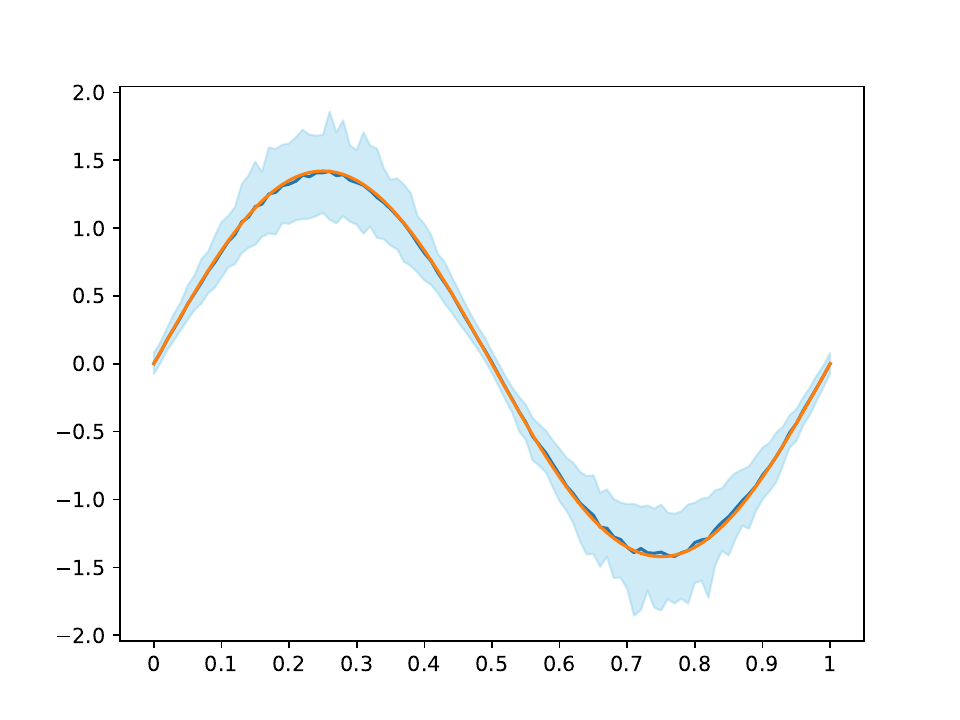}
        \caption{$\gamma=0.7$, $\tau = -0.1$, $\kappa=0.5$.}
    \end{subfigure}
    \hfill
    \begin{subfigure}{0.325\textwidth}
        \includegraphics[scale=0.35]{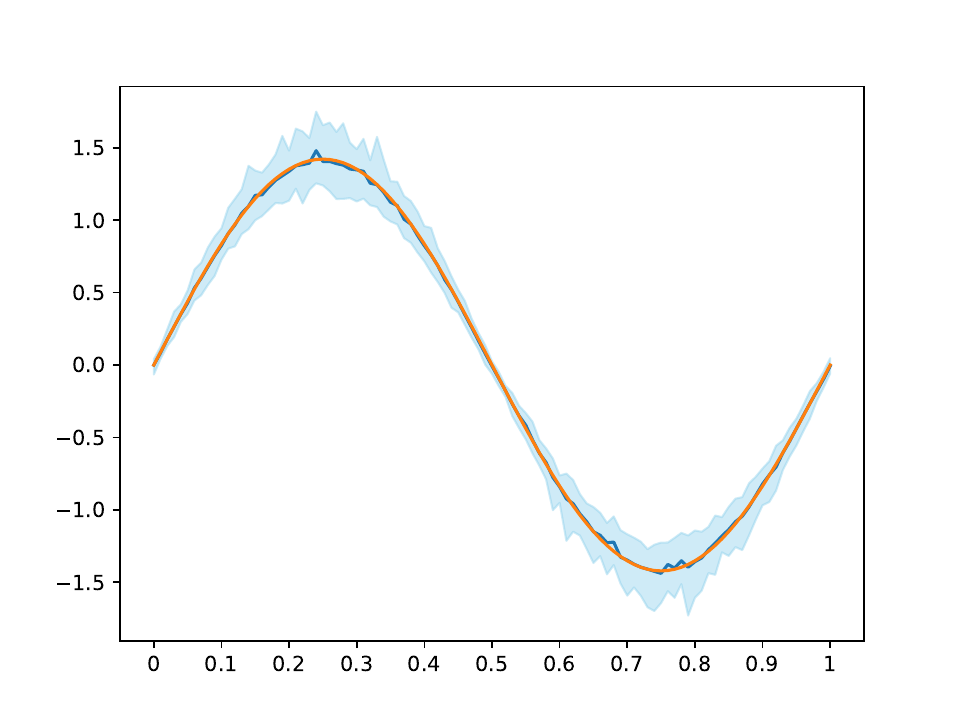}
        \caption{$\gamma=0.7$, $\tau = -0.5$, $\kappa=0.5$.}
    \end{subfigure}
    \hfill
    \begin{subfigure}{0.325\textwidth}
        \includegraphics[scale=0.35]{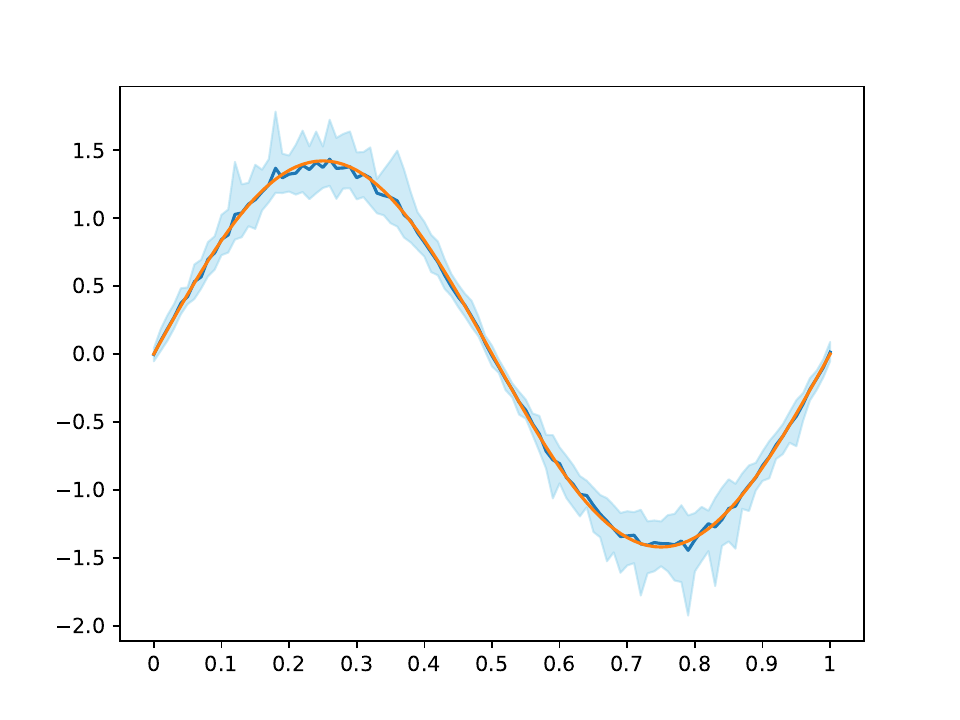}
        \caption{$\gamma=0.7$, $\tau = -0.9$, $\kappa=0.5$.}
    \end{subfigure}

      
\caption{Simulation results with serial dependence 'classic ARMA-resp. + classic GARCH-noise', \emph{i.e.,} $(\phi_{\rm resp},\theta_{\rm resp})=(0.8,-0.3)$ and $(\omega_{\rm noise},\alpha_{\rm noise},\beta_{\rm noise})=(0.05,0.1,0.85)$. }
    \label{fig:st_ARMA_st_GARCH}
\end{figure}

\begin{figure}[p]
    \centering
    
    \begin{subfigure}{0.325\textwidth}
        \includegraphics[scale=0.35]{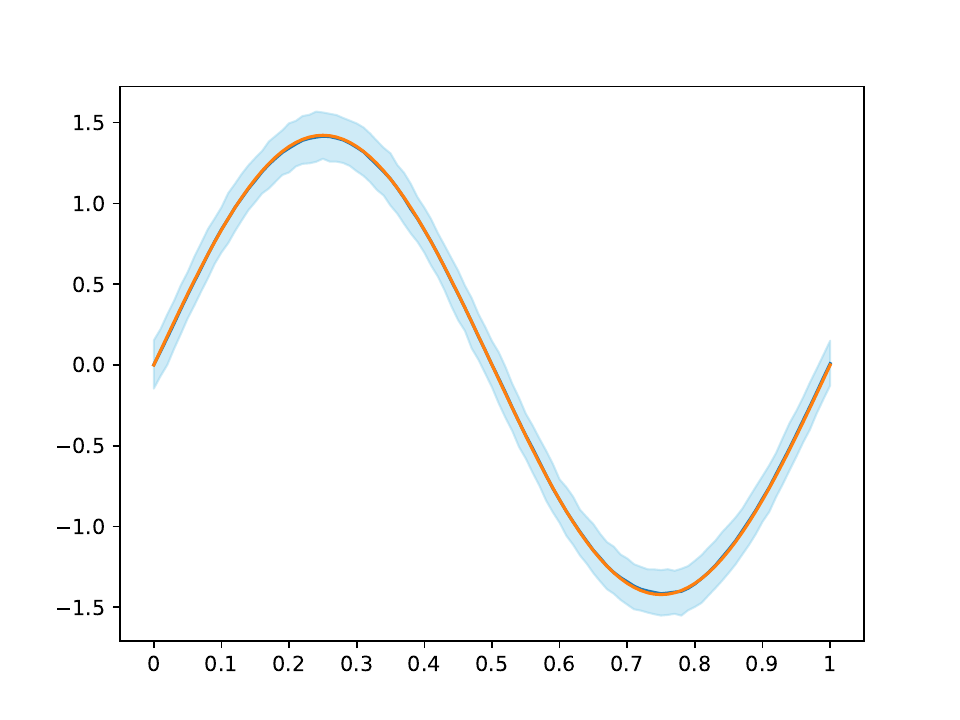}
        \caption{$\gamma=0.1$, $\tau = -0.1$, $\kappa=0.5$.}
    \end{subfigure}
    \hfill
    \begin{subfigure}{0.325\textwidth}
        \includegraphics[scale=0.35]{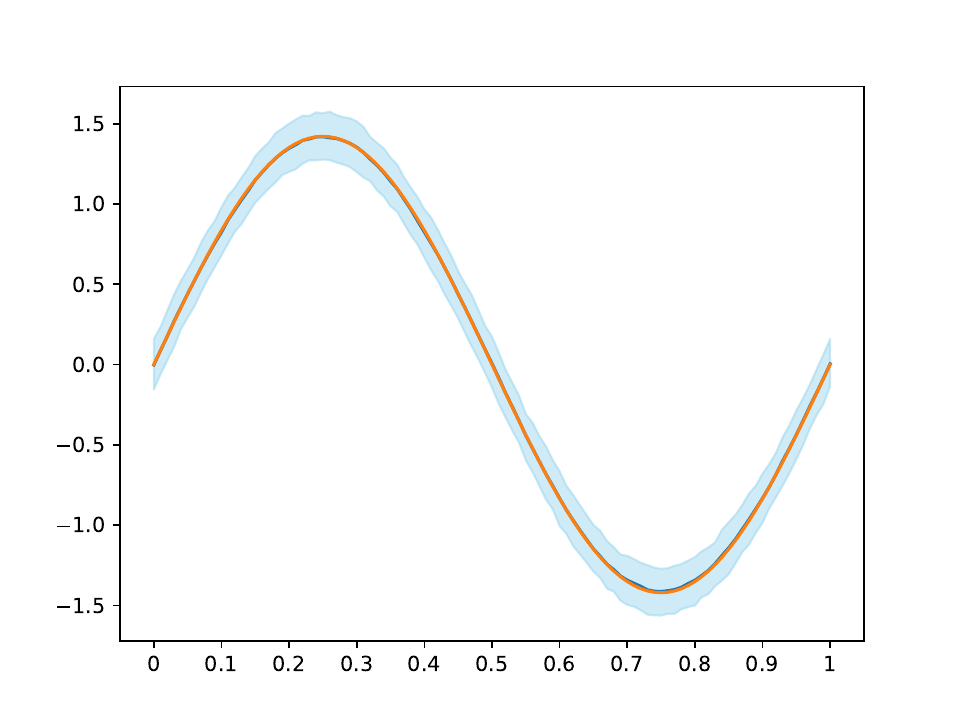}
        \caption{$\gamma=0.1$, $\tau = -0.5$, $\kappa=0.5$.}
    \end{subfigure}
    \hfill
    \begin{subfigure}{0.325\textwidth}
        \includegraphics[scale=0.35]{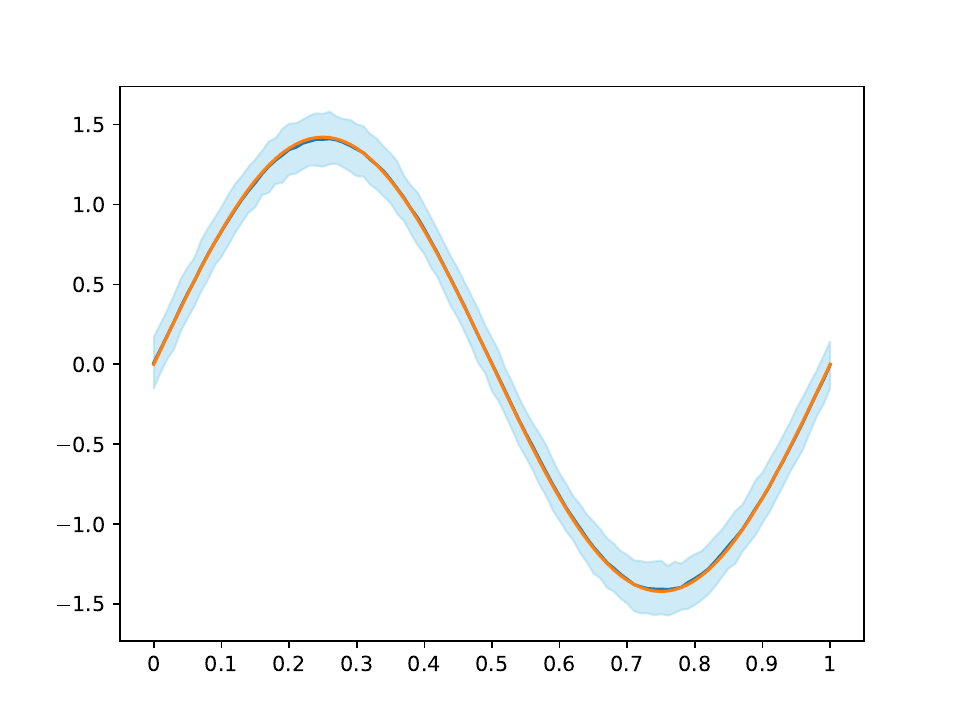}
        \caption{$\gamma=0.1$, $\tau = -0.9$, $\kappa=0.5$.}
    \end{subfigure}

    \medskip 
    \begin{subfigure}{0.325\textwidth}
        \includegraphics[scale=0.35]{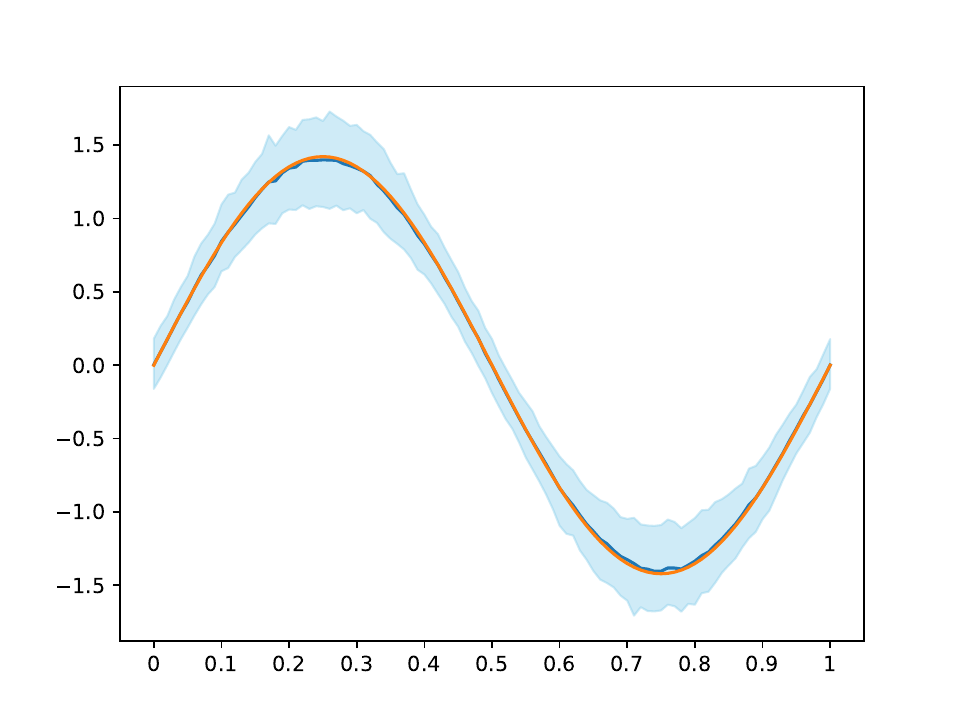}
        \caption{$\gamma=0.4$, $\tau = -0.1$, $\kappa=0.5$.}
    \end{subfigure}
    \hfill
    \begin{subfigure}{0.325\textwidth}
        \includegraphics[scale=0.35]{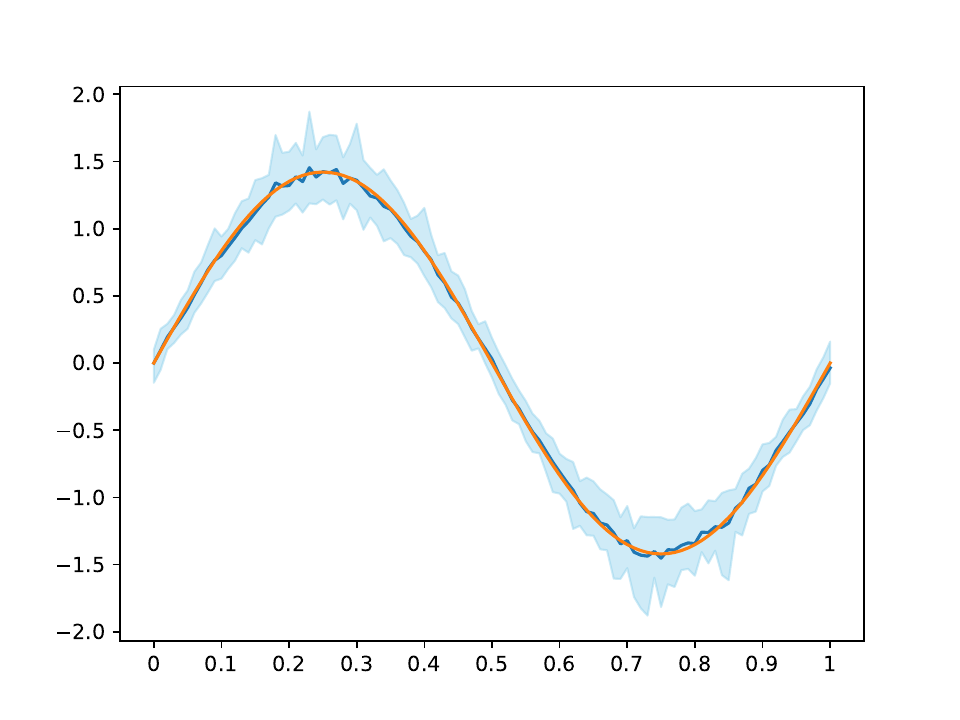}
        \caption{$\gamma=0.4$, $\tau = -0.5$, $\kappa=0.5$.}
    \end{subfigure}
    \hfill
    \begin{subfigure}{0.325\textwidth}
        \includegraphics[scale=0.35]{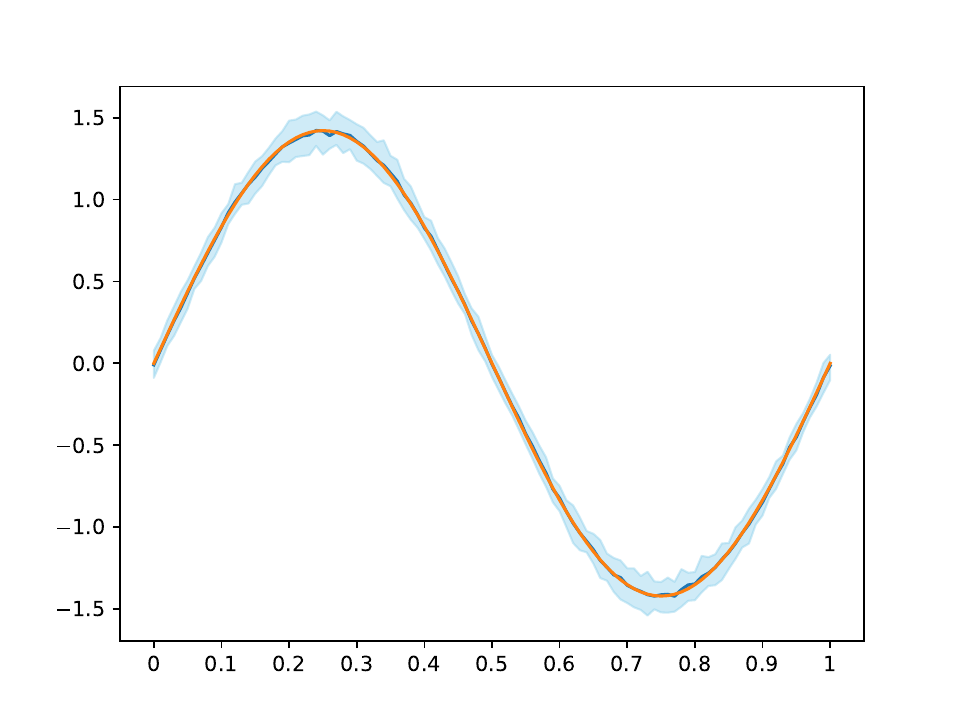}
        \caption{$\gamma=0.4$, $\tau = -0.9$, $\kappa=0.5$.}
    \end{subfigure}

    \medskip 

    \begin{subfigure}{0.325\textwidth}
        \includegraphics[scale=0.35]{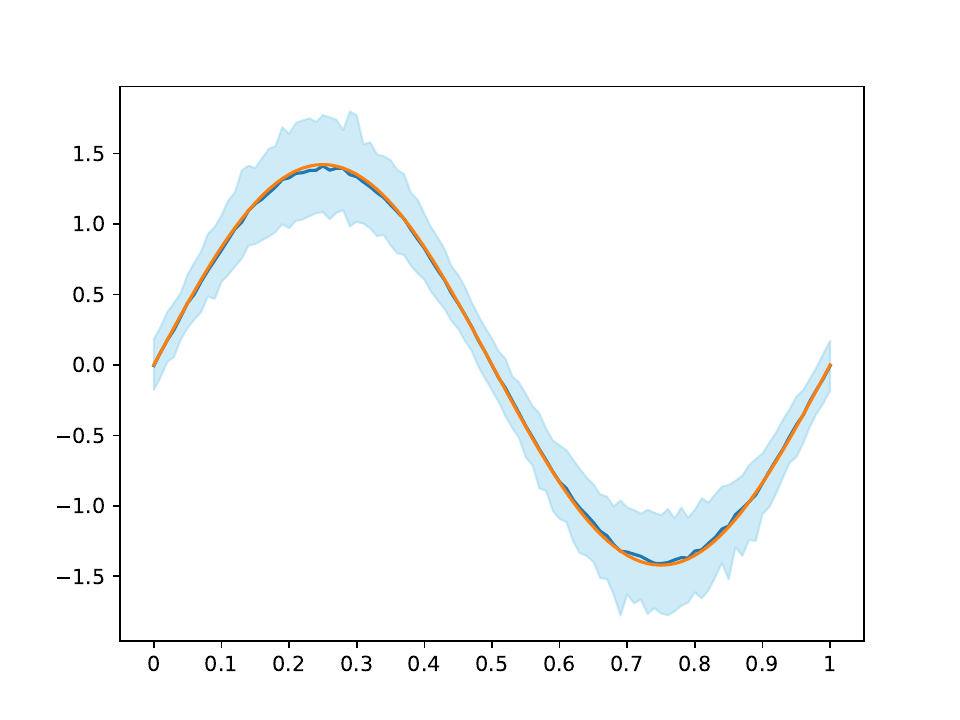}
        \caption{$\gamma=0.7$, $\tau = -0.1$, $\kappa=0.5$.}
    \end{subfigure}
    \hfill
    \begin{subfigure}{0.325\textwidth}
        \includegraphics[scale=0.35]{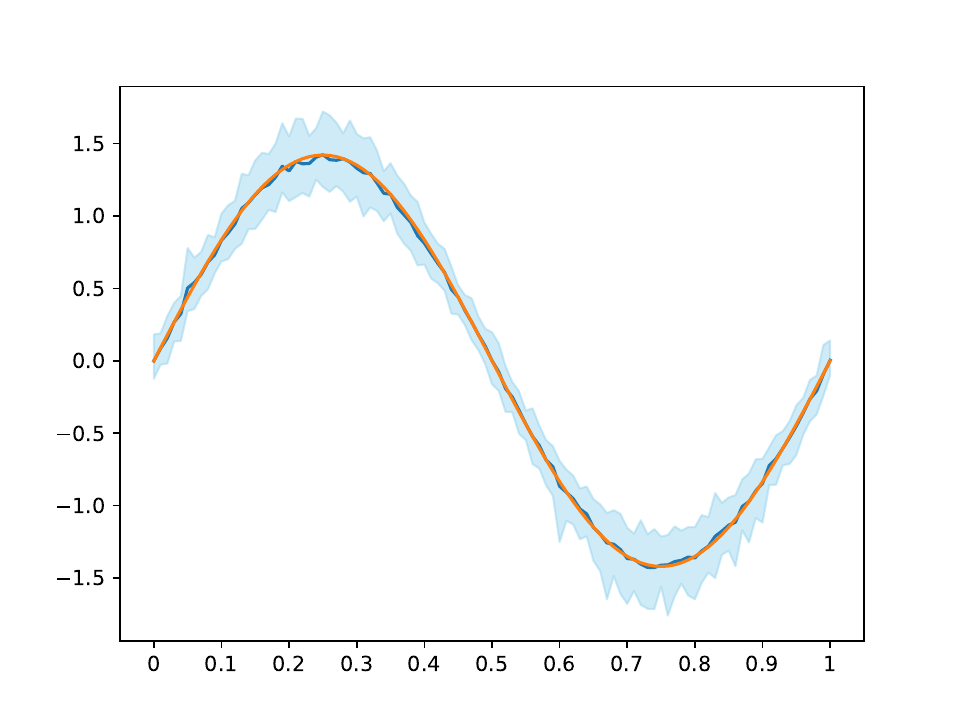}
        \caption{$\gamma=0.7$, $\tau = -0.5$, $\kappa=0.5$.}
    \end{subfigure}
    \hfill
    \begin{subfigure}{0.325\textwidth}
        \includegraphics[scale=0.35]{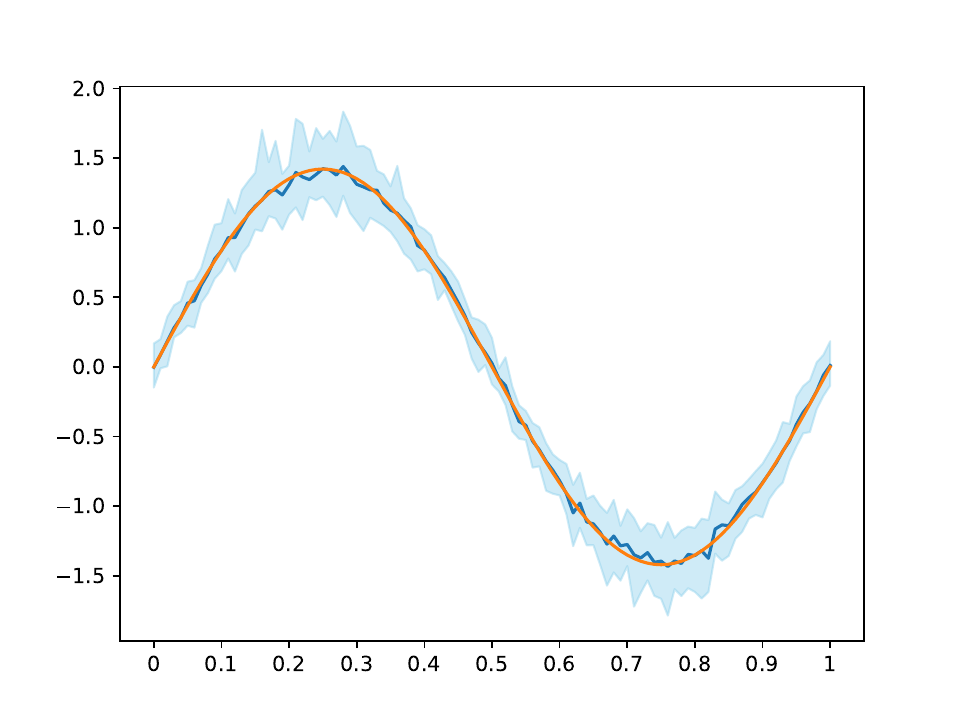}
        \caption{$\gamma=0.7$, $\tau = -0.9$, $\kappa=0.5$.}
    \end{subfigure}
      
\caption{Simulation results on the inverse model with serial dependence of type 'classic ARMA-resp. + IGARCH-like noise', \emph{i.e.,} $(\phi_{\rm resp},\theta_{\rm resp})=(0.8,-0.3)$ and $(\omega_{\rm noise},\alpha_{\rm noise},\beta_{\rm noise})=(0.05,0.05,0.94)$.}
    \label{fig:st_ARMA_patho_GARCH}
\end{figure}

\begin{figure}[p]
    \centering
    
    \begin{subfigure}{0.325\textwidth}
        \includegraphics[scale=0.35]{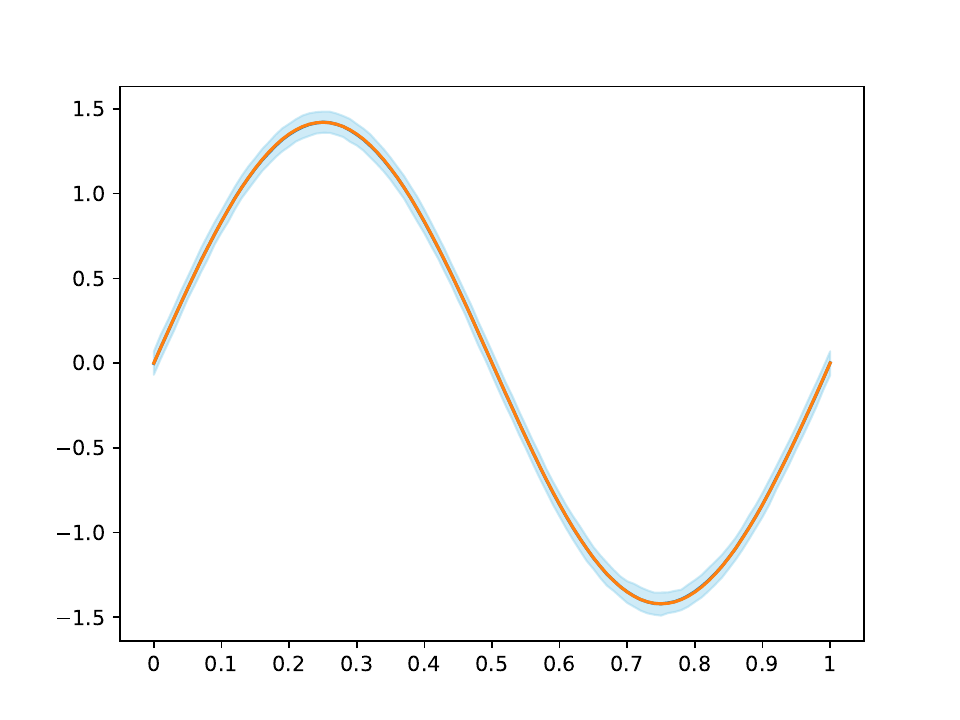}
        \caption{$\gamma=0.1$, $\tau = -0.1$, $\kappa=0.5$.}
    \end{subfigure}
    \hfill
    \begin{subfigure}{0.325\textwidth}
        \includegraphics[scale=0.35]{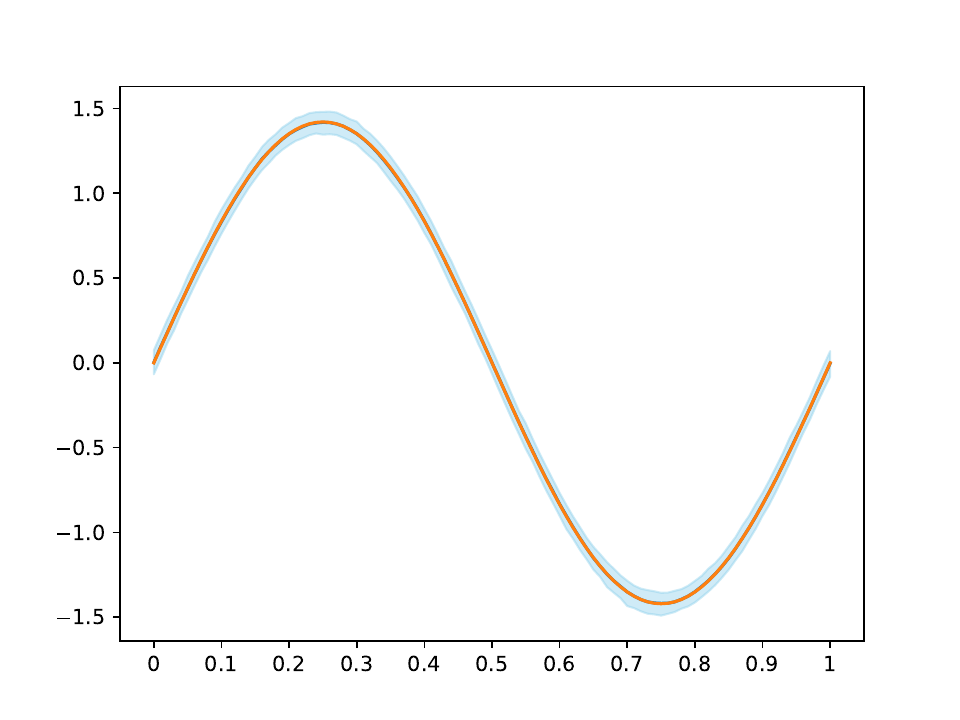}
        \caption{$\gamma=0.1$, $\tau = -0.5$, $\kappa=0.5$.}
    \end{subfigure}
    \hfill
    \begin{subfigure}{0.325\textwidth}
        \includegraphics[scale=0.35]{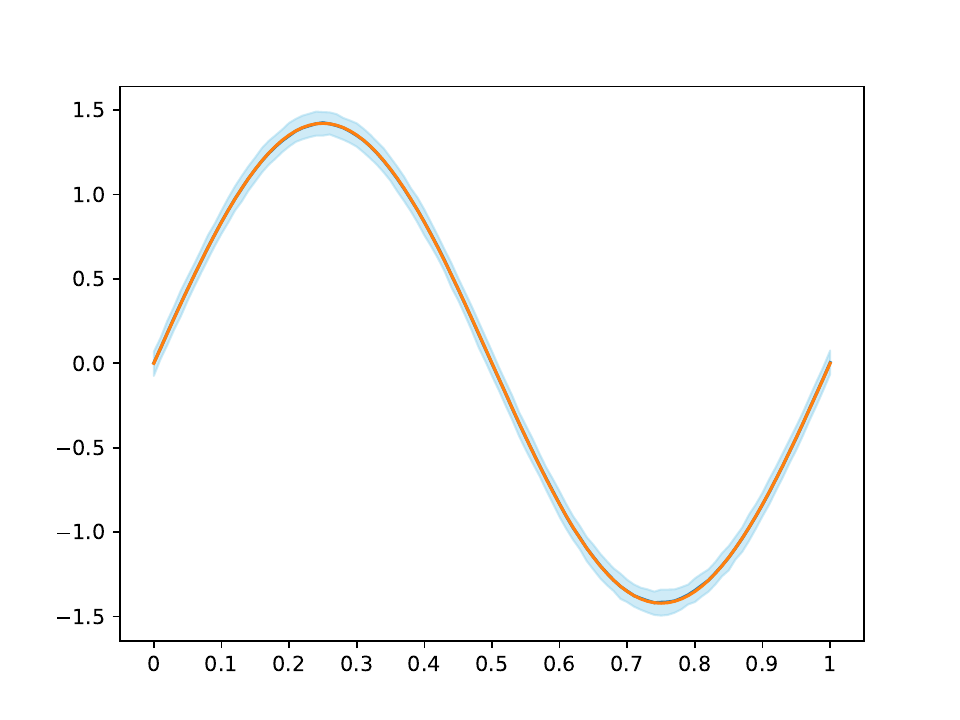}
        \caption{$\gamma=0.1$, $\tau = -0.9$, $\kappa=0.5$.}
    \end{subfigure}

    \medskip 
    \begin{subfigure}{0.325\textwidth}
        \includegraphics[scale=0.35]{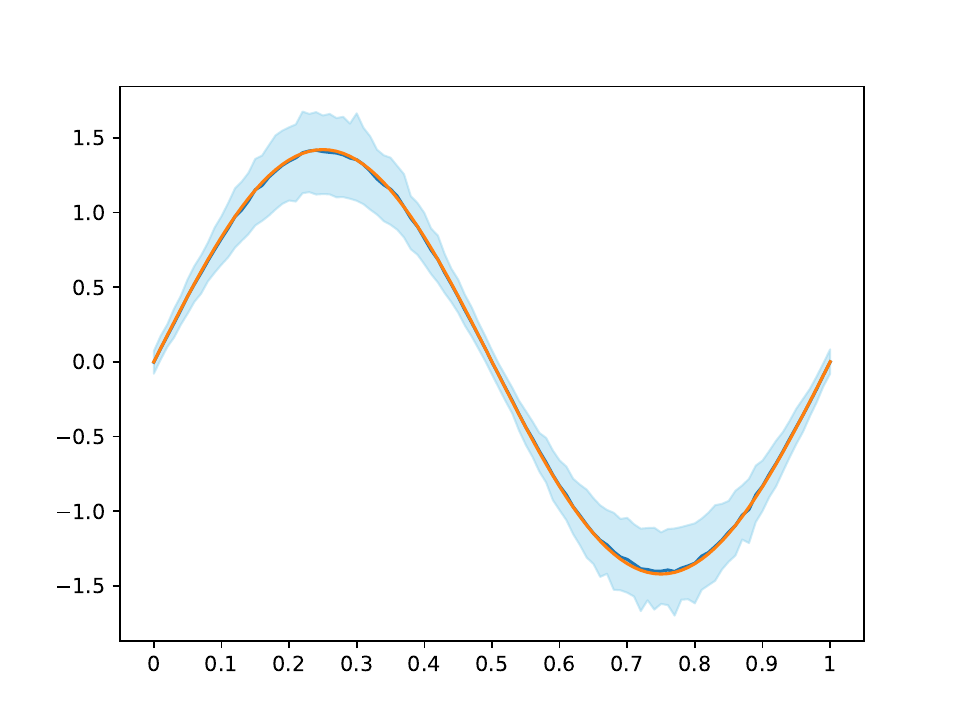}
        \caption{$\gamma=0.4$, $\tau = -0.1$, $\kappa=0.5$.}
    \end{subfigure}
    \hfill
    \begin{subfigure}{0.325\textwidth}
        \includegraphics[scale=0.35]{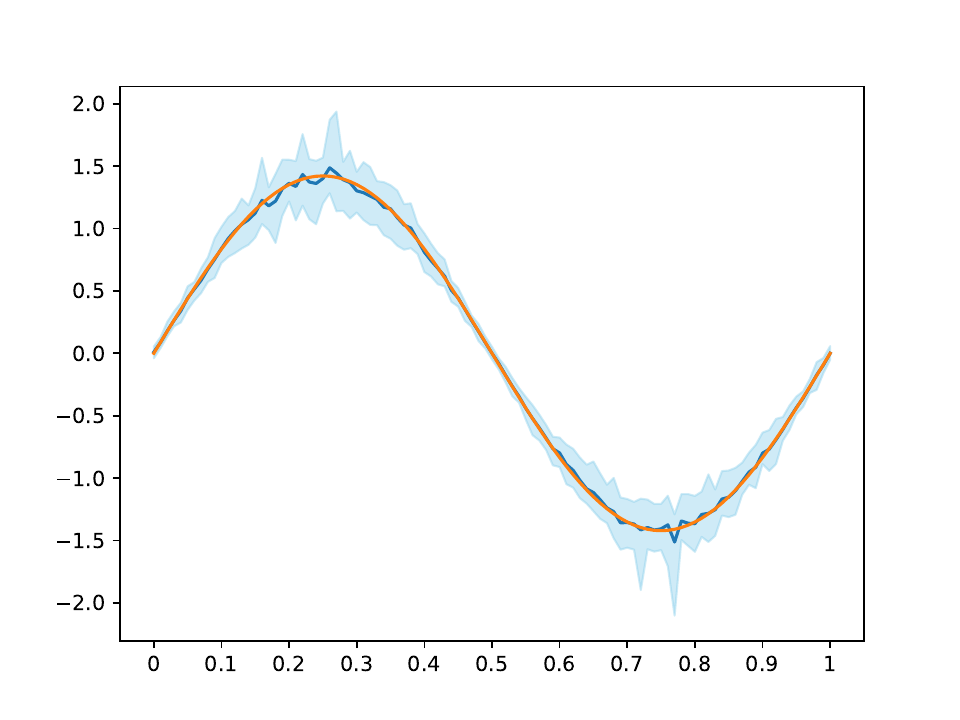}
        \caption{$\gamma=0.4$, $\tau = -0.5$, $\kappa=0.5$.}
    \end{subfigure}
    \hfill
    \begin{subfigure}{0.325\textwidth}
        \includegraphics[scale=0.35]{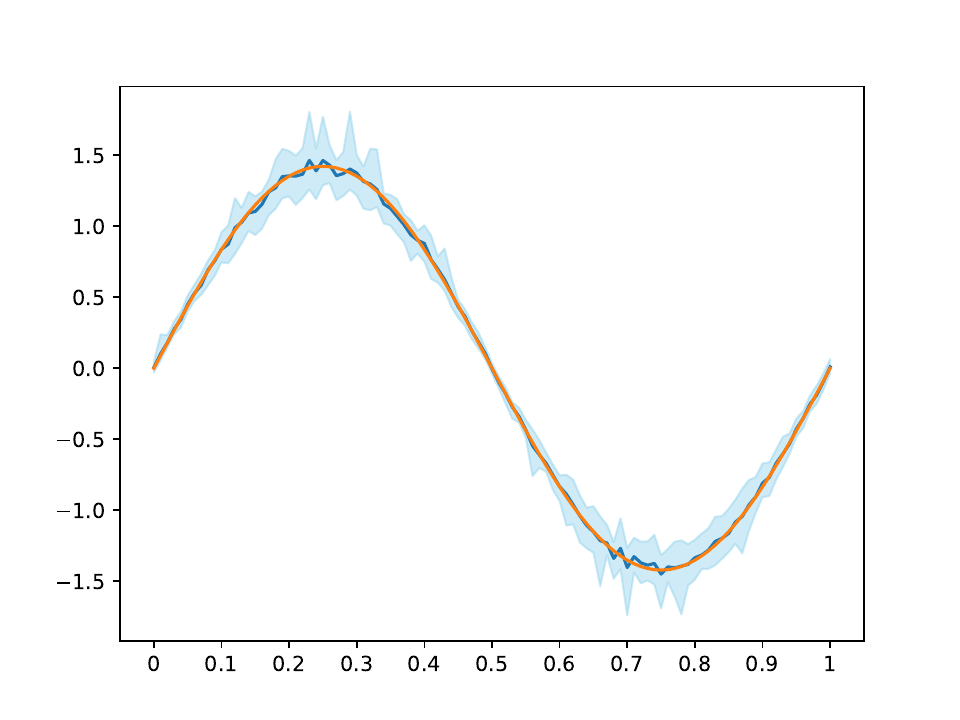}
        \caption{$\gamma=0.4$, $\tau = -0.9$, $\kappa=0.5$.}
    \end{subfigure}

    \medskip 

    \begin{subfigure}{0.325\textwidth}
        \includegraphics[scale=0.35]{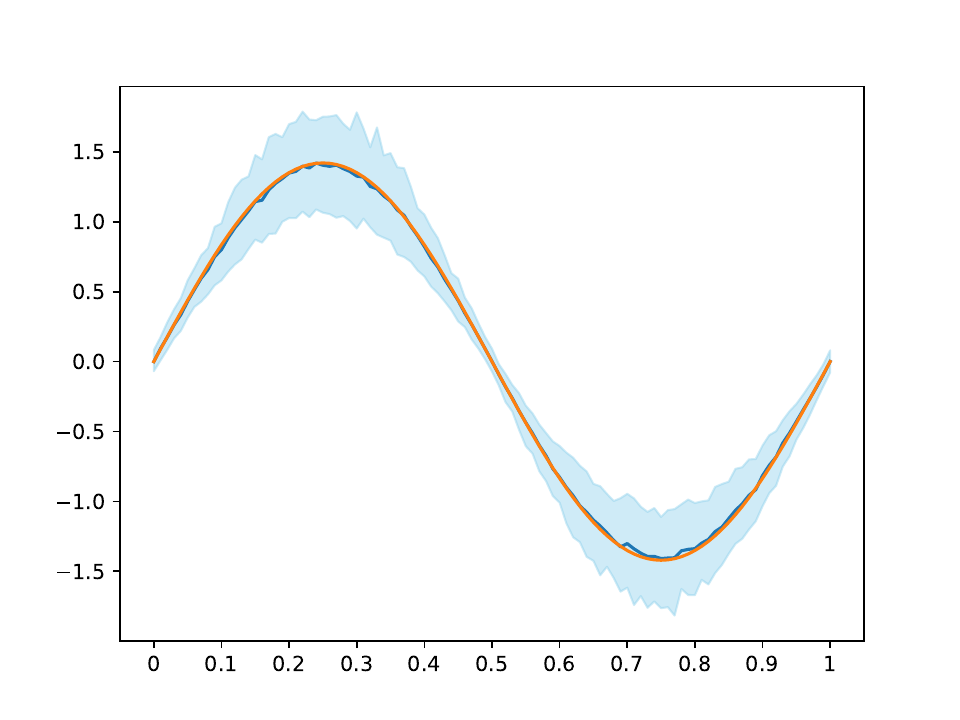}
        \caption{$\gamma=0.7$, $\tau = -0.1$, $\kappa=0.5$.}
    \end{subfigure}
    \hfill
    \begin{subfigure}{0.325\textwidth}
        \includegraphics[scale=0.35]{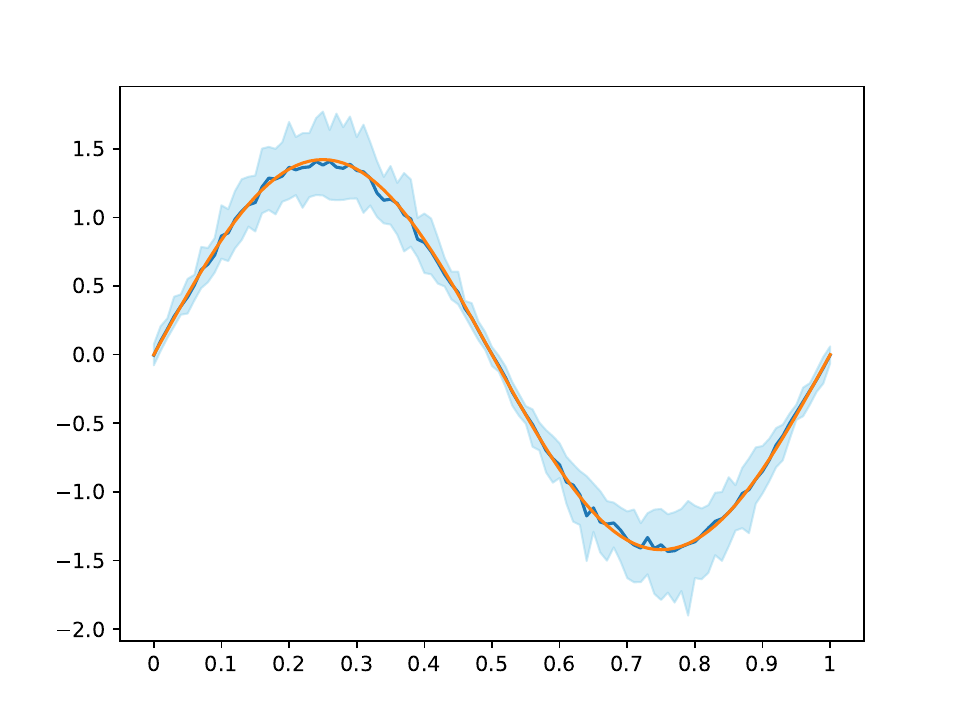}
        \caption{$\gamma=0.7$, $\tau = -0.5$, $\kappa=0.5$.}
    \end{subfigure}
    \hfill
    \begin{subfigure}{0.325\textwidth}
        \includegraphics[scale=0.35]{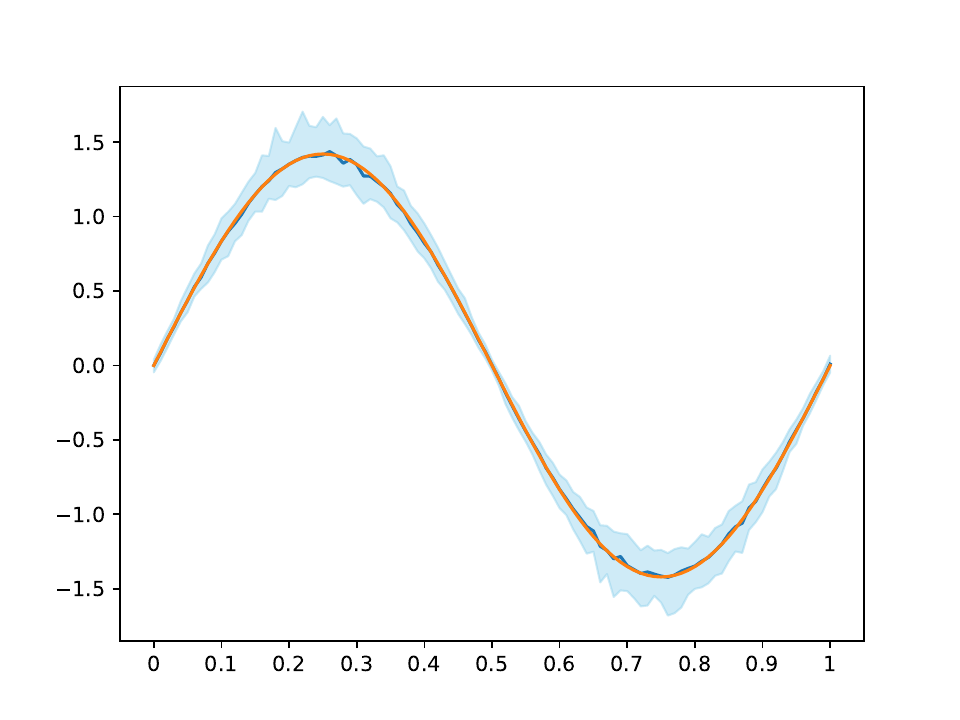}
        \caption{$\gamma=0.7$, $\tau = -0.9$, $\kappa=0.5$.}
    \end{subfigure}
      
\caption{Simulation results on the inverse model with serial dependence of type 'pathological ARMA-resp. + classic GARCH noise', \emph{i.e.,} $(\phi_{\rm resp},\theta_{\rm resp})= (0.99,-0.98)$ and $(\omega_{\rm noise},\alpha_{\rm noise},\beta_{\rm noise})=(0.05,0.1,0.85)$.}
    \label{fig:patho_ARMA_st_GARCH}
\end{figure}

\begin{figure}[p]
    \centering
    
    \begin{subfigure}{0.325\textwidth}
        \includegraphics[scale=0.35]{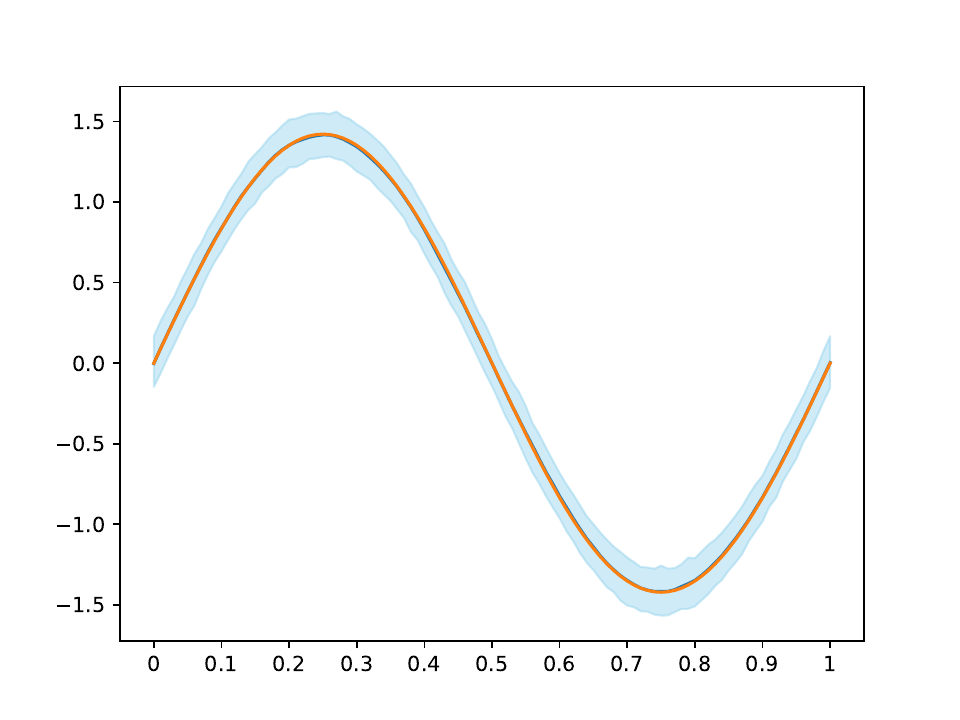}
        \caption{$\gamma=0.1$, $\tau = -0.1$, $\kappa=0.5$.}
    \end{subfigure}
    \hfill
    \begin{subfigure}{0.325\textwidth}
        \includegraphics[scale=0.35]{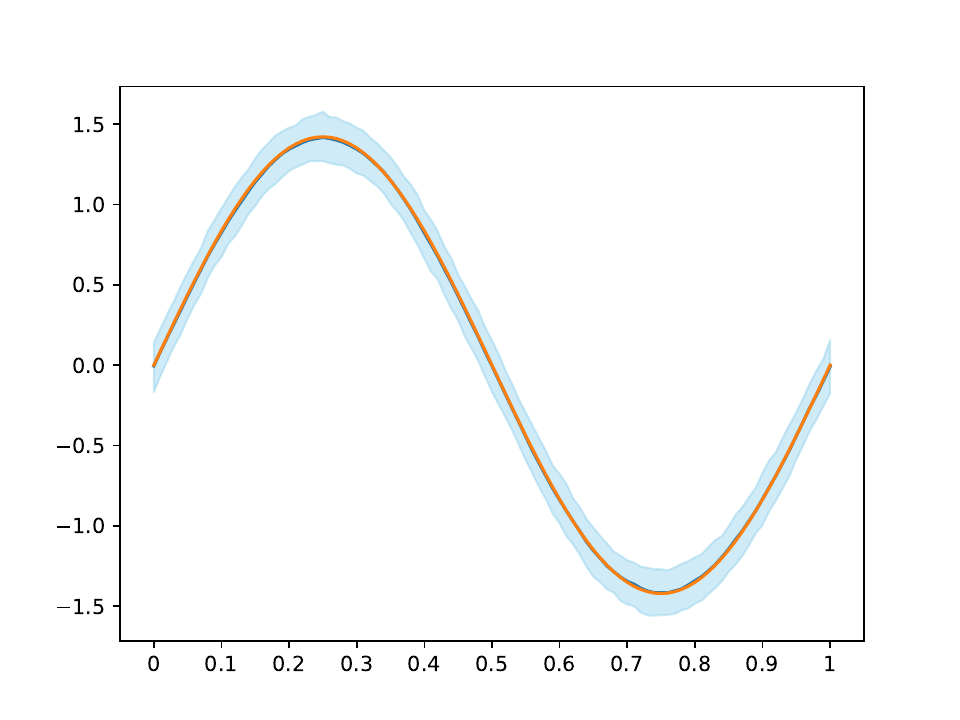}
        \caption{$\gamma=0.1$, $\tau = -0.5$, $\kappa=0.5$.}
    \end{subfigure}
    \hfill
    \begin{subfigure}{0.325\textwidth}
        \includegraphics[scale=0.35]{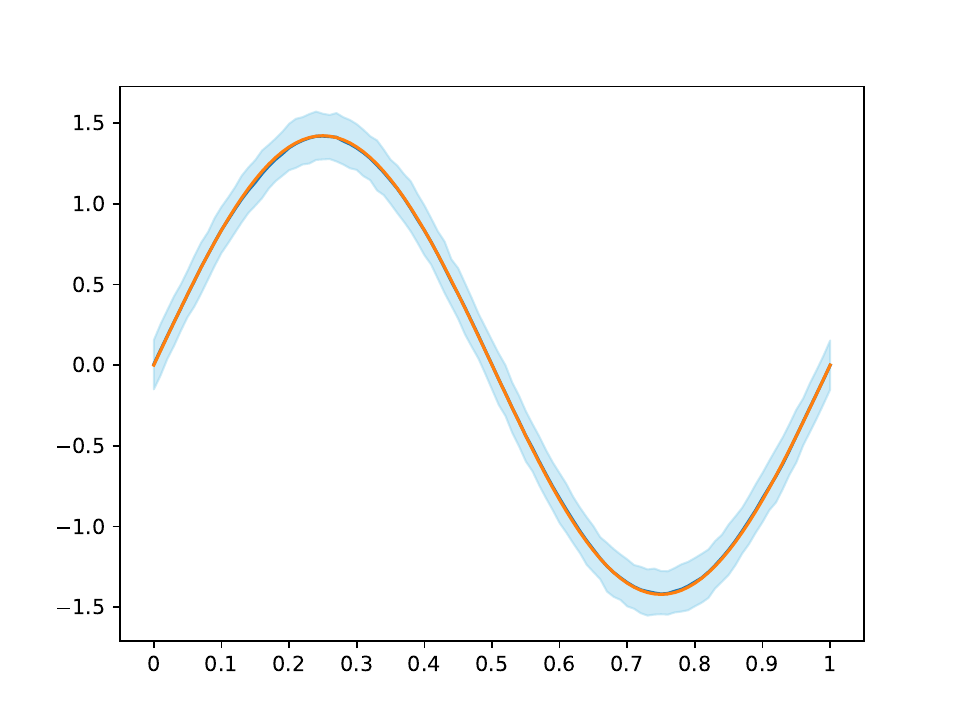}
        \caption{$\gamma=0.1$, $\tau = -0.9$, $\kappa=0.5$.}
    \end{subfigure}

    \medskip 
    \begin{subfigure}{0.325\textwidth}
        \includegraphics[scale=0.35]{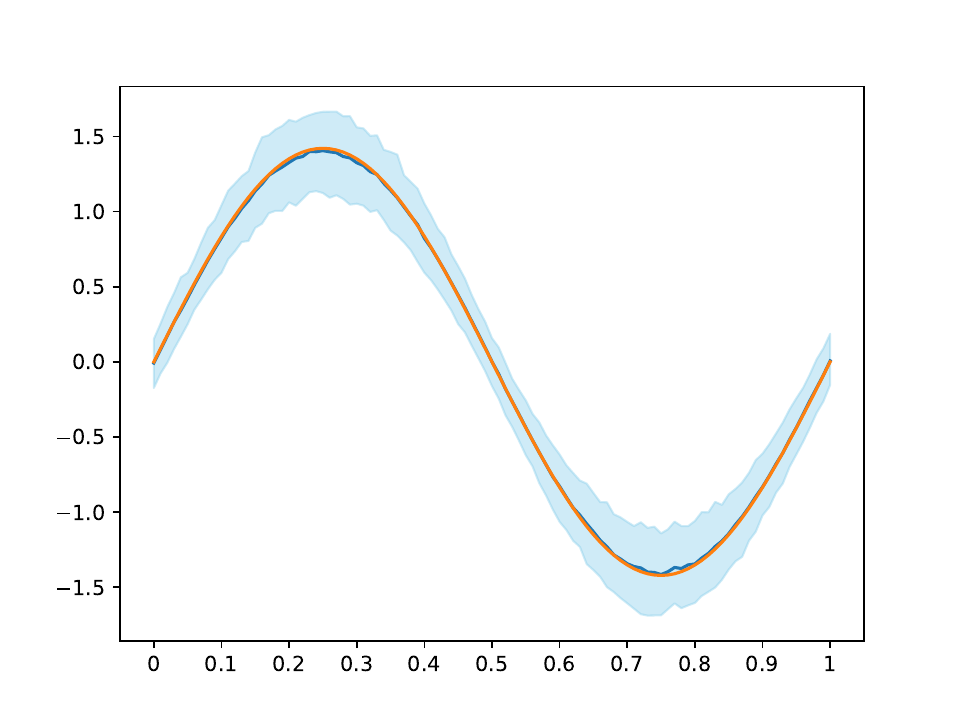}
        \caption{$\gamma=0.4$, $\tau = -0.1$, $\kappa=0.5$.}
    \end{subfigure}
    \hfill
    \begin{subfigure}{0.325\textwidth}
        \includegraphics[scale=0.35]{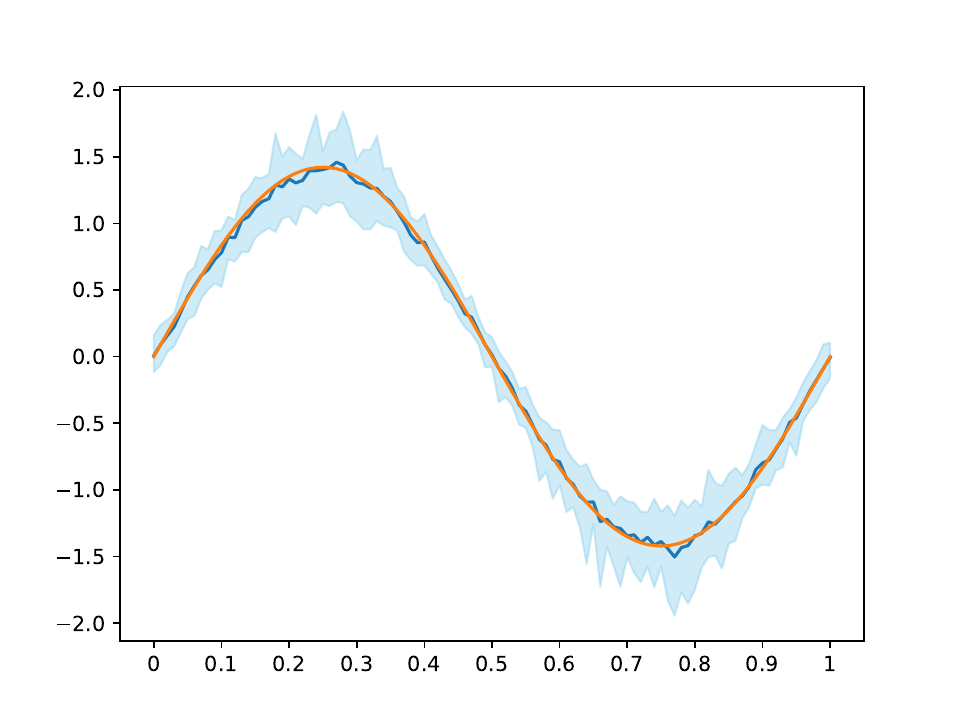}
        \caption{$\gamma=0.4$, $\tau = -0.5$, $\kappa=0.5$.}
    \end{subfigure}
    \hfill
    \begin{subfigure}{0.325\textwidth}
        \includegraphics[scale=0.35]{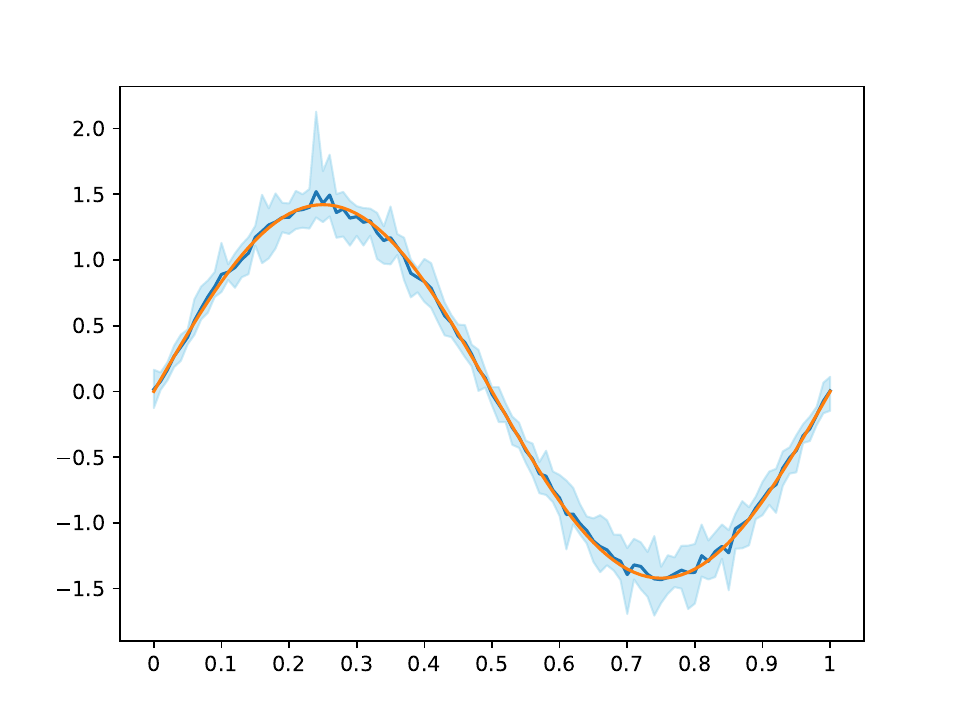}
        \caption{$\gamma=0.4$, $\tau = -0.9$, $\kappa=0.5$.}
    \end{subfigure}

    \medskip 

    \begin{subfigure}{0.325\textwidth}
        \includegraphics[scale=0.35]{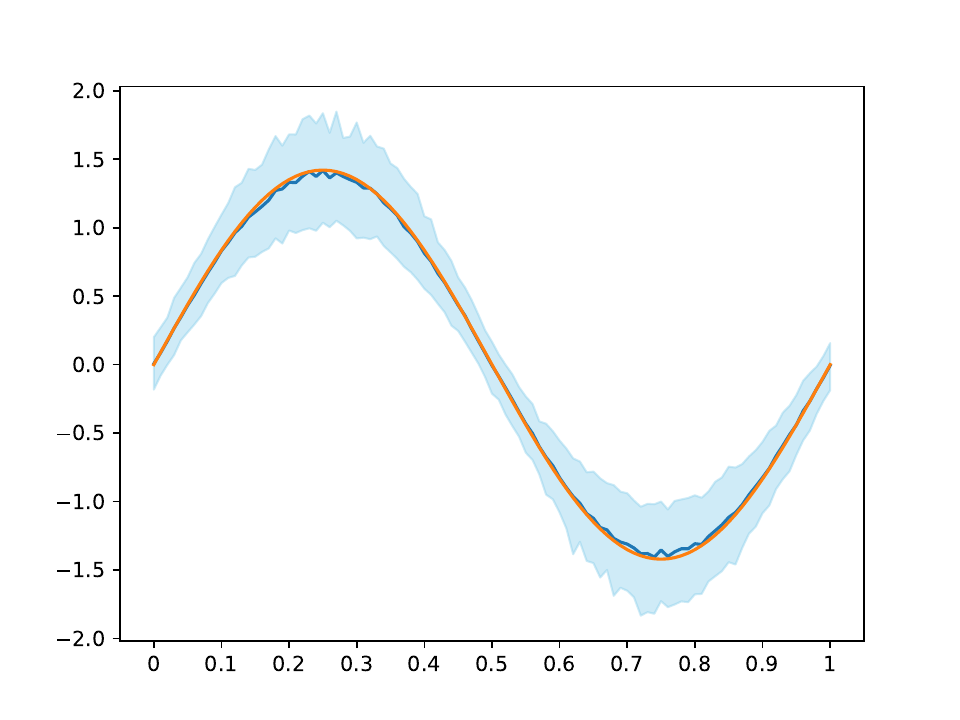}
        \caption{$\gamma=0.7$, $\tau = -0.1$, $\kappa=0.5$.}
    \end{subfigure}
    \hfill
    \begin{subfigure}{0.325\textwidth}
        \includegraphics[scale=0.35]{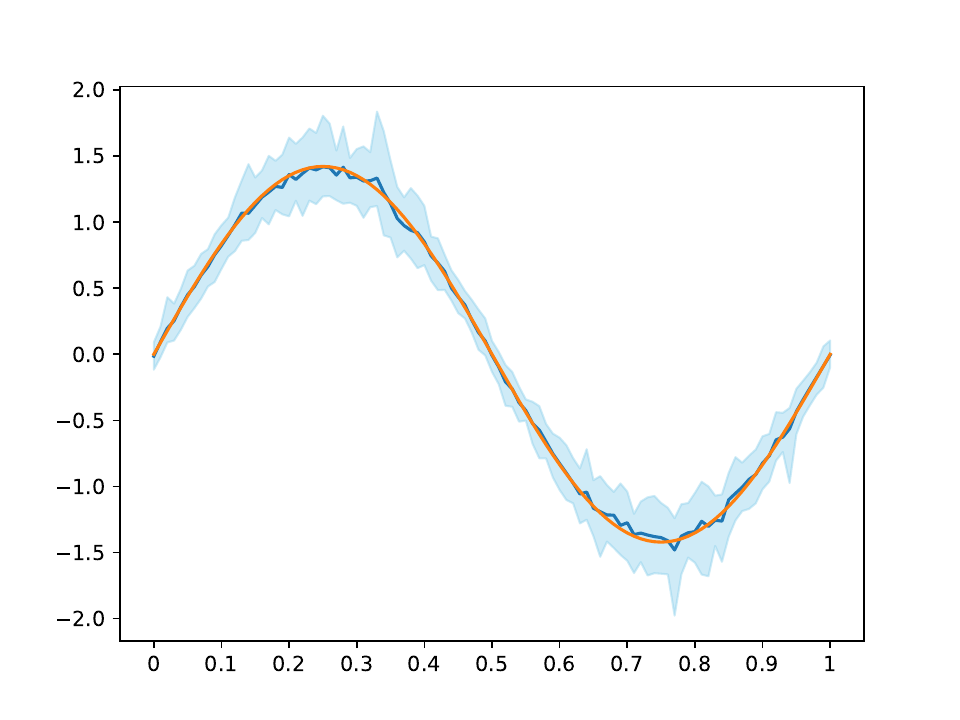}
        \caption{$\gamma=0.7$, $\tau = -0.5$, $\kappa=0.5$.}
    \end{subfigure}
    \hfill
    \begin{subfigure}{0.325\textwidth}
        \includegraphics[scale=0.35]{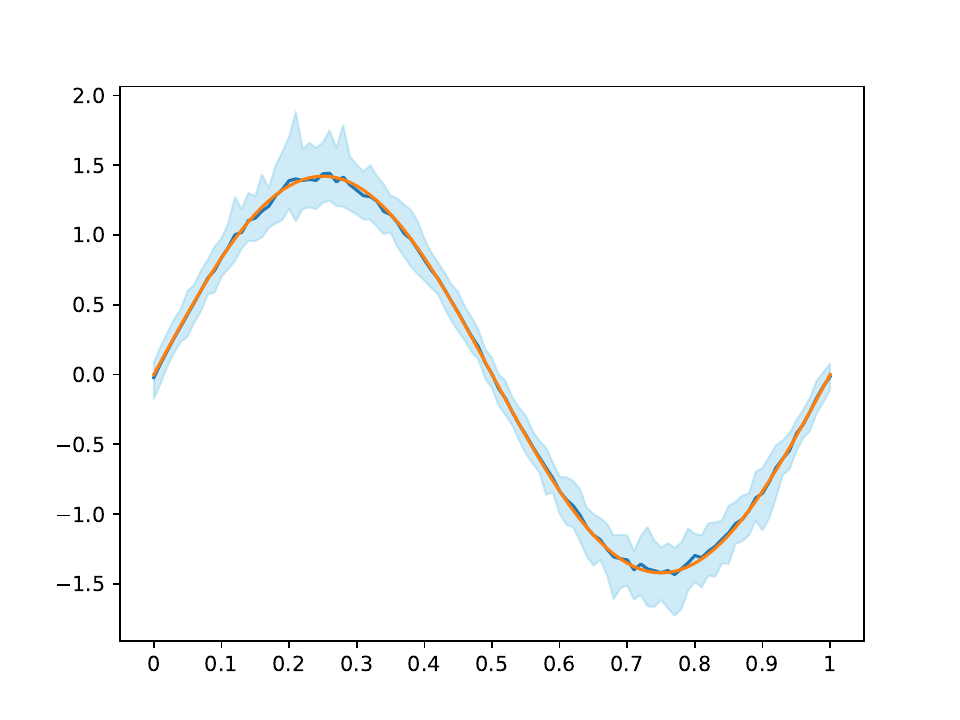}
        \caption{$\gamma=0.7$, $\tau = -0.9$, $\kappa=0.5$.}
    \end{subfigure}
      
\caption{Simulation results on the inverse model with serial dependence of type 'pathological ARMA-resp. + IGARCH-like noise', \emph{i.e.,} $(\phi_{\rm resp},\theta_{\rm resp})= (0.99,-0.98)$ and $(\omega_{\rm noise},\alpha_{\rm noise},\beta_{\rm noise})=(0.05,0.05,0.94)$.}
    \label{fig:patho_ARMA_patho_GARCH}
\end{figure}

\begin{figure}[p]
    \centering
    
    \begin{subfigure}{0.325\textwidth}
        \includegraphics[scale=0.35]{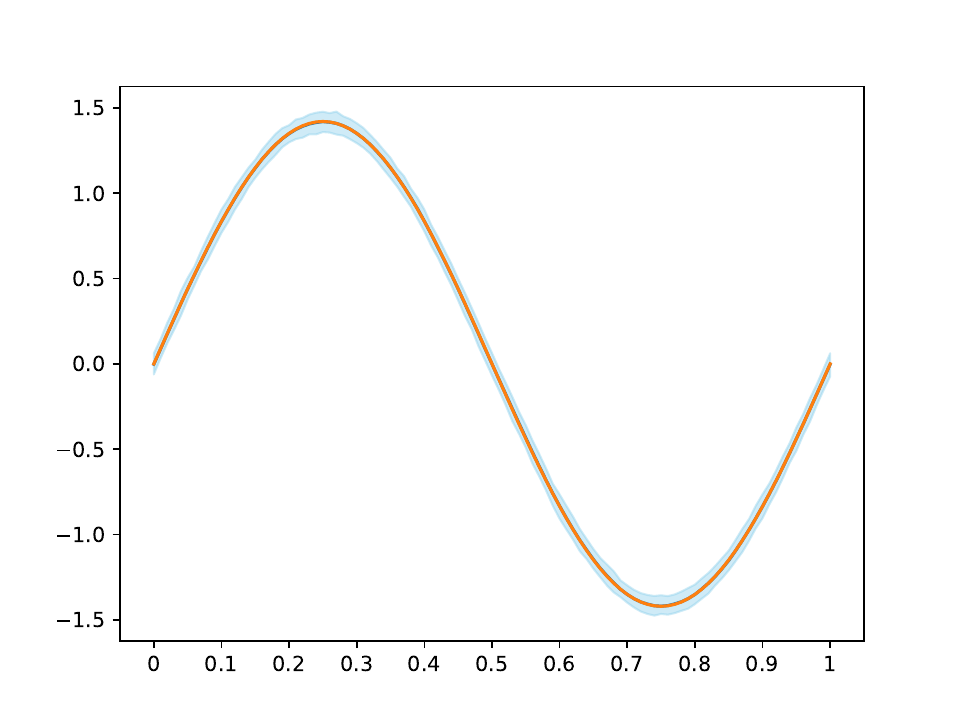}
        \caption{$\gamma=0.1$, $\tau = -0.1$, $\kappa=0.5$.}
    \end{subfigure}
    \hfill
    \begin{subfigure}{0.325\textwidth}
        \includegraphics[scale=0.35]{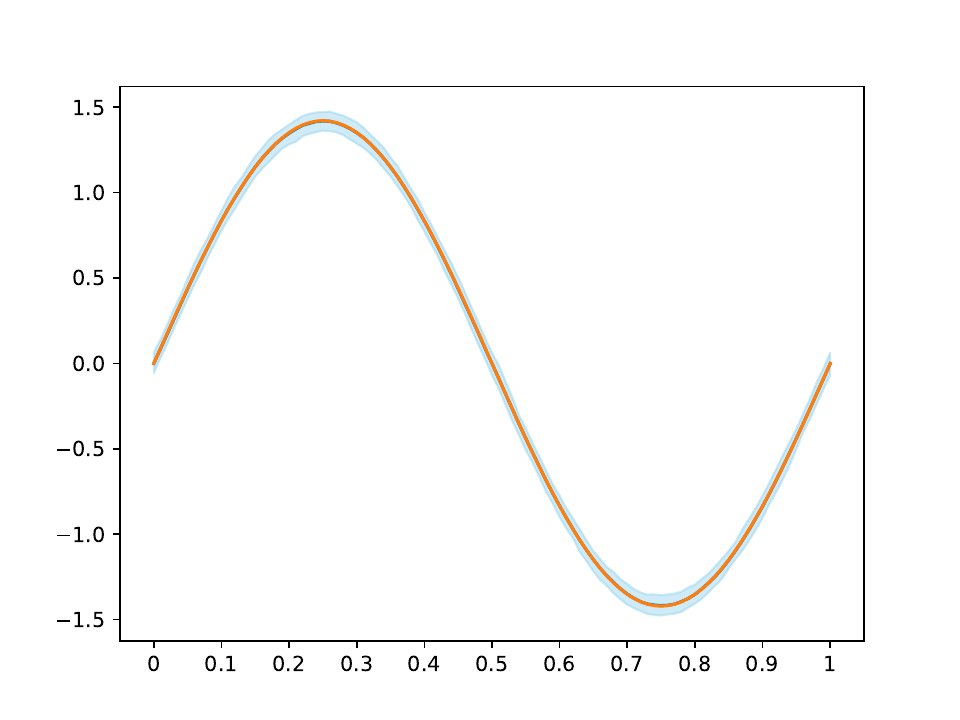}
        \caption{$\gamma=0.1$, $\tau = -0.5$, $\kappa=0.5$.}
    \end{subfigure}
    \hfill
    \begin{subfigure}{0.325\textwidth}
        \includegraphics[scale=0.35]{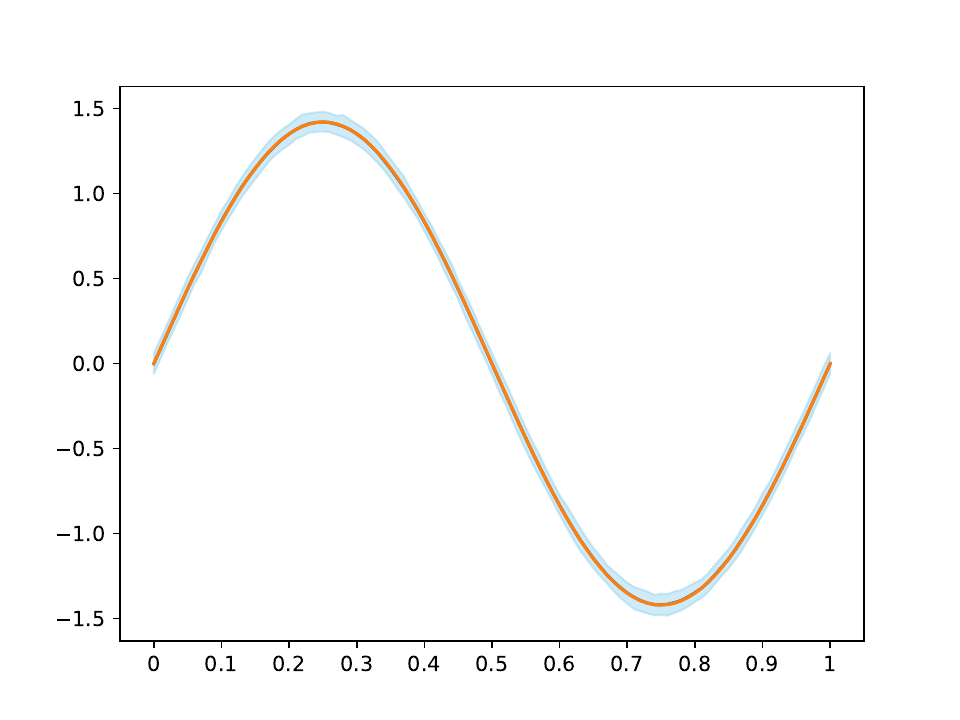}
        \caption{$\gamma=0.1$, $\tau = -0.9$, $\kappa=0.5$.}
    \end{subfigure}

    \medskip 
   \begin{subfigure}{0.325\textwidth}
        \includegraphics[scale=0.35]{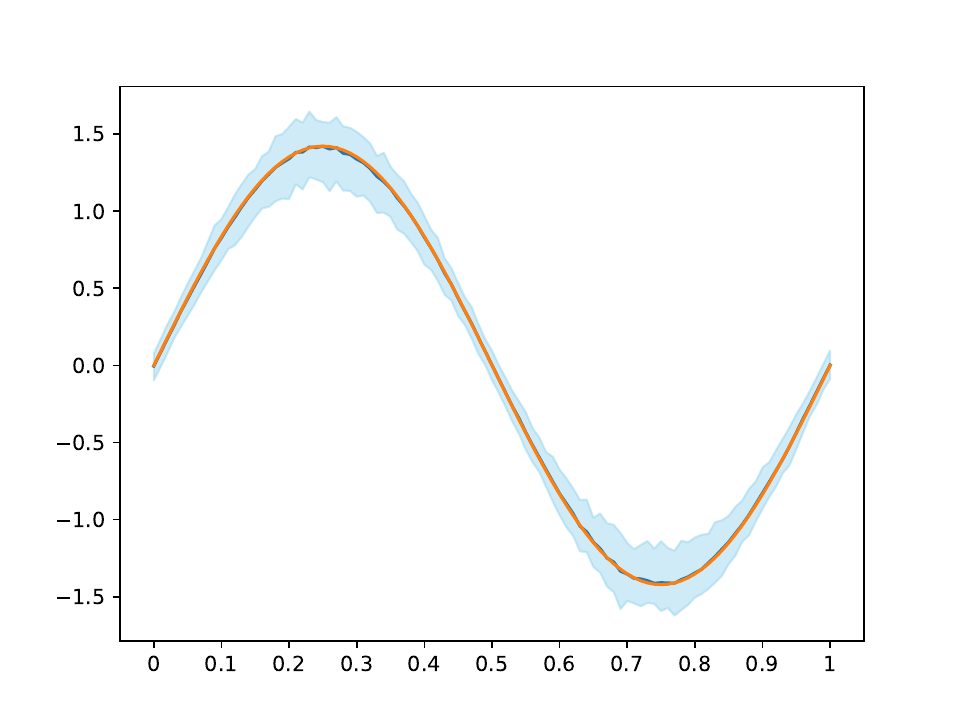}
        \caption{$\gamma=0.4$, $\tau = -0.1$, $\kappa=0.5$.}
    \end{subfigure}
    \hfill
    \begin{subfigure}{0.325\textwidth}
        \includegraphics[scale=0.35]{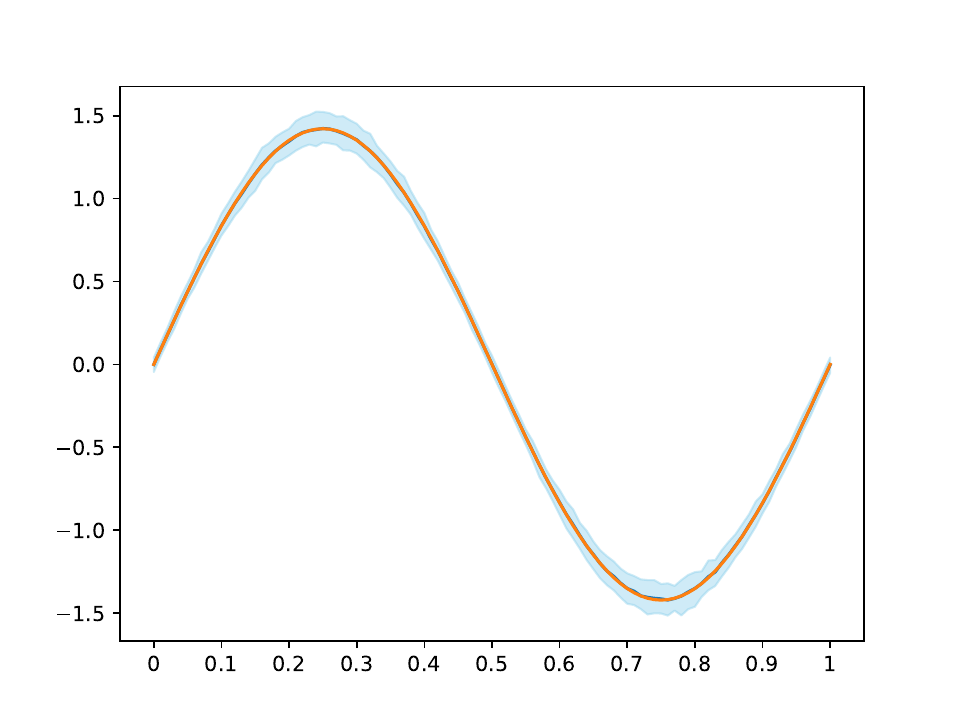}
        \caption{$\gamma=0.4$, $\tau = -0.5$, $\kappa=0.5$.}
    \end{subfigure}
    \hfill
    \begin{subfigure}{0.325\textwidth}
        \includegraphics[scale=0.35]{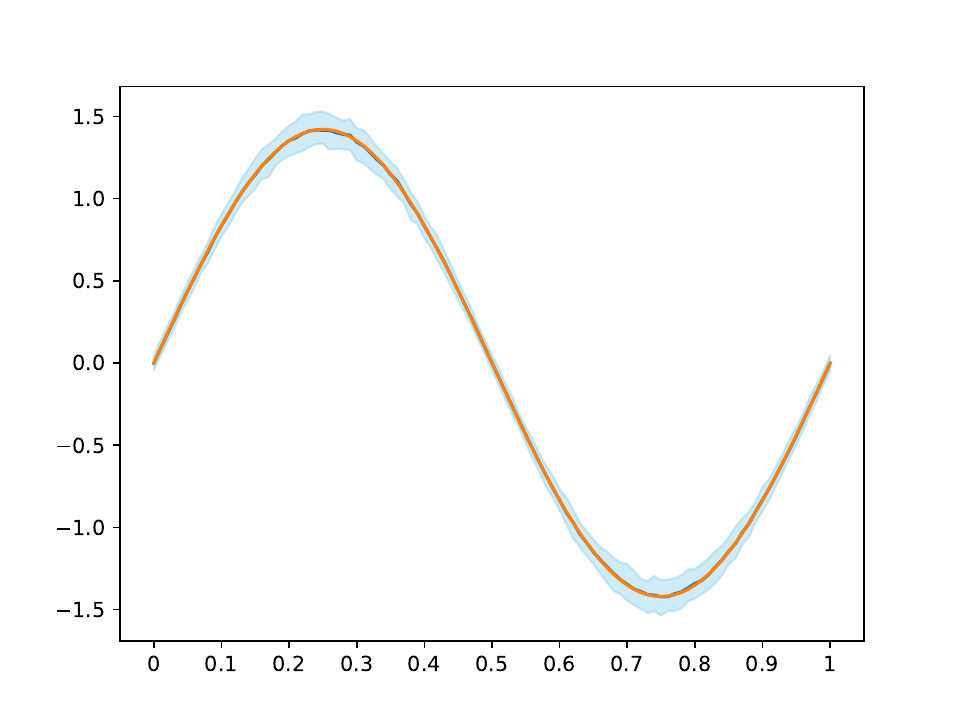}
        \caption{$\gamma=0.4$, $\tau = -0.9$, $\kappa=0.5$.}
    \end{subfigure}

    \medskip 

  \begin{subfigure}{0.325\textwidth}
        \includegraphics[scale=0.35]{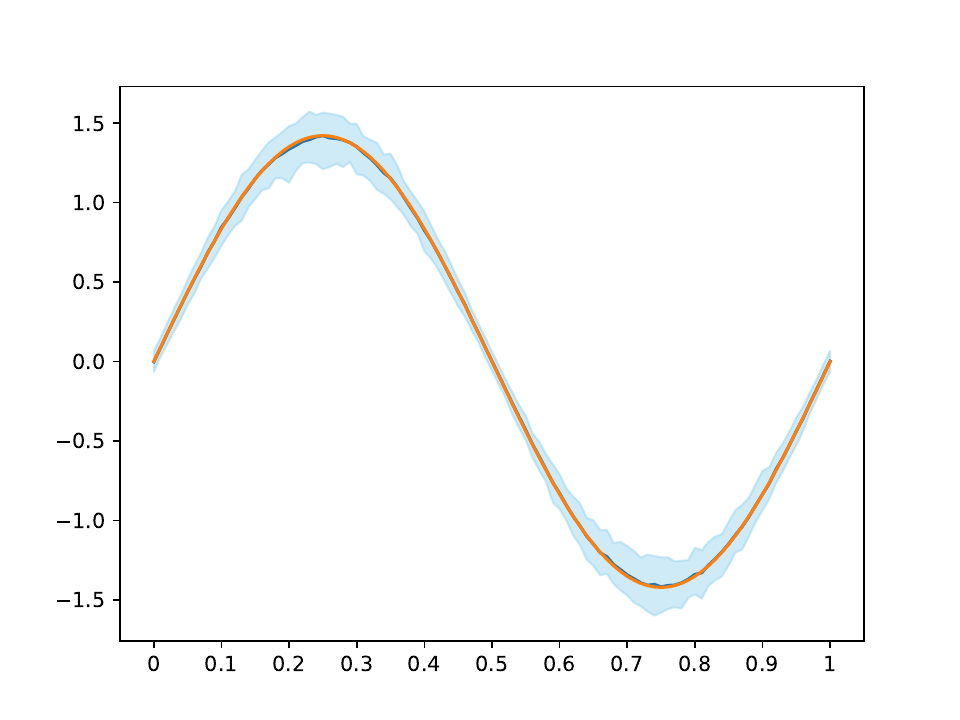}
        \caption{$\gamma=0.5$, $\tau = -0.1$, $\kappa=0.5$.}
    \end{subfigure}
    \hfill
    \begin{subfigure}{0.325\textwidth}
        \includegraphics[scale=0.35]{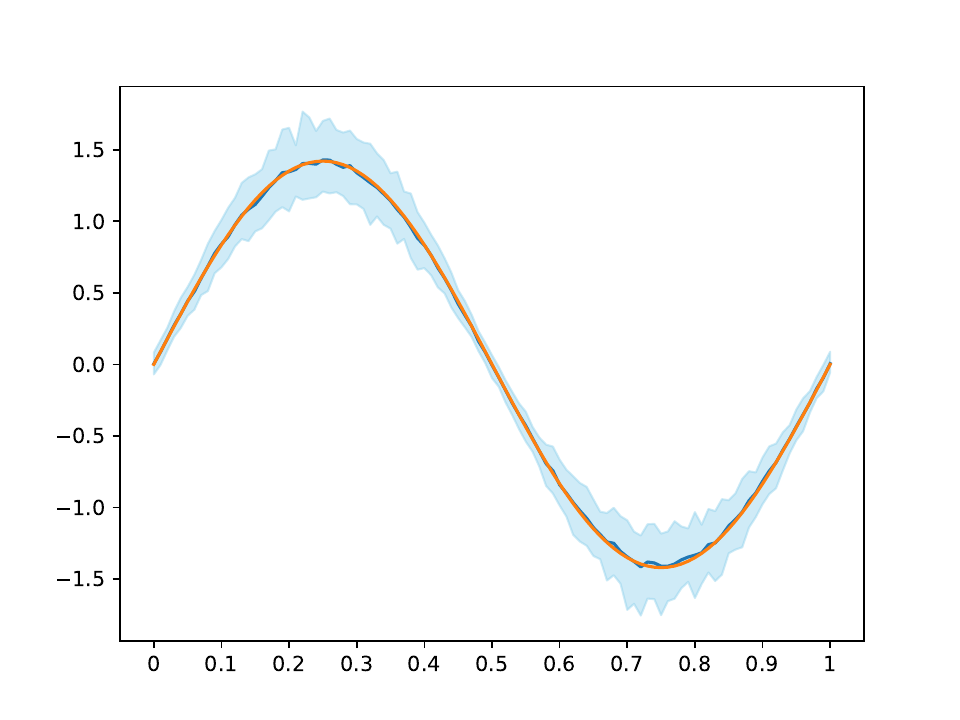}
        \caption{$\gamma=0.5$, $\tau = -0.5$, $\kappa=0.5$.}
    \end{subfigure}
    \hfill
    \begin{subfigure}{0.325\textwidth}
        \includegraphics[scale=0.35]{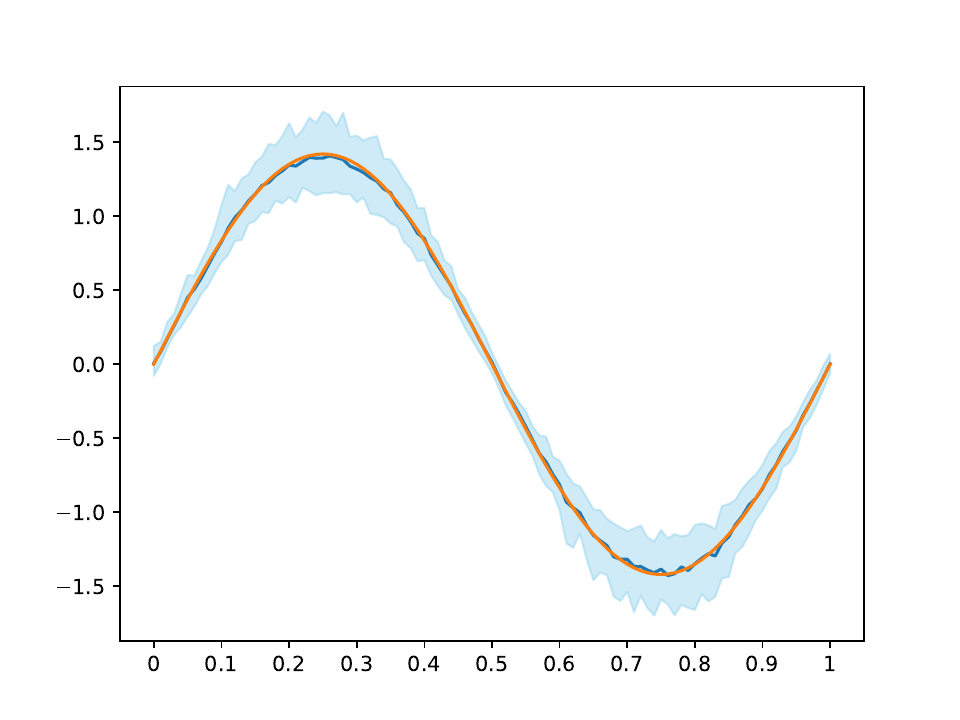}
        \caption{$\gamma=0.5$, $\tau = -0.9$, $\kappa=0.5$.}
    \end{subfigure}
      
\caption{Simulation results on the inverse model with serial dependence of type standard GARCH-resp and standard ARMA-noise, \emph{i.e.,}  $(\omega_{\rm resp},\alpha_{\rm resp},\beta_{\rm resp})=(0.05,0.1,0.85)$ and $(\phi_{\rm noise},\theta_{\rm noise})= (0.8,-0.3)$.}
    \label{fig:st_GARCH_st_ARMA}
\end{figure}

\begin{figure}[p]
    \centering
    
    \begin{subfigure}{0.325\textwidth}
        \includegraphics[scale=0.35]{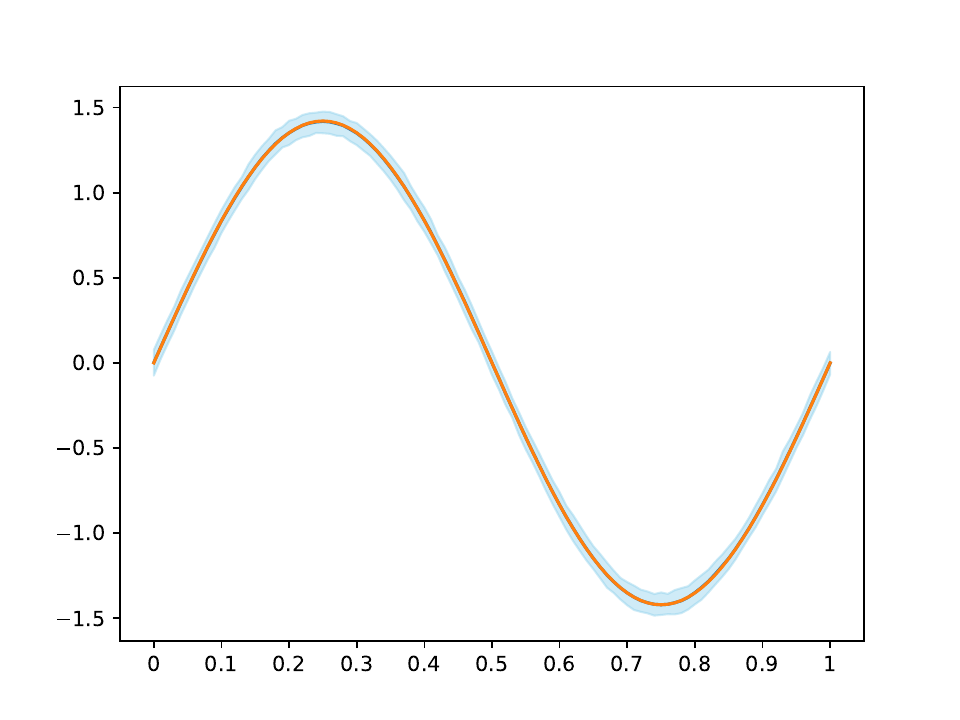}
        \caption{$\gamma=0.1$, $\tau = -0.1$, $\kappa=0.5$.}
    \end{subfigure}
    \hfill
    \begin{subfigure}{0.325\textwidth}
        \includegraphics[scale=0.35]{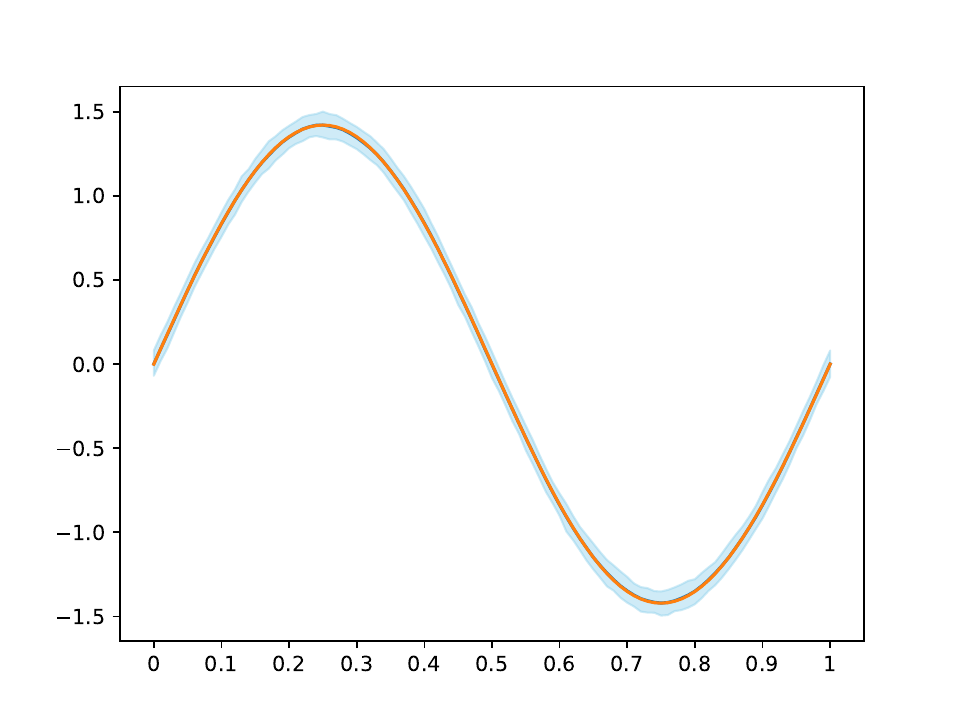}
        \caption{$\gamma=0.1$, $\tau = -0.5$, $\kappa=0.5$.}
    \end{subfigure}
    \hfill
    \begin{subfigure}{0.325\textwidth}
        \includegraphics[scale=0.35]{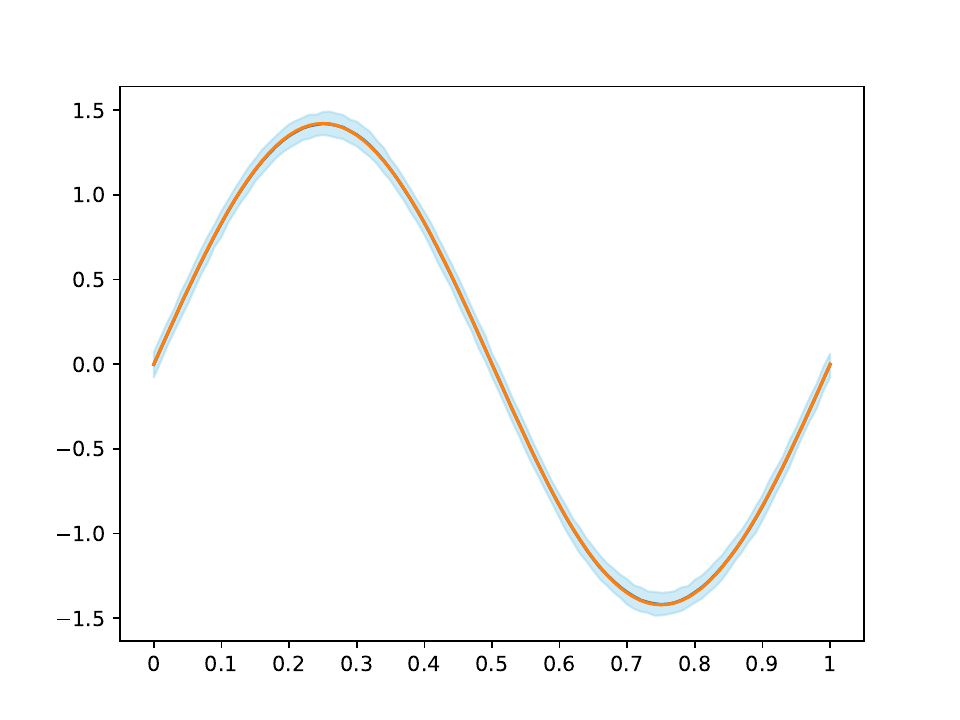}
        \caption{$\gamma=0.1$, $\tau = -0.9$, $\kappa=0.5$.}
    \end{subfigure}

    \medskip 
   \begin{subfigure}{0.325\textwidth}
        \includegraphics[scale=0.35]{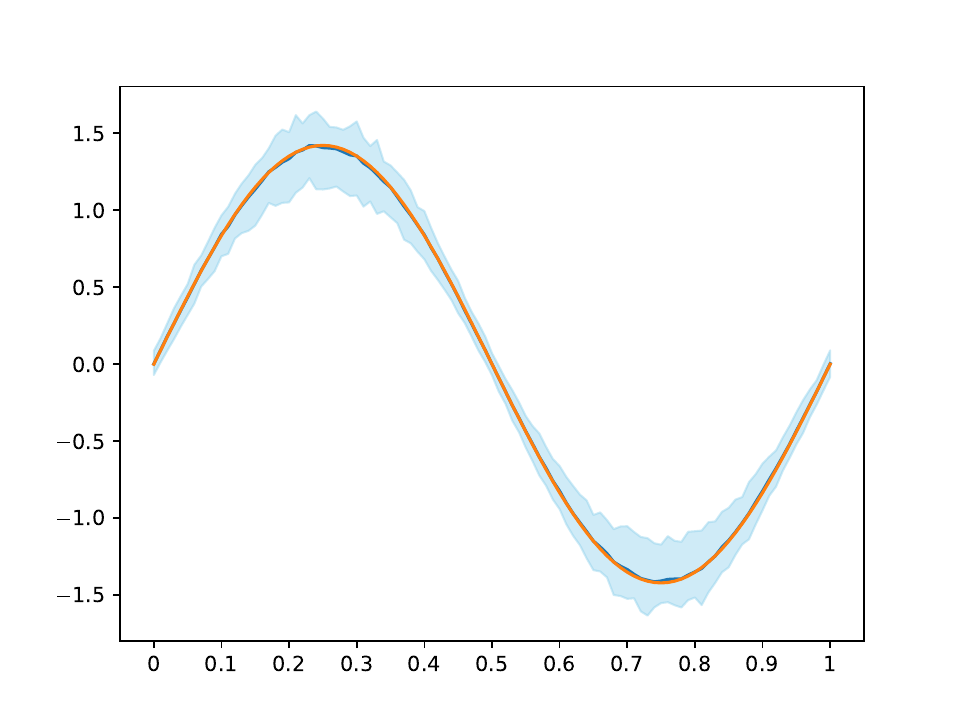}
        \caption{$\gamma=0.4$, $\tau = -0.1$, $\kappa=0.5$.}
    \end{subfigure}
    \hfill
    \begin{subfigure}{0.325\textwidth}
        \includegraphics[scale=0.35]{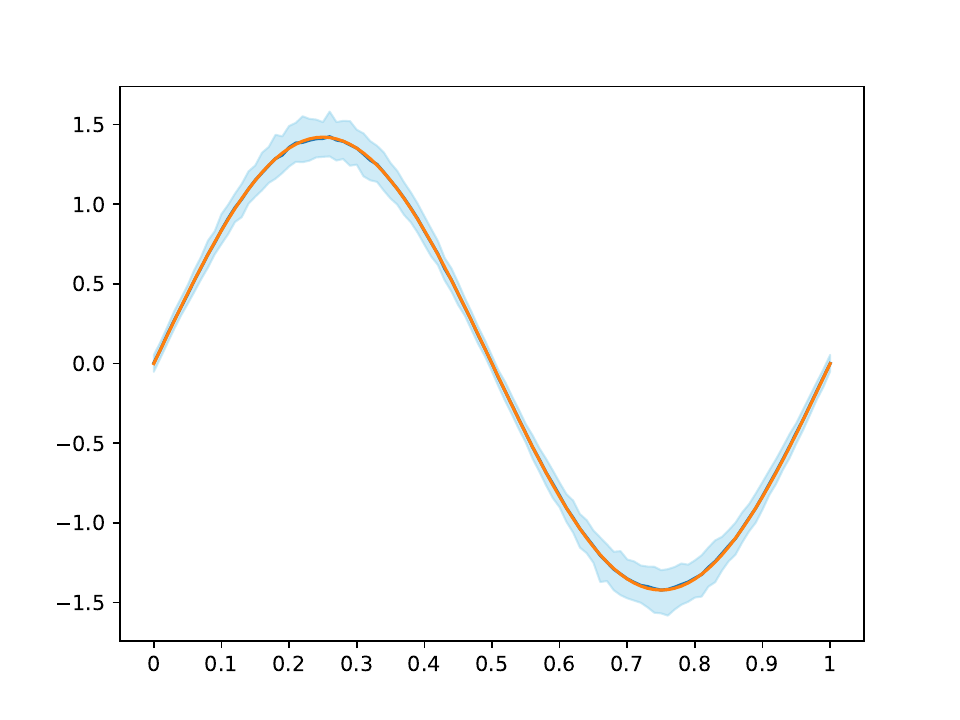}
        \caption{$\gamma=0.4$, $\tau = -0.5$, $\kappa=0.5$.}
    \end{subfigure}
    \hfill
    \begin{subfigure}{0.325\textwidth}
        \includegraphics[scale=0.35]{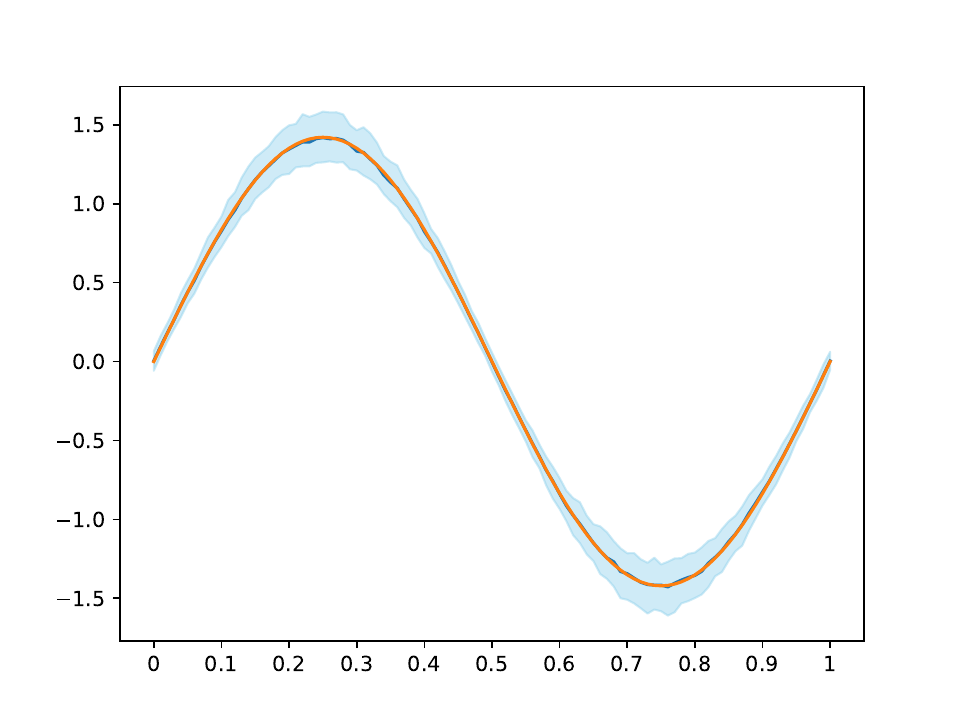}
        \caption{$\gamma=0.4$, $\tau = -0.9$, $\kappa=0.5$.}
    \end{subfigure}

    \medskip 

  \begin{subfigure}{0.325\textwidth}
        \includegraphics[scale=0.35]{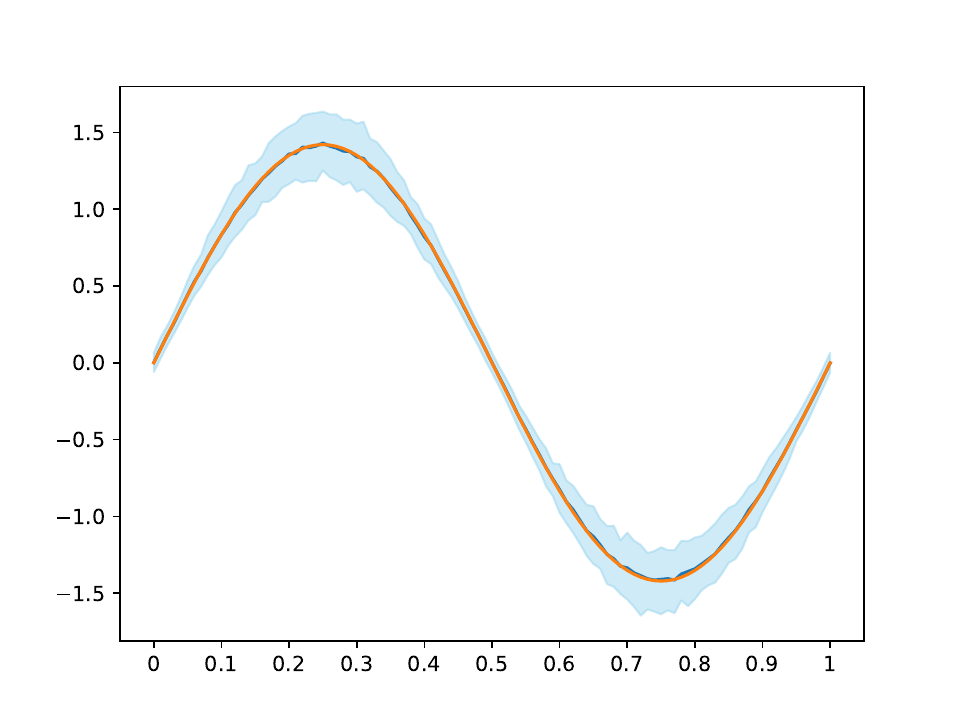}
        \caption{$\gamma=0.5$, $\tau = -0.1$, $\kappa=0.5$.}
    \end{subfigure}
    \hfill
    \begin{subfigure}{0.325\textwidth}
        \includegraphics[scale=0.35]{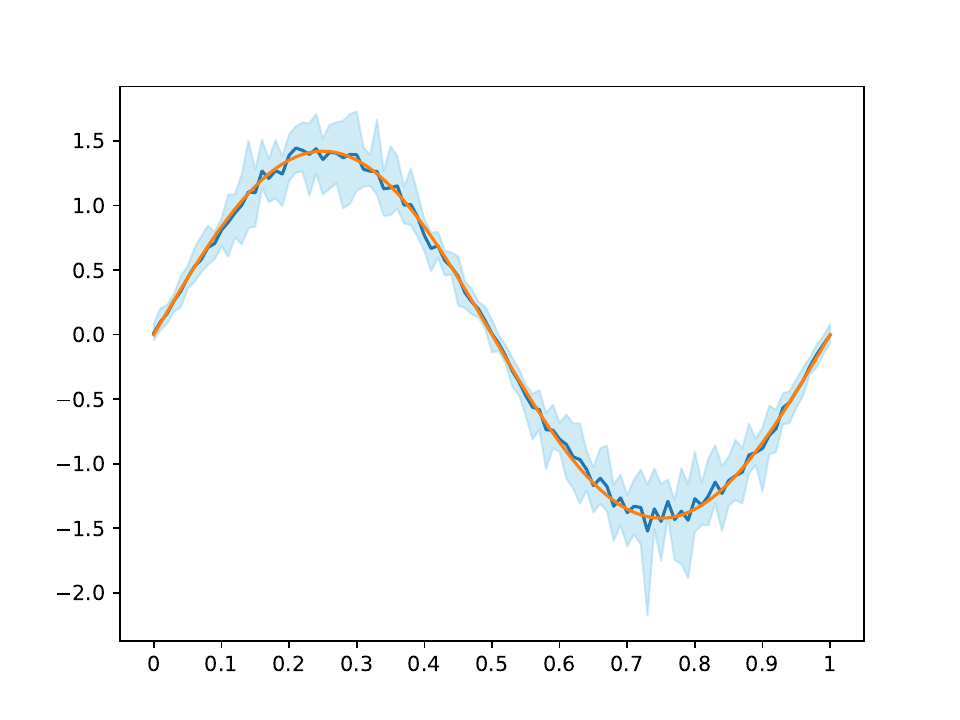}
        \caption{$\gamma=0.5$, $\tau = -0.5$, $\kappa=0.5$.}
    \end{subfigure}
    \hfill
    \begin{subfigure}{0.325\textwidth}
        \includegraphics[scale=0.35]{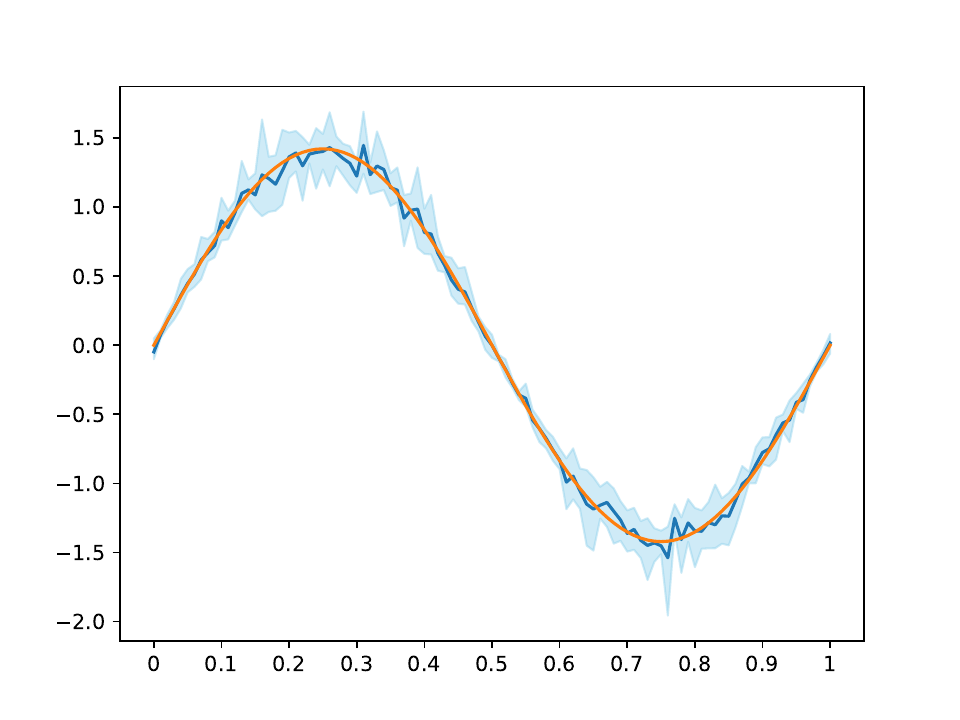}
        \caption{$\gamma=0.5$, $\tau = -0.9$, $\kappa=0.5$.}
    \end{subfigure}
      
\caption{Simulation results on the inverse model with serial dependence of type IGARCH-like response and standard ARMA-noise, \emph{i.e.,}  $(\omega_{\rm resp},\alpha_{\rm resp},\beta_{\rm resp})=(0.05,0.05,0.94)$ and $(\phi_{\rm noise},\theta_{\rm noise})= (0.8,-0.3)$.}
    \label{fig:patho_GARCH_st_ARMA}
\end{figure}

\begin{figure}[p]
    \centering
    
    \begin{subfigure}{0.325\textwidth}
        \includegraphics[scale=0.35]{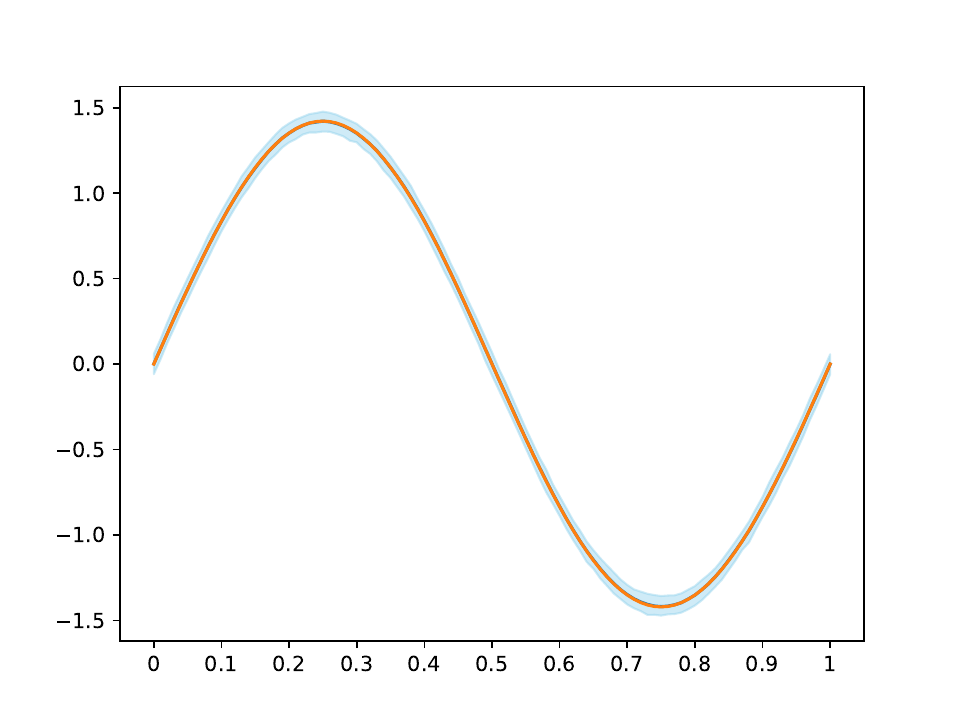}
        \caption{$\gamma=0.1$, $\tau = -0.1$, $\kappa=0.5$.}
    \end{subfigure}
    \hfill
    \begin{subfigure}{0.325\textwidth}
        \includegraphics[scale=0.35]{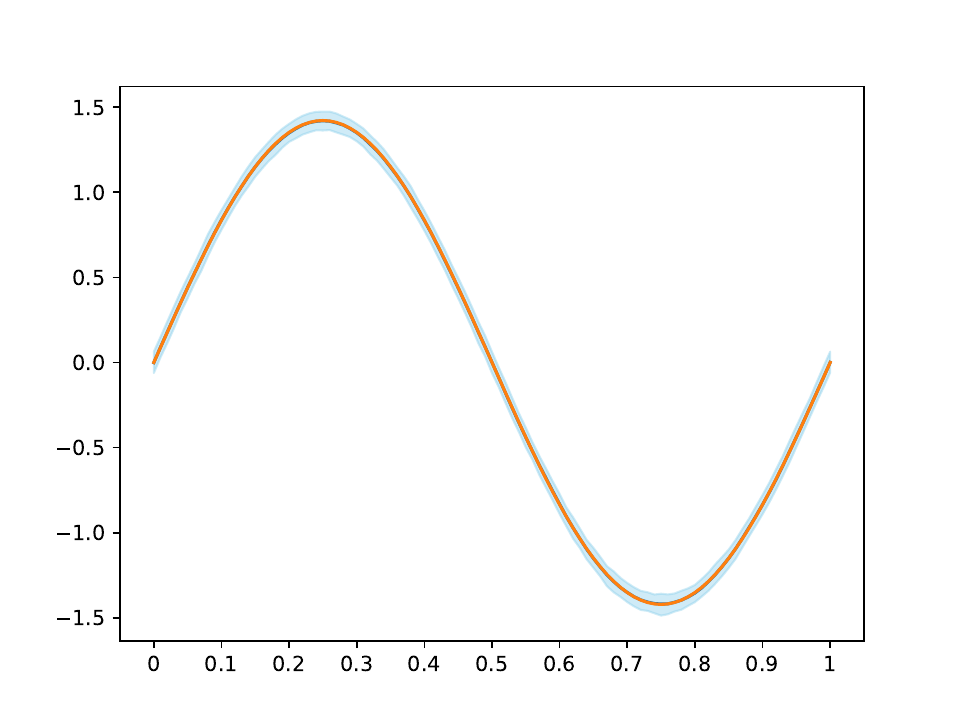}
        \caption{$\gamma=0.1$, $\tau = -0.5$, $\kappa=0.5$.}
    \end{subfigure}
    \hfill
    \begin{subfigure}{0.325\textwidth}
        \includegraphics[scale=0.35]{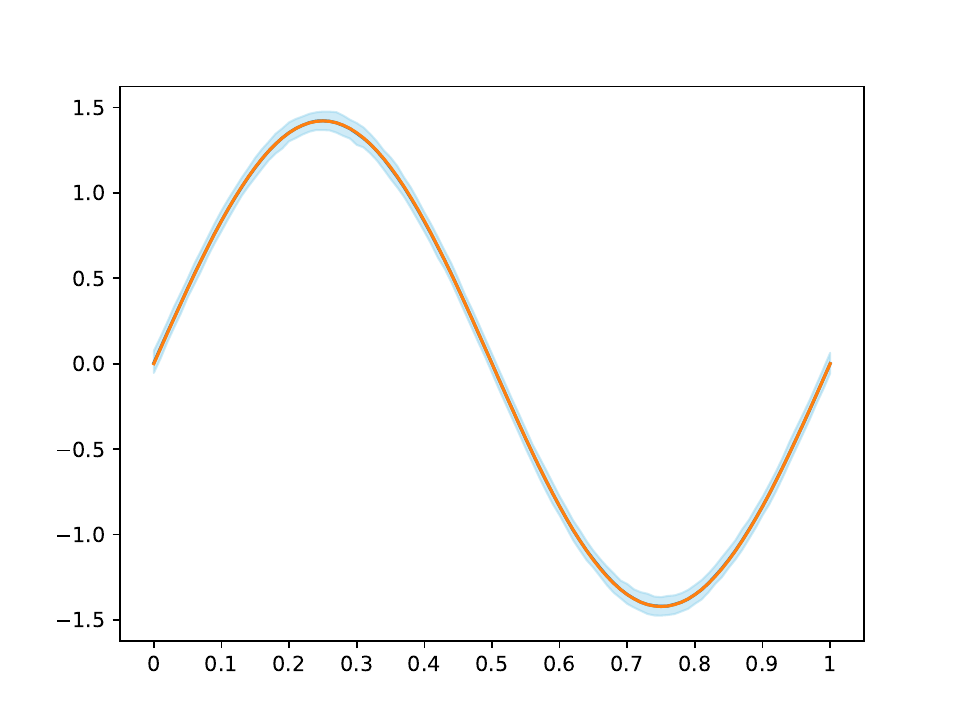}
        \caption{$\gamma=0.1$, $\tau = -0.9$, $\kappa=0.5$.}
    \end{subfigure}

    \medskip 
   \begin{subfigure}{0.325\textwidth}
        \includegraphics[scale=0.35]{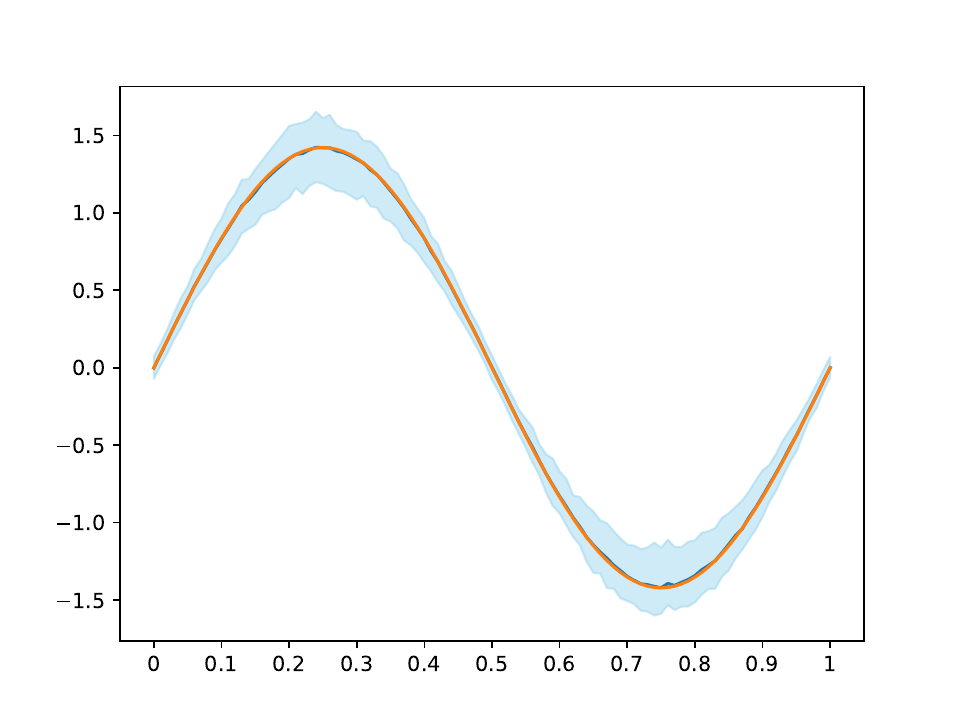}
        \caption{$\gamma=0.4$, $\tau = -0.1$, $\kappa=0.5$.}
    \end{subfigure}
    \hfill
    \begin{subfigure}{0.325\textwidth}
        \includegraphics[scale=0.35]{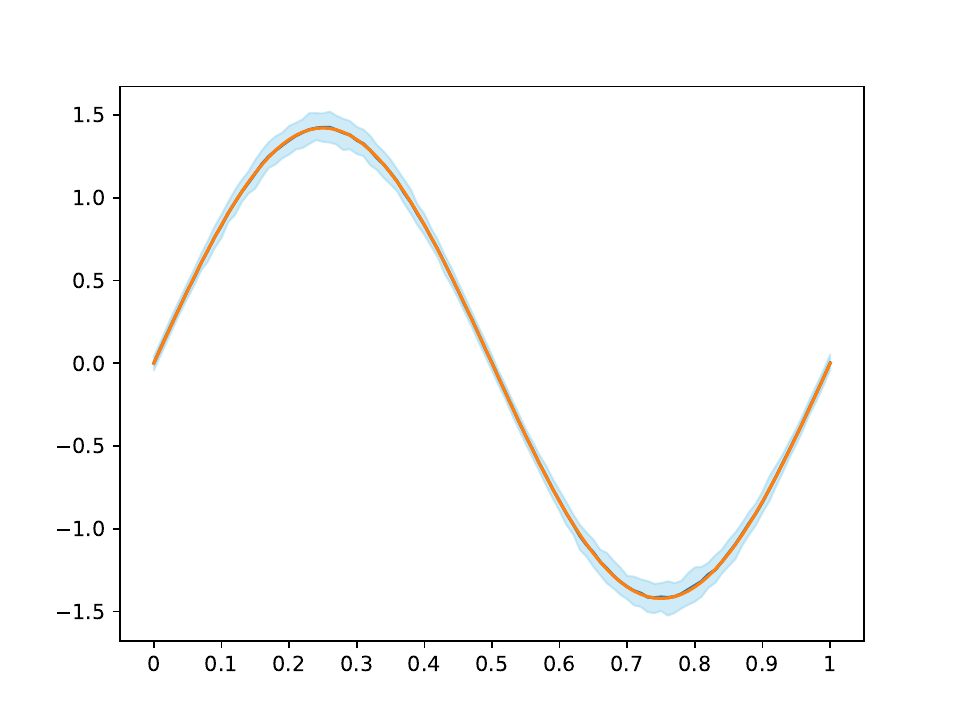}
        \caption{$\gamma=0.4$, $\tau = -0.5$, $\kappa=0.5$.}
    \end{subfigure}
    \hfill
    \begin{subfigure}{0.325\textwidth}
        \includegraphics[scale=0.35]{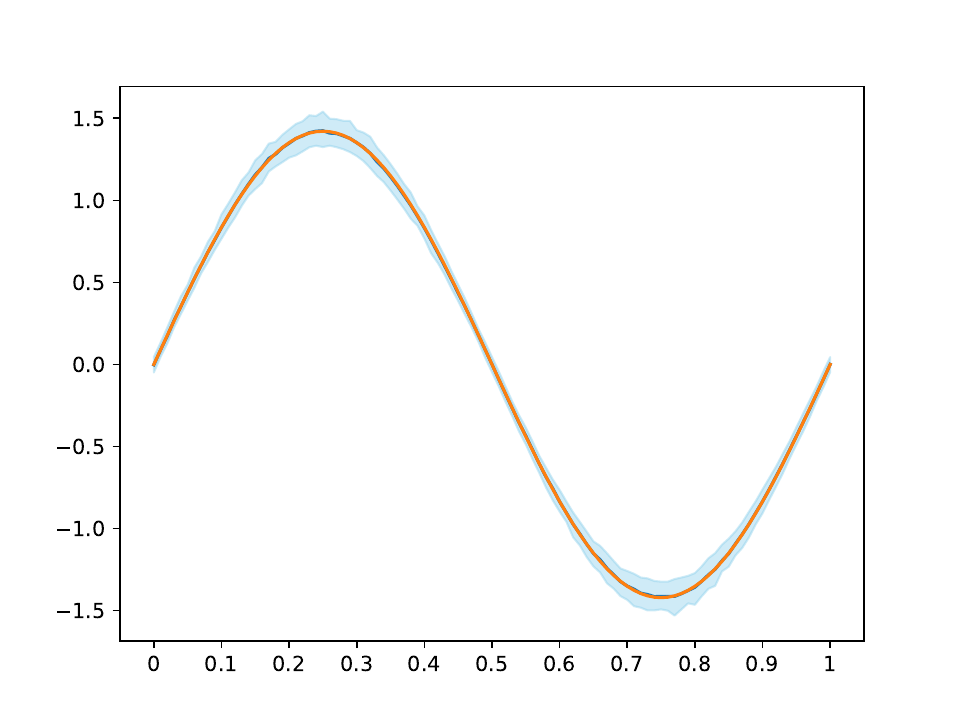}
        \caption{$\gamma=0.4$, $\tau = -0.9$, $\kappa=0.5$.}
    \end{subfigure}

    \medskip 

  \begin{subfigure}{0.325\textwidth}
        \includegraphics[scale=0.35]{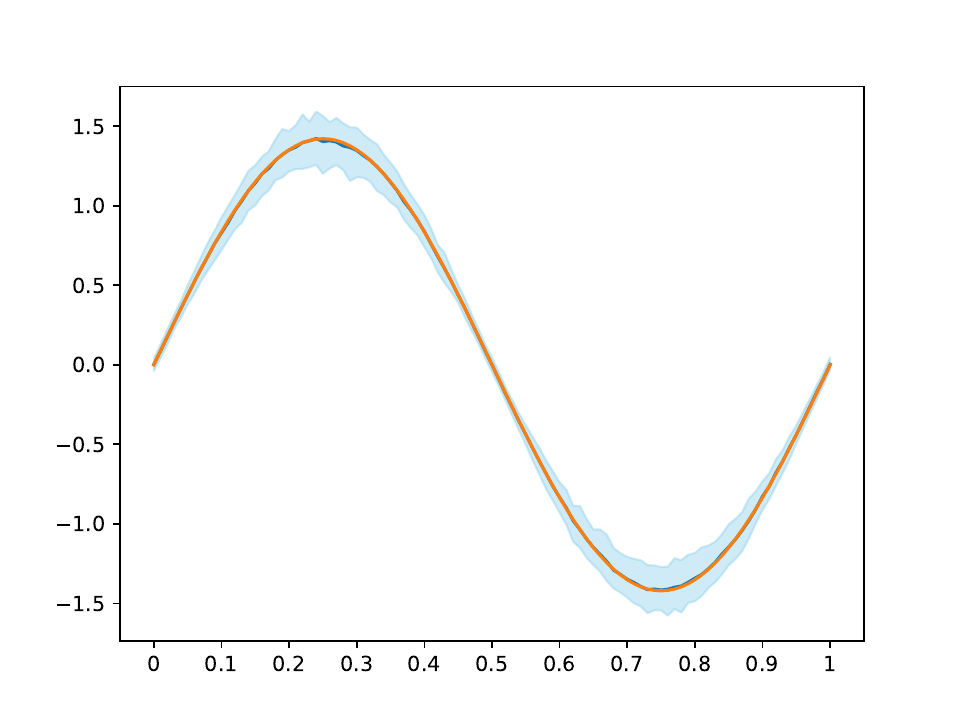}
        \caption{$\gamma=0.5$, $\tau = -0.1$, $\kappa=0.5$.}
    \end{subfigure}
    \hfill
    \begin{subfigure}{0.325\textwidth}
        \includegraphics[scale=0.35]{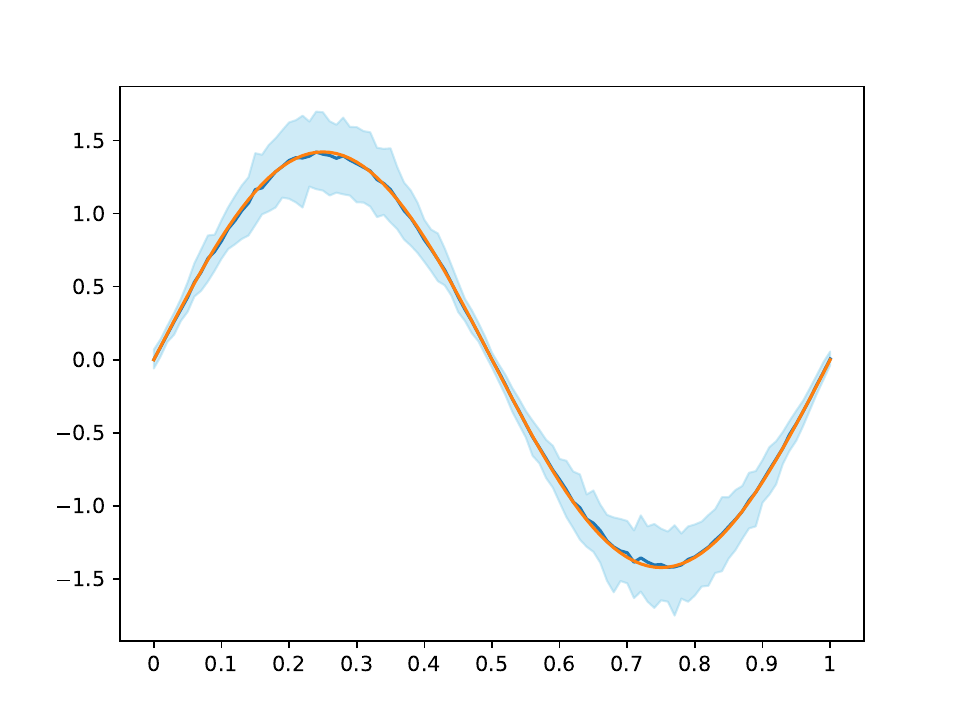}
        \caption{$\gamma=0.5$, $\tau = -0.5$, $\kappa=0.5$.}
    \end{subfigure}
    \hfill
    \begin{subfigure}{0.325\textwidth}
        \includegraphics[scale=0.35]{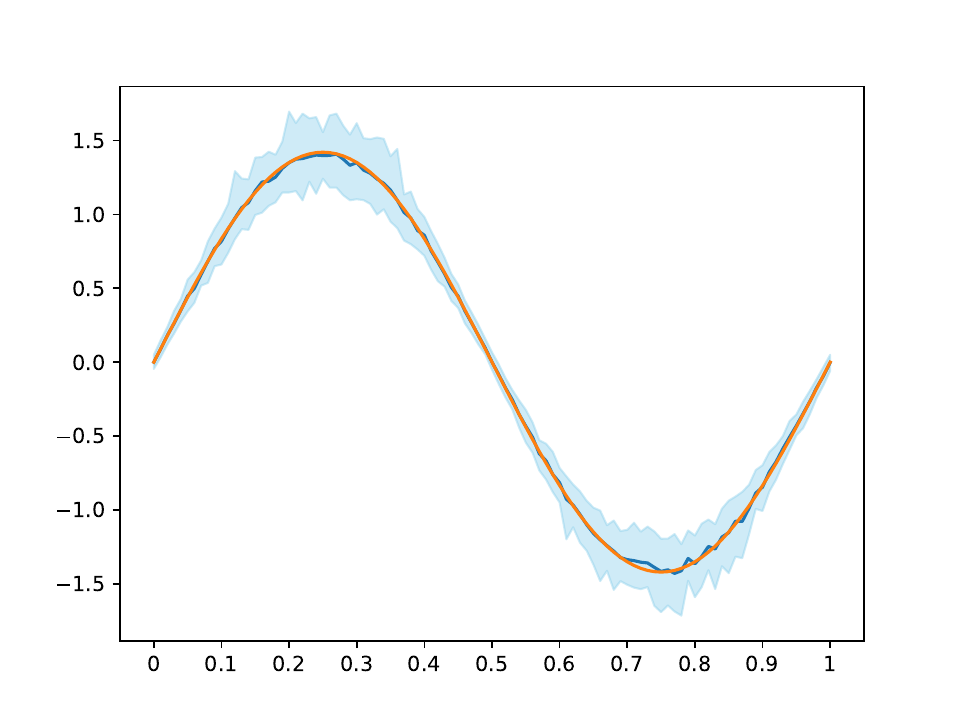}
        \caption{$\gamma=0.5$, $\tau = -0.9$, $\kappa=0.5$.}
    \end{subfigure}
      
\caption{Simulation results on the inverse model with serial dependence of type standard GARCH-response and pathological ARMA-noise, \emph{i.e.,}  $(\omega_{\rm resp},\alpha_{\rm resp},\beta_{\rm resp})=(0.05,0.1,0.85)$ and $(\phi_{\rm noise},\theta_{\rm noise})= (0.99,-0.98)$.}
    \label{fig:st_GARCH_patho_ARMA}
\end{figure}

\begin{figure}[p]
    \centering
    
    \begin{subfigure}{0.325\textwidth}
        \includegraphics[scale=0.35]{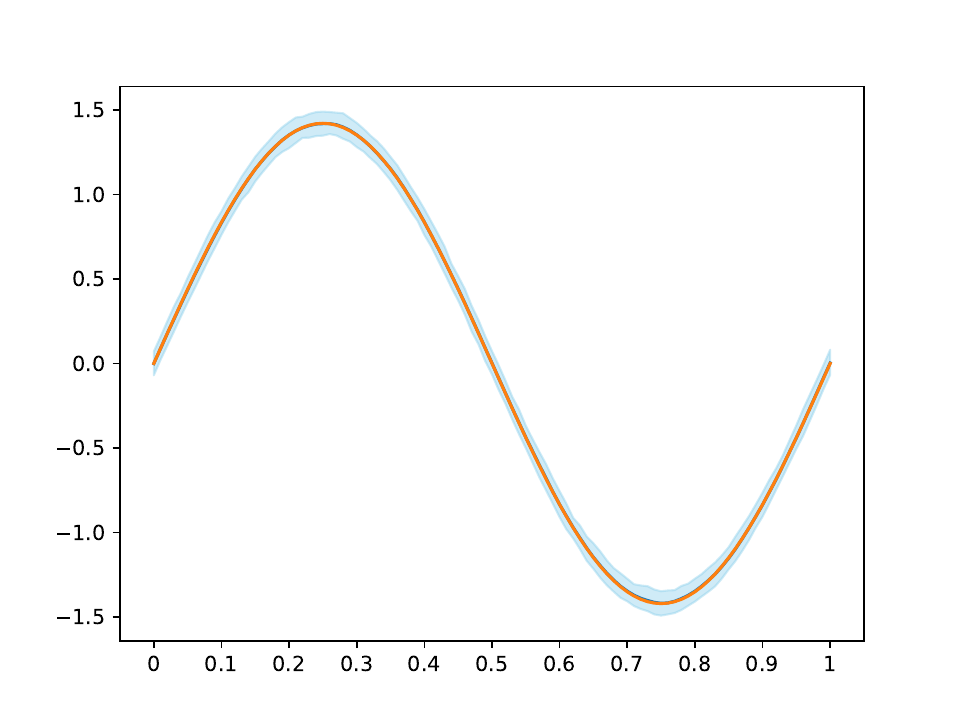}
        \caption{$\gamma=0.1$, $\tau = -0.1$, $\kappa=0.5$.}
    \end{subfigure}
    \hfill
    \begin{subfigure}{0.325\textwidth}
        \includegraphics[scale=0.35]{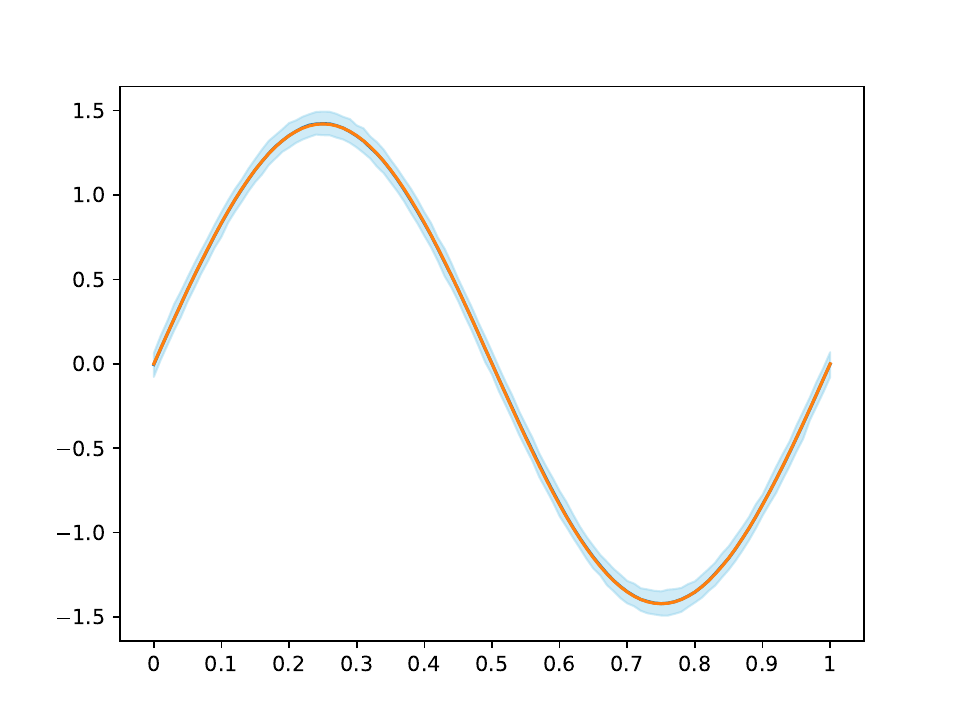}
        \caption{$\gamma=0.1$, $\tau = -0.5$, $\kappa=0.5$.}
    \end{subfigure}
    \hfill
    \begin{subfigure}{0.325\textwidth}
        \includegraphics[scale=0.35]{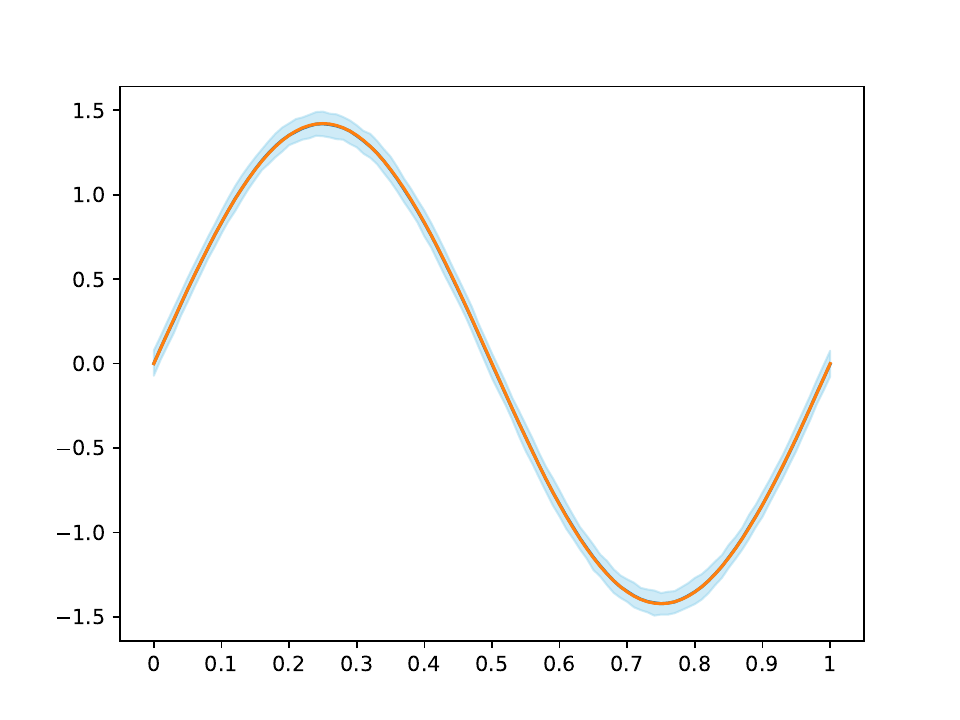}
        \caption{$\gamma=0.1$, $\tau = -0.9$, $\kappa=0.5$.}
    \end{subfigure}

    \medskip 
   \begin{subfigure}{0.325\textwidth}
        \includegraphics[scale=0.35]{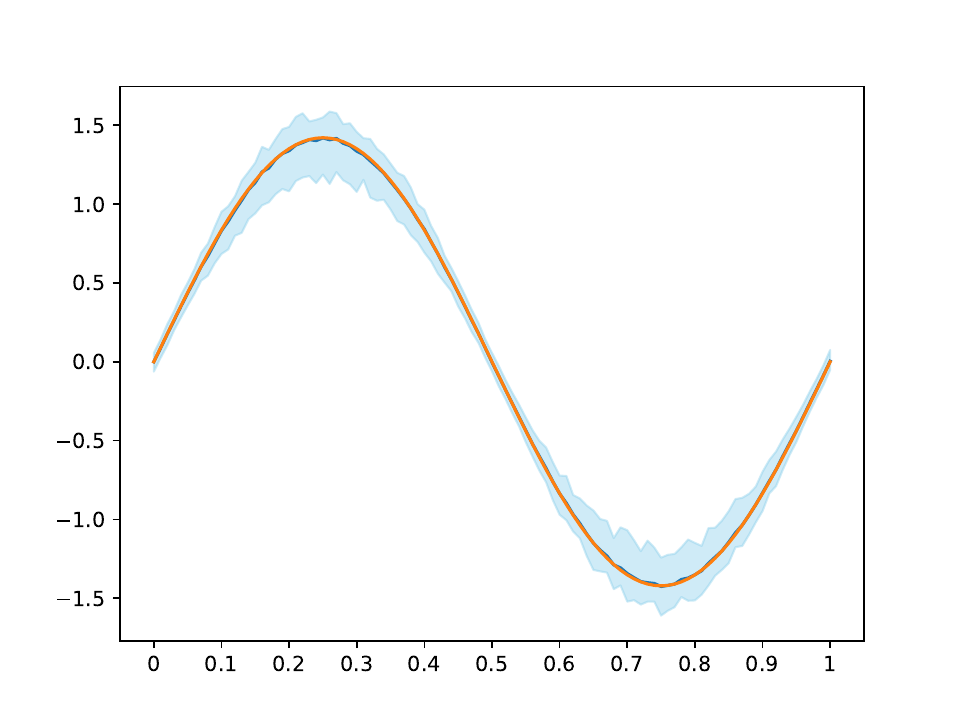}
        \caption{$\gamma=0.4$, $\tau = -0.1$, $\kappa=0.5$.}
    \end{subfigure}
    \hfill
    \begin{subfigure}{0.325\textwidth}
        \includegraphics[scale=0.35]{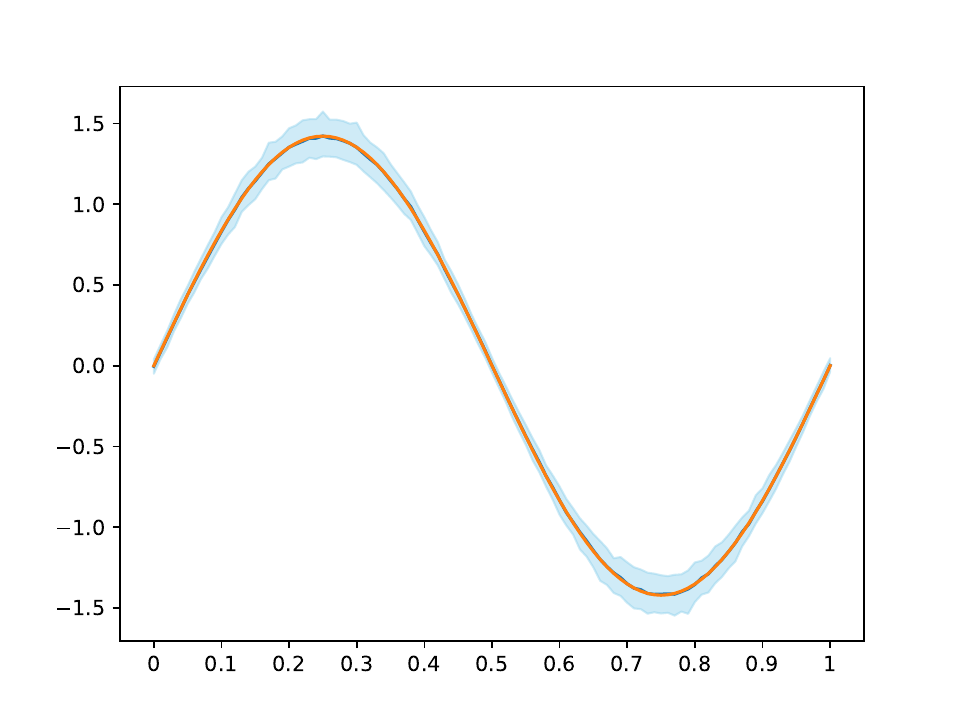}
        \caption{$\gamma=0.4$, $\tau = -0.5$, $\kappa=0.5$.}
    \end{subfigure}
    \hfill
    \begin{subfigure}{0.325\textwidth}
        \includegraphics[scale=0.35]{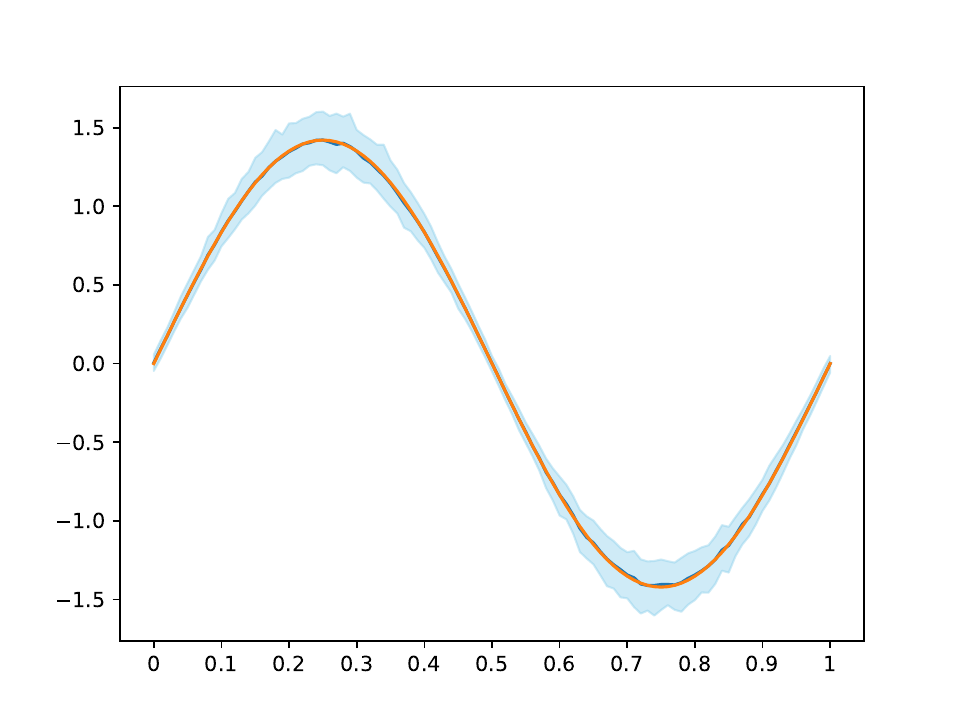}
        \caption{$\gamma=0.4$, $\tau = -0.9$, $\kappa=0.5$.}
    \end{subfigure}

    \medskip 

  \begin{subfigure}{0.325\textwidth}
        \includegraphics[scale=0.35]{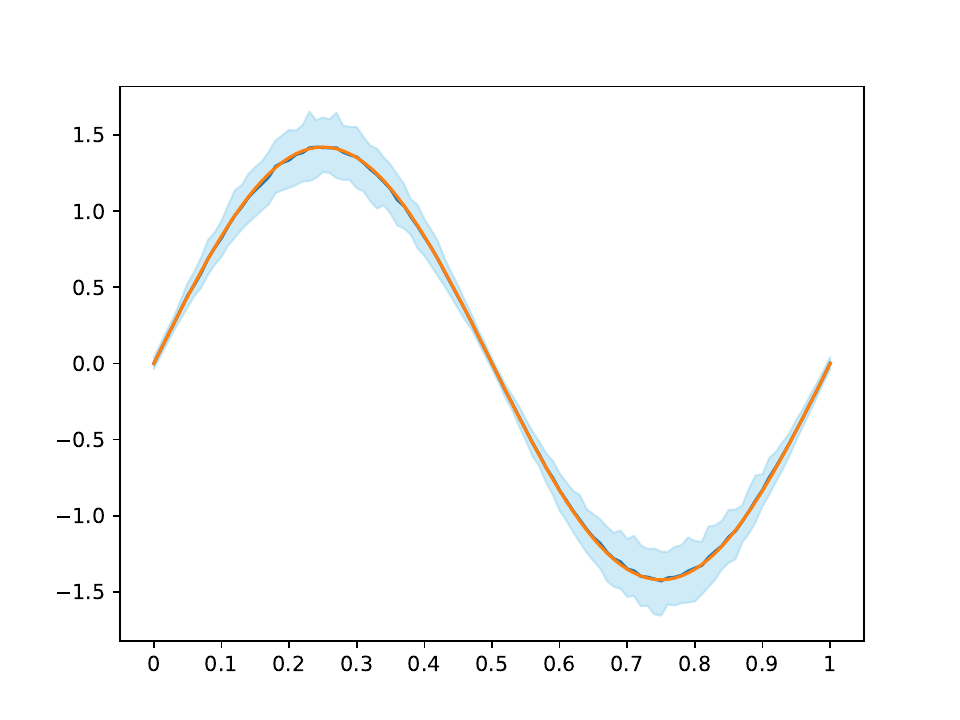}
        \caption{$\gamma=0.5$, $\tau = -0.1$, $\kappa=0.5$.}
    \end{subfigure}
    \hfill
    \begin{subfigure}{0.325\textwidth}
        \includegraphics[scale=0.35]{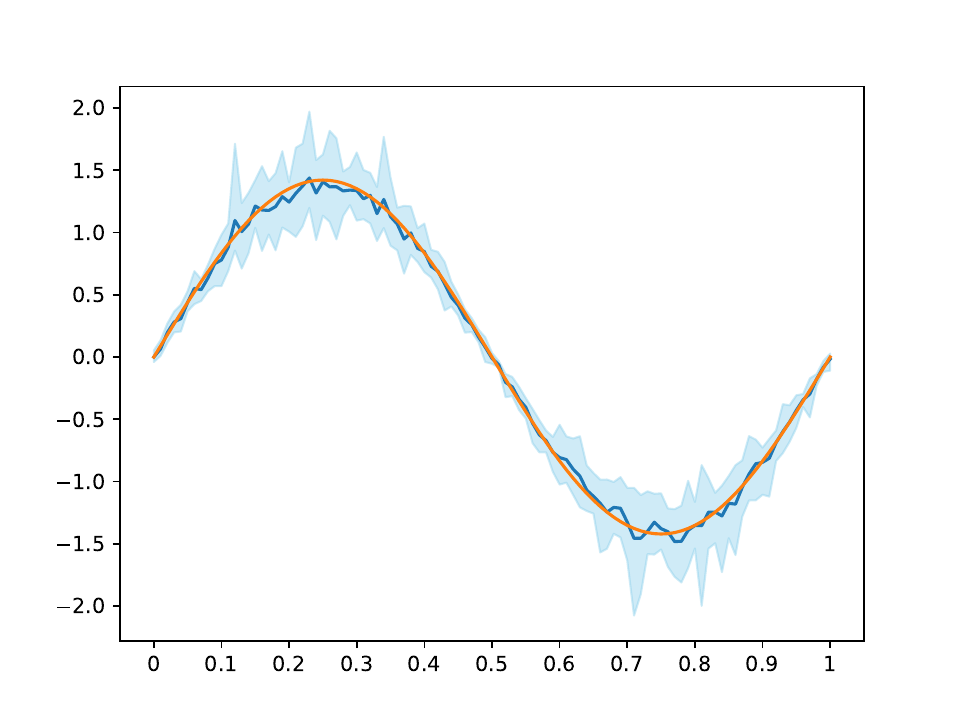}
        \caption{$\gamma=0.5$, $\tau = -0.5$, $\kappa=0.5$.}
    \end{subfigure}
    \hfill
    \begin{subfigure}{0.325\textwidth}
        \includegraphics[scale=0.35]{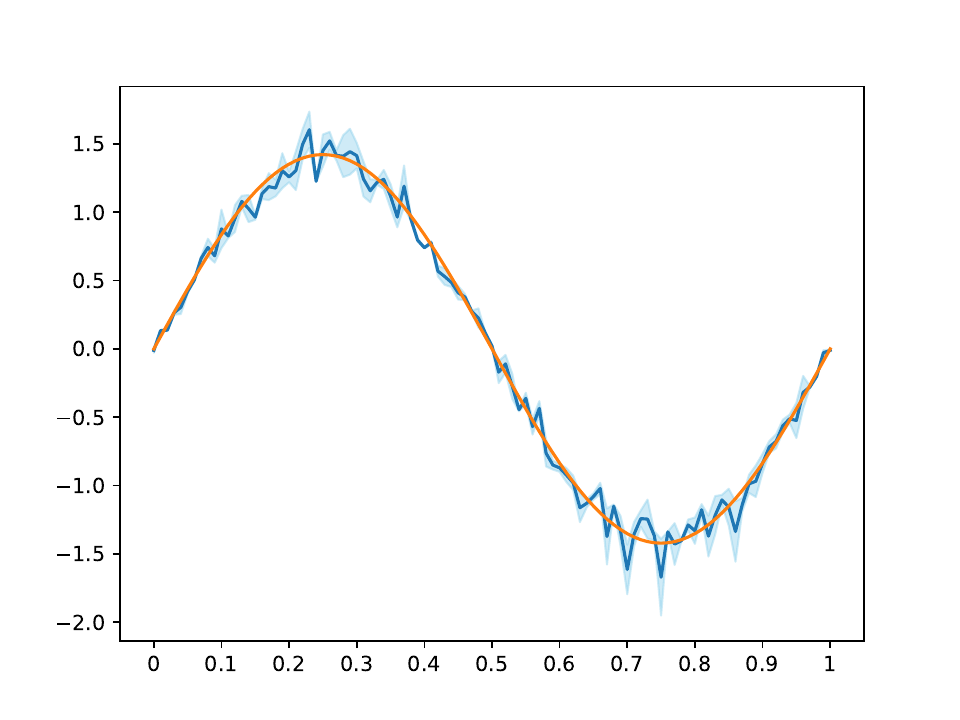}
        \caption{$\gamma=0.5$, $\tau = -0.9$, $\kappa=0.5$.}
    \end{subfigure}
      
\caption{Simulation results on the inverse model with serial dependence of type IGARCH-like response and pathological ARMA-noise, \emph{i.e.,}  $(\omega_{\rm resp},\alpha_{\rm resp},\beta_{\rm resp})=(0.05,0.05,0.94)$ and $(\phi_{\rm noise},\theta_{\rm noise})= (0.99,-0.98)$.}
    \label{fig:patho_GARCH_patho_ARMA}
\end{figure}

\begin{figure}[p]
    \centering
        \begin{subfigure}{0.325\textwidth}
        \includegraphics[scale=0.35]{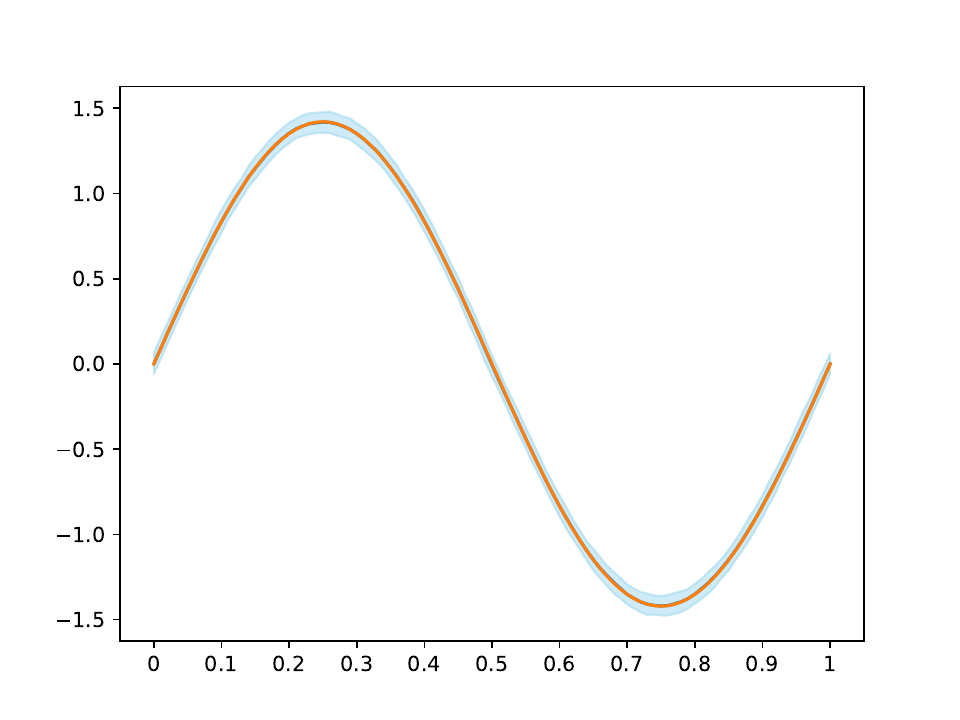}
        \caption{$\gamma=0.1$, $\tau = -0.1$, $\kappa=0.5$.}
    \end{subfigure}
    \hfill
    \begin{subfigure}{0.325\textwidth}
        \includegraphics[scale=0.35]{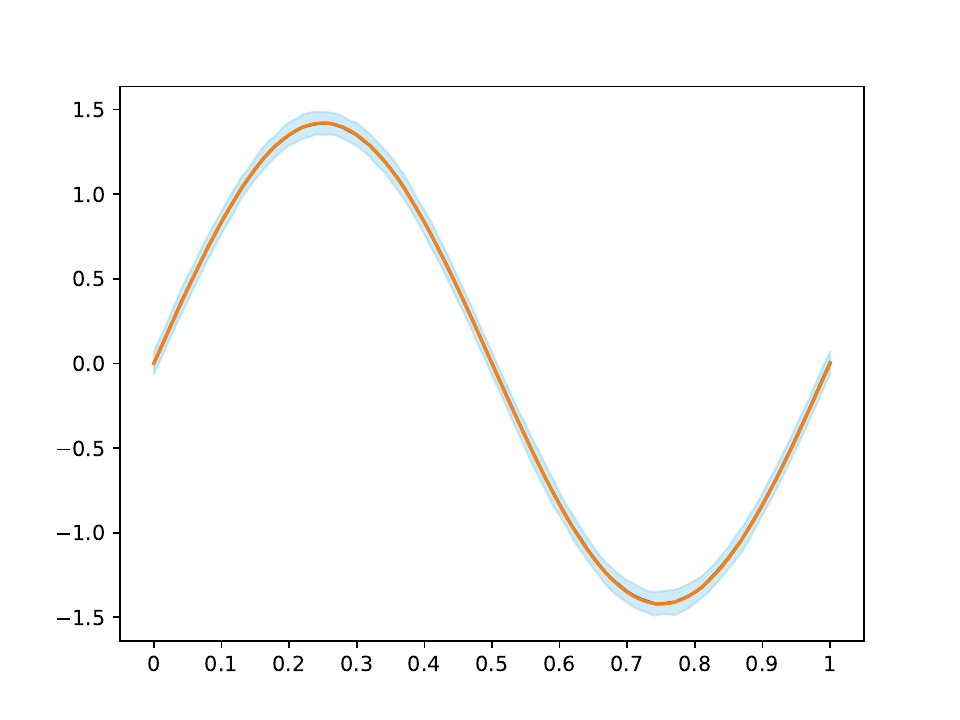}
        \caption{$\gamma=0.1$, $\tau = -0.5$, $\kappa=0.5$.}
    \end{subfigure}
    \hfill
    \begin{subfigure}{0.325\textwidth}
        \includegraphics[scale=0.35]{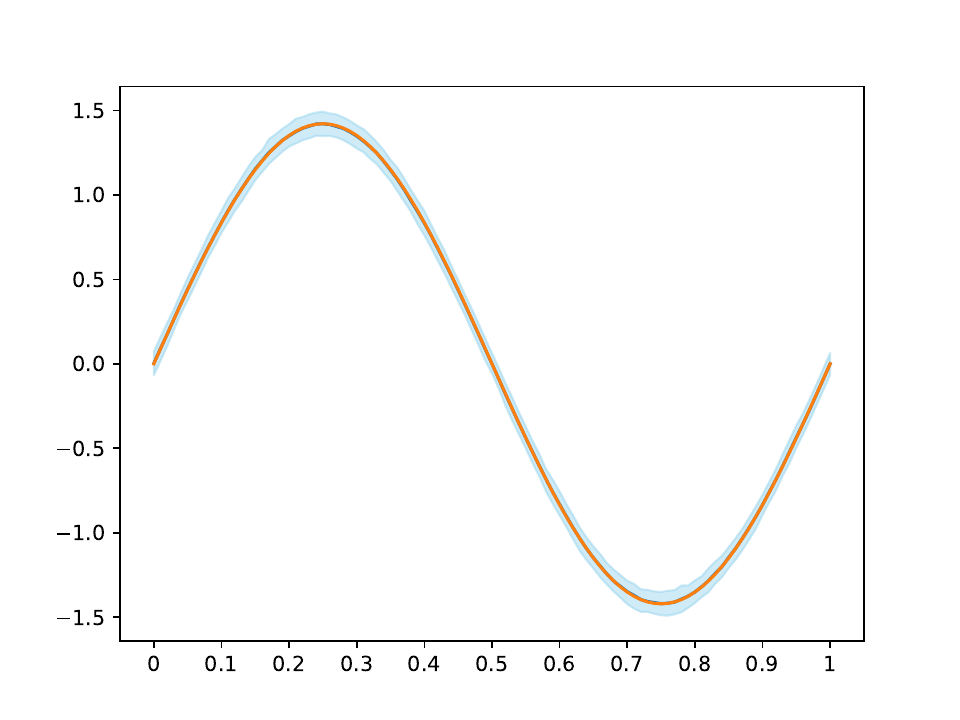}
        \caption{$\gamma=0.1$, $\tau = -0.9$, $\kappa=0.5$.}
    \end{subfigure}
    \medskip 
                \begin{subfigure}{0.325\textwidth}
        \includegraphics[scale=0.35]{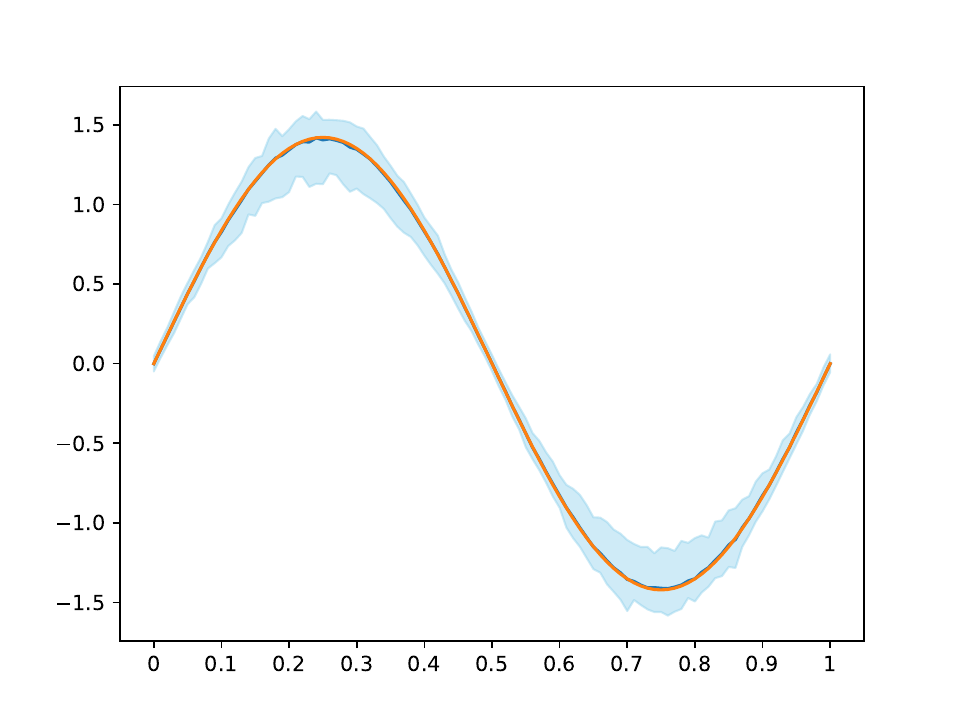}
        \caption{$\gamma=0.5$, $\tau = -0.1$, $\kappa=0.5$.}
    \end{subfigure}
    \hfill
    \begin{subfigure}{0.325\textwidth}
        \includegraphics[scale=0.35]{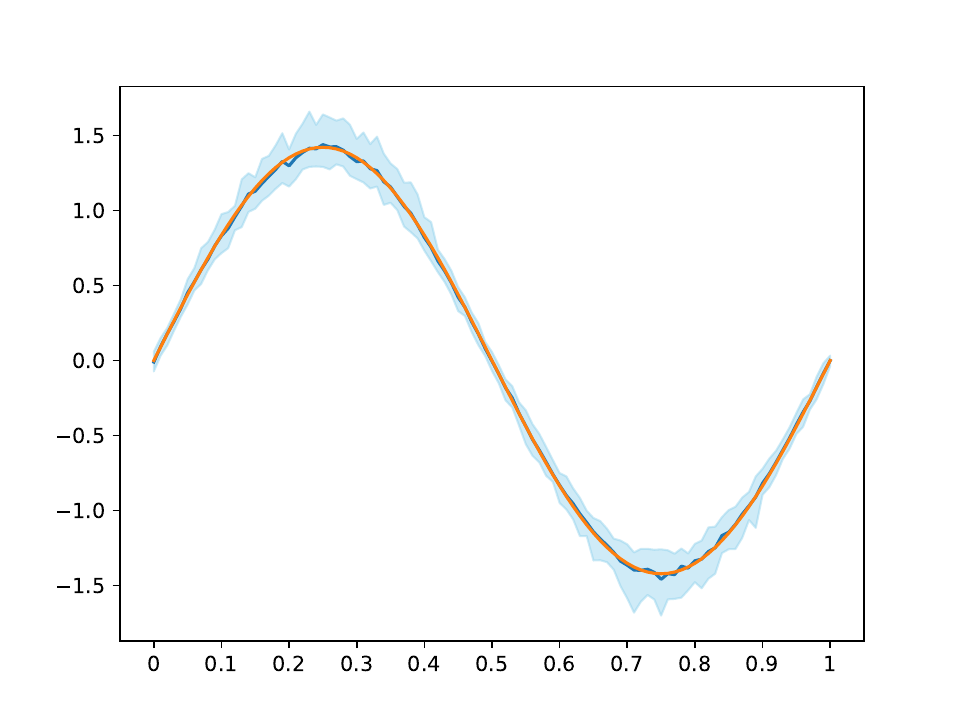}
        \caption{$\gamma=0.5$, $\tau = -0.5$, $\kappa=0.5$.}
    \end{subfigure}
    \hfill
    \begin{subfigure}{0.325\textwidth}
        \includegraphics[scale=0.35]{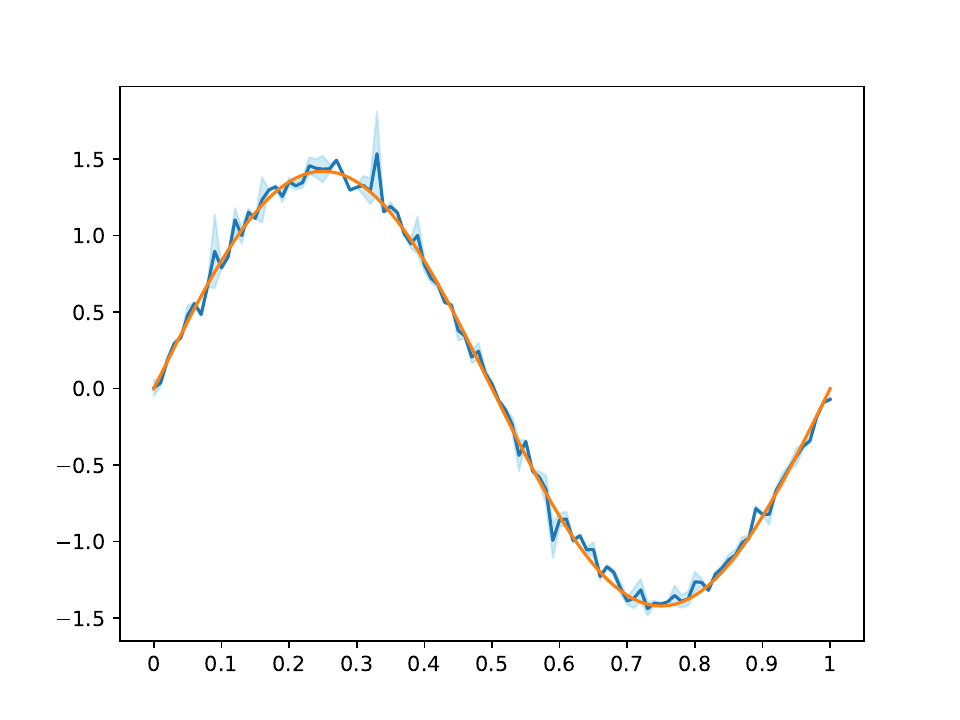}
        \caption{$\gamma=0.5$, $\tau = -0.9$, $\kappa=0.5$.}
    \end{subfigure}
    \medskip 

      
\caption{Simulation results with serial dependence of type ESTAR for the response and standard GARCH for the noise.}
    \label{fig:estar_st_GARCH}
\end{figure}

\begin{figure}[p]
    \centering
        \begin{subfigure}{0.325\textwidth}
        \includegraphics[scale=0.35]{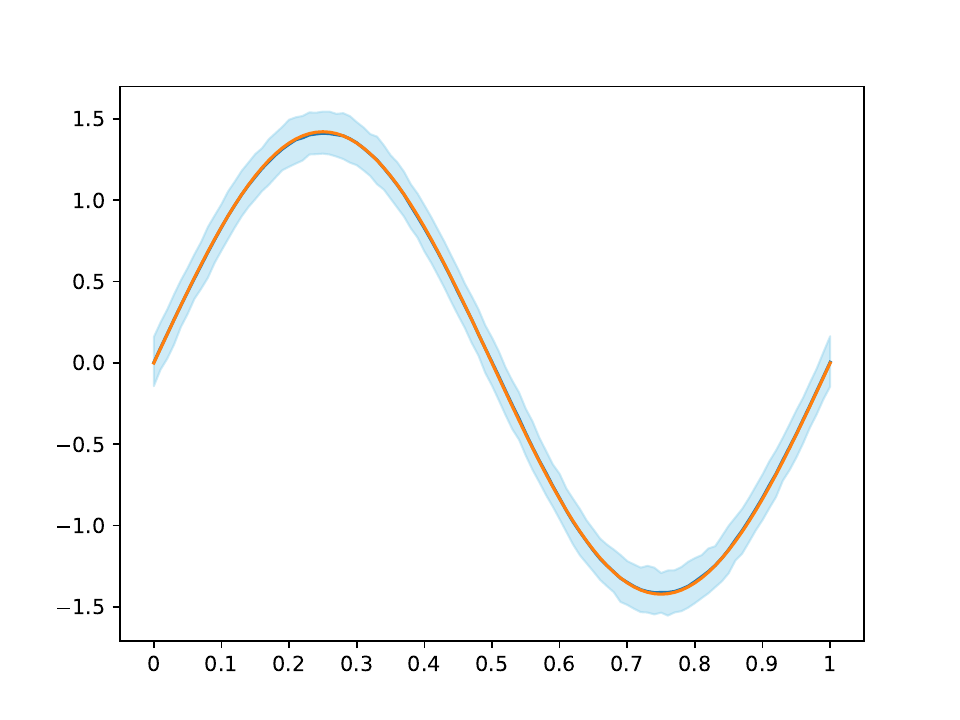}
        \caption{$\gamma=0.1$, $\tau = -0.1$, $\kappa=0.5$.}
    \end{subfigure}
    \hfill
    \begin{subfigure}{0.325\textwidth}
        \includegraphics[scale=0.35]{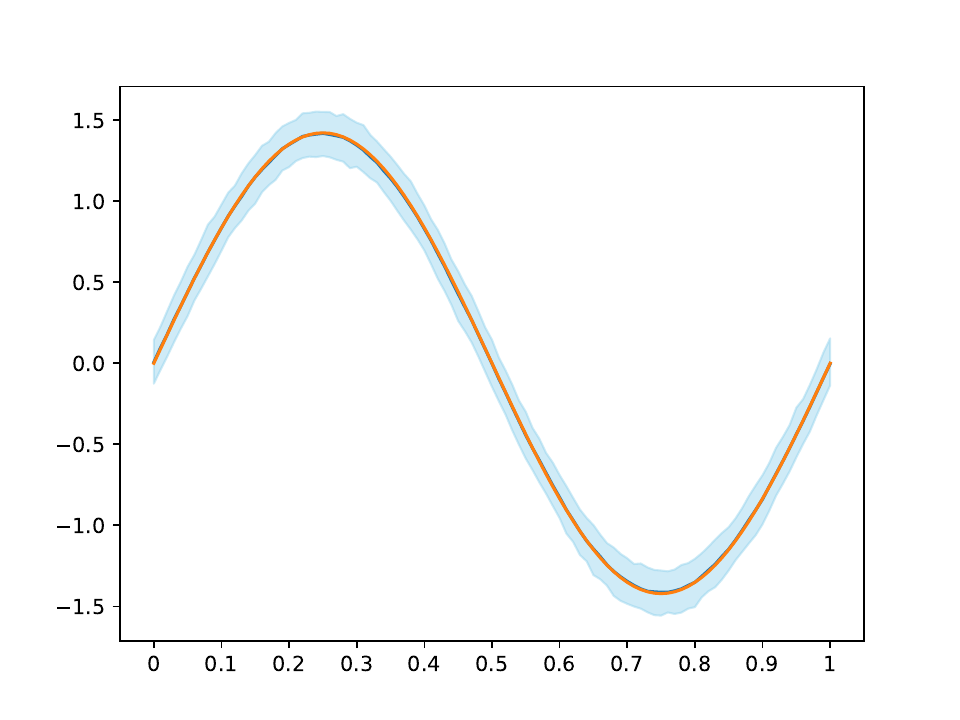}
        \caption{$\gamma=0.1$, $\tau = -0.5$, $\kappa=0.5$.}
    \end{subfigure}
    \hfill
    \begin{subfigure}{0.325\textwidth}
        \includegraphics[scale=0.35]{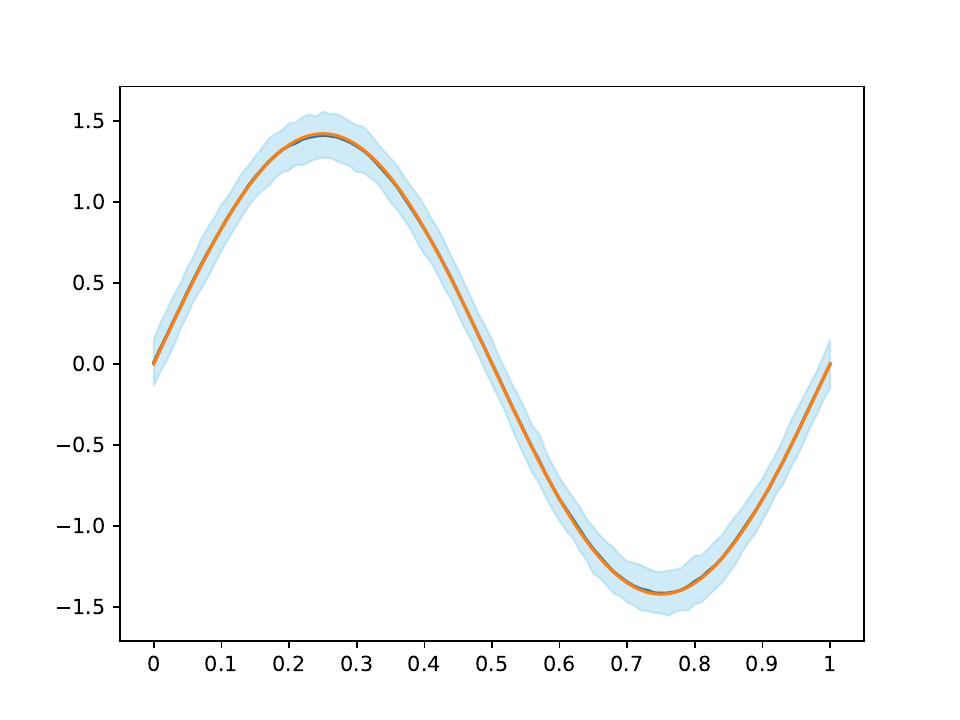}
        \caption{$\gamma=0.1$, $\tau = -0.9$, $\kappa=0.5$.}
    \end{subfigure}
    \medskip 
                \begin{subfigure}{0.325\textwidth}
        \includegraphics[scale=0.35]{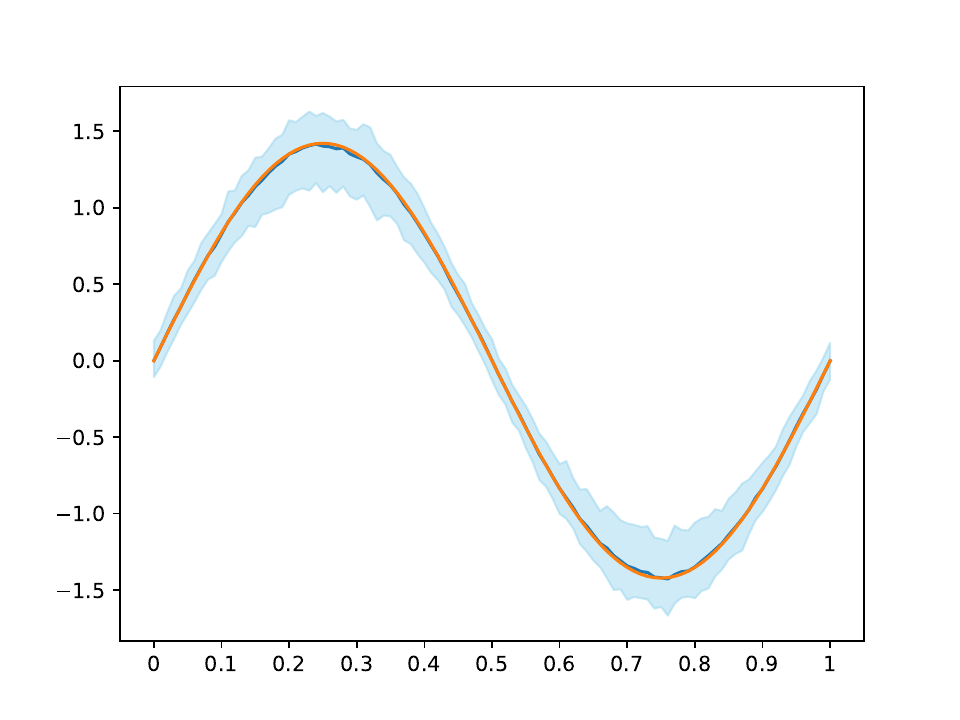}
        \caption{$\gamma=0.5$, $\tau = -0.1$, $\kappa=0.5$.}
    \end{subfigure}
    \hfill
    \begin{subfigure}{0.325\textwidth}
        \includegraphics[scale=0.35]{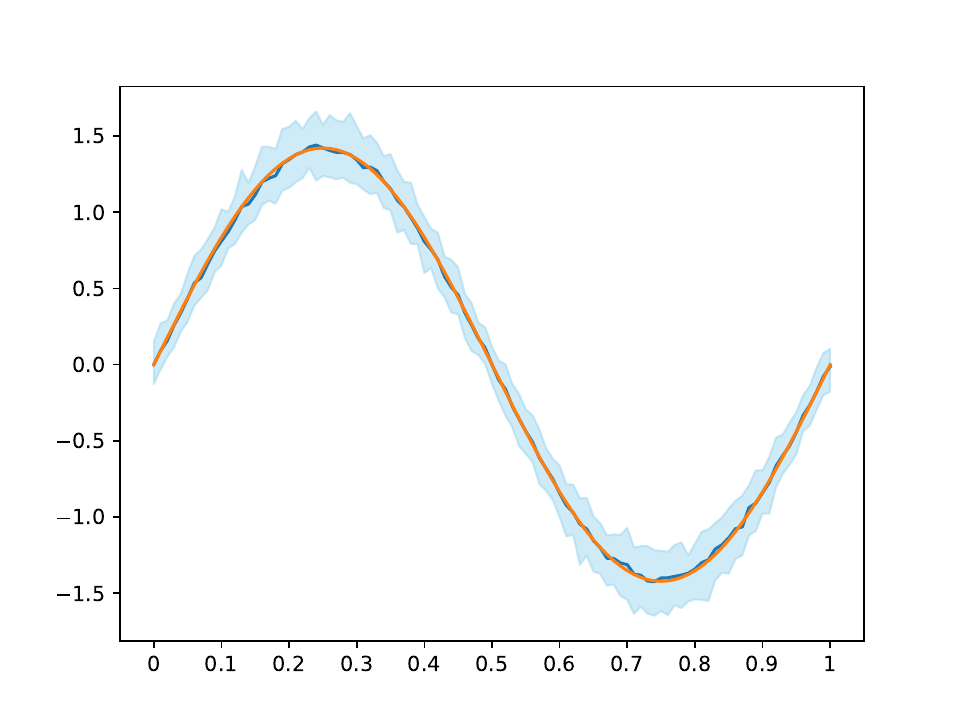}
        \caption{$\gamma=0.5$, $\tau = -0.5$, $\kappa=0.5$.}
    \end{subfigure}
    \hfill
    \begin{subfigure}{0.325\textwidth}
        \includegraphics[scale=0.35]{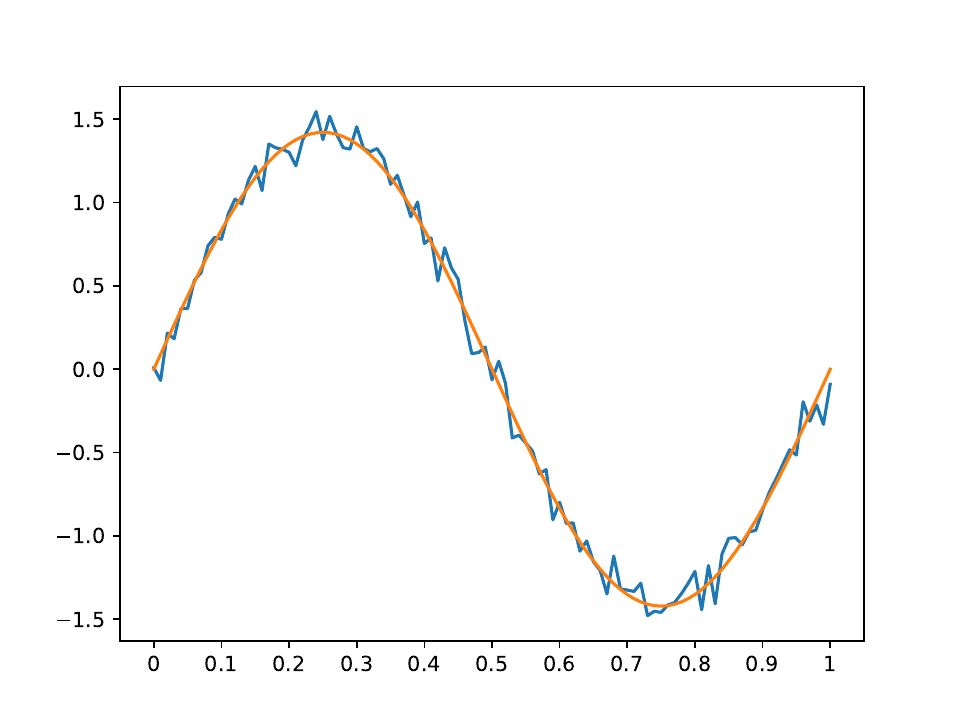}
        \caption{$\gamma=0.5$, $\tau = -0.9$, $\kappa=0.5$.}
    \end{subfigure}

      
\caption{Simulation results with serial dependence of type ESTAR for the response and both cases pathological GARCH for the noise.}
    \label{fig:estar_patho_GARCH}
\end{figure}

\section{Illustration on real data}
\label{sec-reel}

We exemplify our extreme dimension reduction technique on real-world environmental data from the National Oceanic and Atmospheric Administration (NOAA), the primary U.S. agency for climate observation. Our analysis focuses on the state of Texas, a region frequently affected by extreme weather events, using daily data from NOAA’s Global Historical Climatology Network (GHCN) over the five-year period from January 1st, 2020, to December 31st, 2024 ($1827$ days in total). The data are described in more details in Section~\ref{sub-data}. The competitors as well as the performance assessment method are presented in Section~\ref{sub-compet} and the associated results are provided in Section~\ref{sub-results2}.

\subsection{Data selection protocol}
\label{sub-data}
We begin by parsing the complete station inventory provided in \texttt{ghcnd-stations.txt}, which contains metadata (identifier, location, name, etc.) for all GHCN stations worldwide. We retain only stations located within Texas, identified via latitude/longitude filtering or the ``TX'' state code. From this subset, we further restrict attention to stations with names containing keywords suggestive of institutional or professional meteorological infrastructure (e.g., ``airport'', ``afb'', ``intl'', ``weather service'', ``observatory'', etc.). This heuristic aims to prioritize stations likely to have more reliable and continuous records, often associated with airports, military bases, or official weather bureaus.

We extract two variables: daily total precipitation (\texttt{PRCP}, in mm) and daily maximum temperature (\texttt{TMAX}, in degrees Celsius). The dataset distinguishes three cases for any variable on a given date: (i) no entry at all, (ii) a \texttt{NaN} value, or (iii) a valid numerical value. We interpret the absence of any entry (i) as indicating that the station was not operational on that day, and exclude such days from the analysis. This simplification removes any masking mechanism related to such absences, under the virtual assumption that no data implies no potential signal. The \texttt{NaN} values, which correspond to known but missing measurements, are imputed to zero. Admittedly, this may lead to an underestimation of the true missing-ness, as certain days with no recorded value could still reflect relevant information. Addressing this issue would require attributing causes to each missing or empty value, potentially through spatial borrowing from nearby stations, which we acknowledge as a meaningful but computationally intensive extension beyond the scope of this study. We retain only those stations where precipitation data is available (i.e., not \texttt{NaN}) for at least $99\%$ of the reduced time range. Any remaining \texttt{NaN} entries are dropped, along with their corresponding dates. This leads to $Y$ having in total $n=1755$ daily observations. For the maximum temperature, we only consider stations with between $5\%$ and $20\%$ missing values on the same date range as $Y$. From the cleaned dataset, we form all possible triplets of aligned time series $(Y, X^{(1)}, X^{(2)})$, where $Y$ denotes the daily precipitation from one station, and $(X^{(1)}, X^{(2)})$ the maximum temperature from two other stations (which may coincide with $Y$’s location). It yields $27910$ possible triplets of stations.


To ensure statistical soundness in our extreme value analysis, we impose a series of filtering criteria on the triplets $(Y, X^{(1)}, X^{(2)})$, concerning both temporal stationarity and tail behavior beyond data completeness. First, we retain only those triplets where the response variable $Y$ (precipitation) satisfies standard weak stationarity conditions. This is verified via the Augmented Dickey–Fuller (ADF) test (testing for the absence of unit roots) and the Kwiatkowski–Phillips–Schmidt–Shin (KPSS) test (testing for stationarity around a constant mean). Both tests must support the stationarity hypothesis at the 5\% level. We require joint stationarity for the pair $(X^{(1)}, X^{(2)})$, and each individual series must again pass both ADF and KPSS tests. Additionally, we enforce cross-correlation stability between $X^{(1)}$ and $X^{(2)}$, as follows: we compute the rolling correlation between the two series (using a fixed window size, e.g., 100 days), and subject the resulting time series of correlations to an ADF test. This tests the temporal constancy of their dependence structure. Only pairs whose rolling correlations are stationary are considered jointly stationary. In addition to the value processes, we retain only the datasets for which the binary masks $M_1$ and $M_2$ flagging missing-ness are jointly stationary. This is assessed via a chi-squared test for homogeneity applied to joint binary patterns over multiple time windows, with Haldane’s correction (adding $0.5$ to all frequencies to avoid zero counts). The null hypothesis is that the distribution of missing-ness patterns remains stable over time ({\it i.e.} the missing-ness process is stationary). We reject it if the $p$-value is below 5\%; thus, only datasets with $p \geq 0.05$ are considered to exhibit stationary missing-ness and are kept for analysis.

Finally, we check whether the marginal distribution of the response variable exhibits heavy-tailed behavior. A standard tool in extreme value theory is the Hill estimator for estimating the tail-index $\gamma$ of a distribution, based on the $k$ largest order statistics of the sample. A Hill plot displays the estimated tail-index as a function of $k$. A visual plateau in such a plot, that is a region where the estimates remain approximately constant across a range of $k$, is typically interpreted as evidence for a regularly varying tail, with the plateau level serving as an estimate of $\gamma$. To enhance reproducibility, we rather implement an automatic procedure to detect heavy-tailness. This method scans the Hill plot for stable plateaus: contiguous regions of order statistics where the estimated tail-index is approximately flat (low slope), low-variance, and exceeds a minimum threshold ($\gamma > 0.2$). A plateau is declared if these criteria are met within a sufficiently long sliding window, ensuring that transient fluctuations are not misclassified as stable signals. If such a region is detected, we conclude that the data are heavy-tailed and retain the mean Hill estimate over the plateau as our tail-index estimate. 

Overall, the procedure leads to $5794$ admissible triplets $(Y,X^{(1)},X^{(2)})$ of stations.

\subsection{Tail covariance evaluation and competitors}
\label{sub-compet}


To assess the effectiveness of our method, we propose some custom evaluation pipeline based on the conditional tail covariance and statistics over the admissible triplets of stations. 

For a given admissible triplet and a direction $\beta \in \mathbb{R}^2$, we compute the scalar projection $\beta^\top X$ and evaluate its empirical covariance with $Y$, conditionally on $Y$ exceeding a threshold $y$.

A variety of baseline projection competitor strategies is considered: random projections, Extreme PCA (EPCA) computed only on the extreme samples, Extreme Sliced Inverse Regression (ESIR) with a slicing restricted to the top quantile of $Y$, and an extreme variant of Random Forest (ERF) trained to classify extreme events versus non-extreme ones, from which a projection direction is extracted based on feature importance scores. More precisely, fix a threshold $y>0$ and define the exceedances index set as
\[
\mathcal{I}_y \; :=\; \{\, i\in\{1,\dots,n\} : Y_i>y \,\},\qquad N_y:=|\mathcal{I}_{y}|.
\]
Let us denote by $X_y\in\mathbb{R}^{N_y\times 2}$  the submatrix of $(X^{(j)}_i)_{i,j}$ with rows $\{X_i^\top : i\in\mathcal{I}_y\}$ and by $Y_y\in\mathbb{R}^{N_y}$ for the corresponding tail responses. Let $ \mathcal{M}:=\{ \mathrm{EPLS},\mathrm{EPCA},\mathrm{LDA},\mathrm{SIR},\mathrm{ESIR},\mathrm{ERF}\}$, then each method $m\in \mathcal{M}$ produces a projection vector score
\[
Z_{m,y} \; :=\; X_y\,\beta_m(y) \in \mathbb{R}^{N_y},
\]
where the method–specific direction $\beta_m(y)\in\mathbb{R}^p$ is defined below. For each $m\in \mathcal{M}$, we compute the unbiased empirical covariance between the method’s score and the tail responses,
\begin{align}\label{eq:extreme_pearson_cov}
    &\frac{1}{N_y-1}\sum_{i\in\mathcal{I}_y}\left((Z_{m,y})_i-\bar Z_{m,y}\right)\left(Y_i-\bar Y_y\right),
\end{align}
with  $\bar Z_{m,y}=N_y^{-1}\sum_{i\in\mathcal{I}_y} (Z_{m,y})_i$ and $\bar Y_{y}=N_y^{-1}\sum_{i\in\mathcal{I}_y} Y_i$. The four competitors are then computed as follows:

\medskip
\emph{(i) Extreme PCA.}
The PCA is performed on the tail covariates $\{X_i: i\in\mathcal{I}_y\}$, {\it i.e.} $\beta_{\mathrm{EPCA}}(y)$ is the top eigenvector of $\widehat\Sigma_y$ where 
\(
\widehat\Sigma_y=N_y^{-1}\sum_{i\in\mathcal{I}_y} (X_i-\bar X_y)(X_i-\bar X_y)^\top.
\)  
and $\bar X_y=N_y^{-1}\sum_{i\in\mathcal{I}_y} X_i$.
\medskip
\emph{(ii) Extreme LDA.}
Let the \emph{severity} within the tail be $Y_i-y$ for $i\in\mathcal{I}_y$.
Partition $\{Y_i-y: i\in\mathcal{I}_y\}$ into five equal-frequency slices and denote the slice index by $s(i)\in\{1,\dots,5\}$.
For each slice $\ell$, define the slice mean as
\[
\mu_\ell \;=\; \frac{1}{|\{i\in\mathcal{I}_y:s(i)=\ell\}|} \sum_{\substack{i\in\mathcal{I}_y\\ s(i)=\ell}} X_i
\in\mathbb{R}^2,
\]
and compute the tail covariance $\hat\Sigma_y$ as above (on $X_y$). The method forms a contrast vector $\Delta = \mu_K - \frac{1}{K-1}\sum_{k=1}^{K-1}\mu_k$, which captures the shift in $X$ from the less extreme to the most extreme slice. We define $\beta_{\mathrm{LDA}}(y)= \hat\Sigma_y^{+}\Delta $ where $\hat\Sigma_y^{+}$ is the Moore–Penrose pseudo-inverse. 

\medskip
\emph{(iii) Extreme Random Forest.}
A random forest, implemented in \texttt{scikit-learn}, is trained to classify the tail indicator $B_{i,y}:=\mathds{1}_{\{Y_i\geq y\}}$ using all $n$ observations. It yields impurity-based feature importances
\[
\iota=(\iota_1,\iota_2)^\top \in [0,\infty)^2,\qquad  \iota_1+\iota_2= 1,
\] 
which quantify the relative contribution of each predictor to identifying the extreme region ${Y > y}$. Using it as a linear proxy of the nonlinear forest, we define the direction as $\beta_{\mathrm{ERF}}(y)=\iota/\|\iota\|$.

\medskip
\emph{(iv) Random projections.} This naive method consists in sampling
$500$ directions uniformly on the unit sphere and defining $\beta_{\mathrm{R}}$ as their empirical mean. 

\medskip
\emph{(v) Sliced Inverse Regression (SIR).}
The range of the response $Y$ is divided into $H=10$ disjoint slices $\mathcal{S}_1,\dots,\mathcal{S}_H$ (of approximately equal size), based on quantiles of $Y$. Then, for each $h\in \{1,\ldots,H\}$, one computes the slice mean $ m_h = \frac{1}{N_h}\sum_{i: Y_i \in \mathcal{S}_h}  X_i$ where $N_h = \#\{i: Y_i \in \mathcal{S}_h\}$ is the number of response elements in $\mathcal{S}_h$, and also the between-slice covariance matrix
  \[
  \hat\Sigma_B = \sum_{h=1}^H \frac{N_h}{n} \,(m_h - \bar{X})(m_h - \bar{X})^\top,
  \]
  where $\bar{X}=n^{-1}\sum_{i=1}^n X_i$ is the mean of the standardized covariates. 
 The SIR direction $\beta_{\mathrm{SIR}}$ is obtained as the leading generalized eigenvector of
\(
\hat\Sigma_{XX}^{-1}\hat\Sigma_B,
\)
where the sample covariance of $X_i$ is denoted by $\hat{\Sigma}_{XX}=\tfrac1n\sum_{i=1}^n (X_i-\bar{X})(X_i-\bar{X})^\top$.

\medskip
\emph{(vi) Extreme Sliced Inverse Regression (Extreme SIR).}
To obtain a direction $\beta_{\mathrm{ESIR}}$, the above procedure is applied only to the exceedances/tail sample $\{(Y_i,X_i): Y_i > y\}$.

\subsection{Results}
\label{sub-results2}

A large-scale analysis is performed over the admissible triplets of stations and across multiple thresholds. The threshold is chosen as $y=q_Y(\alpha)$ where  $\alpha\in\{0.90,0.91,\ldots,0.99\}$. The tail covariances~\eqref{eq:extreme_pearson_cov} are computed for all methods and ranked from one to six. The evaluation metric then consists in the mean rank (in $\{1,\dots,6\}$, where $1$ is best) achieved by each method over $5794$ station triplets. The results are summarized in Figure~\ref{fig:mean_rank} which represents the evolution of the mean ranks for each method across the different values of the percentile $\alpha$.
The mean curves are fairly stable as the percentile increases. Extreme PLS shows consistently superior performance across all thresholds, maintaining the first rank. Next, Extreme PCA and Extreme RF perform similarly, with both methods following the same trend across percentiles. In contrast, Standard SIR, Extreme SIR and Extreme LDA under-perform over most of the threshold range.  In the most extreme region, all five competitors tend to converge toward similar performance.

As an illustration, we propose to focus on a single triplet $(Y,X^{(1)},X^{(2)})$. The threshold considered is the order statistic $Y_{n-k+1,n}$.  The tail covariance~\eqref{eq:extreme_pearson_cov} is plotted in Figure~\ref{fig:tail_cov_vs_threshold_value_p1_27178} as a function of the threshold index $1\le k\le n$, for all methods. We also include the random baseline by representing the full range for $500$ projections, which displays the range of variation due to randomness. The performance of EPLS is consistent with Figure~\ref{fig:mean_rank}: The EPLS curve closely follows the 'highest random path' in the sense of highest tail covariance among all possibilities. Again, Extreme PCA performs similarly well as Extreme PLS, while Extreme RF does worse and both versions of SIR compare to the average of the random strategy.

The estimates of an extreme conditional quantile (of level $\alpha=0.99$) associated with the three best performing dimension reduction methods (namely, EPLS, Extreme PCA and Extreme RF) are compared in Figure~\ref{fig:quantile_comparison}. 
In all three cases, the explanatory variable is first projected on the estimated dimension reduction axis $\hat\beta$, and second, a kernel based estimator of the conditional quantile is computed on the obtained data set $\{(\hat\beta^\top X_i, Y_i),\, i=1,\dots,n\}$. 
The dimension reduction step has two benefits: it allows us to mitigate the curse of dimensionality in the nonparametric estimator and it moreover yields an easy visualization of the estimated in a two-dimensional space.
It appears that EPLS, Extreme PCA and Extreme RF yield similar projected data and thus estimated conditional curves sharing similar shapes.

\begin{figure}[p]
\centering
\includegraphics[width=0.8\textwidth]{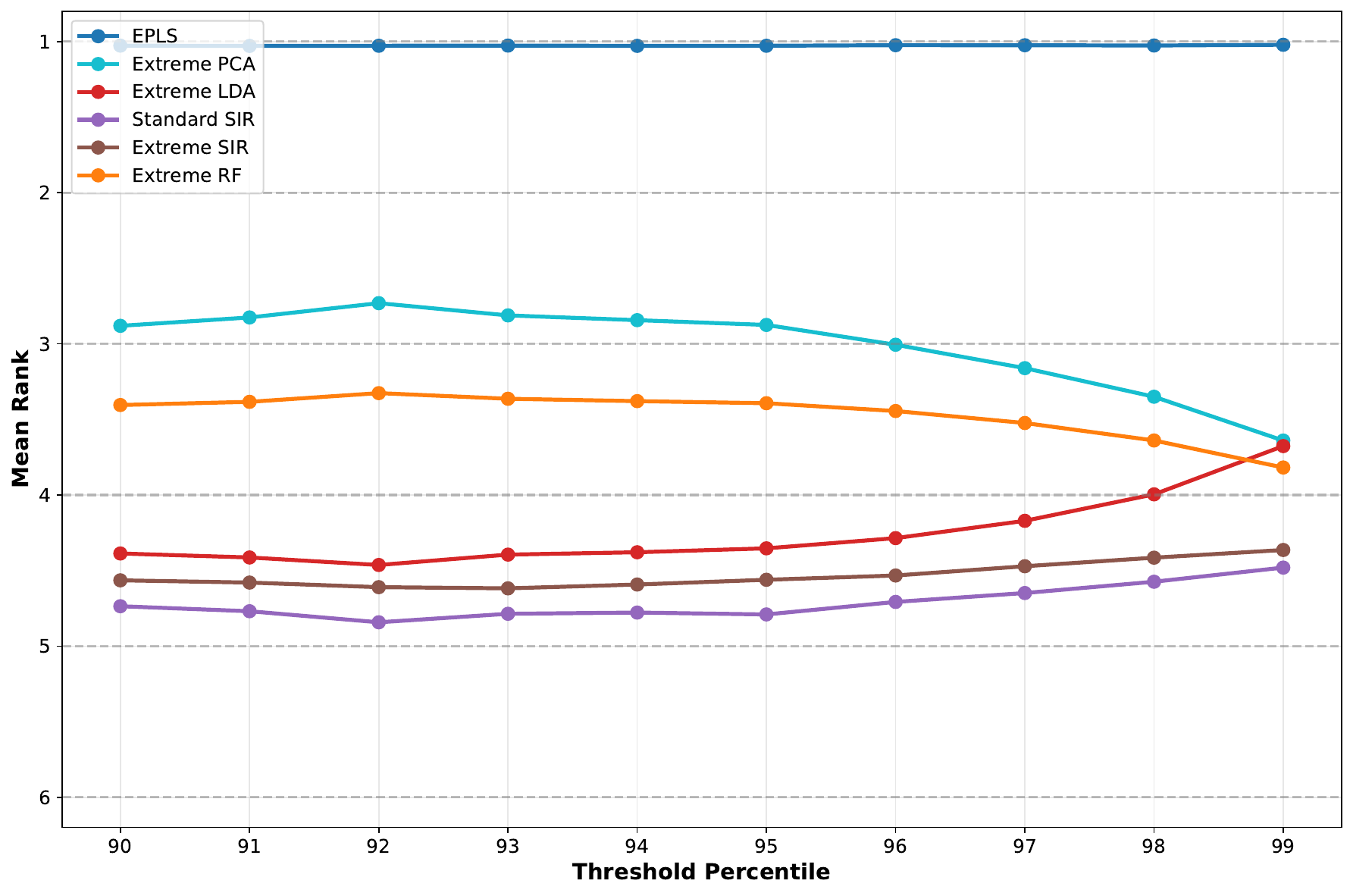}
\caption{Evolution of mean rank performance across threshold percentiles for the different methods on real data. Results are based on evaluation across 5,794 admissible weather station triplets. Lower values indicate better performance (closer to rank one). The random projection results are not displayed here, the mean rank being constant equal to six.}
\label{fig:mean_rank}
\end{figure}

\begin{figure}[p]
\centering \includegraphics[width=0.8\textwidth]{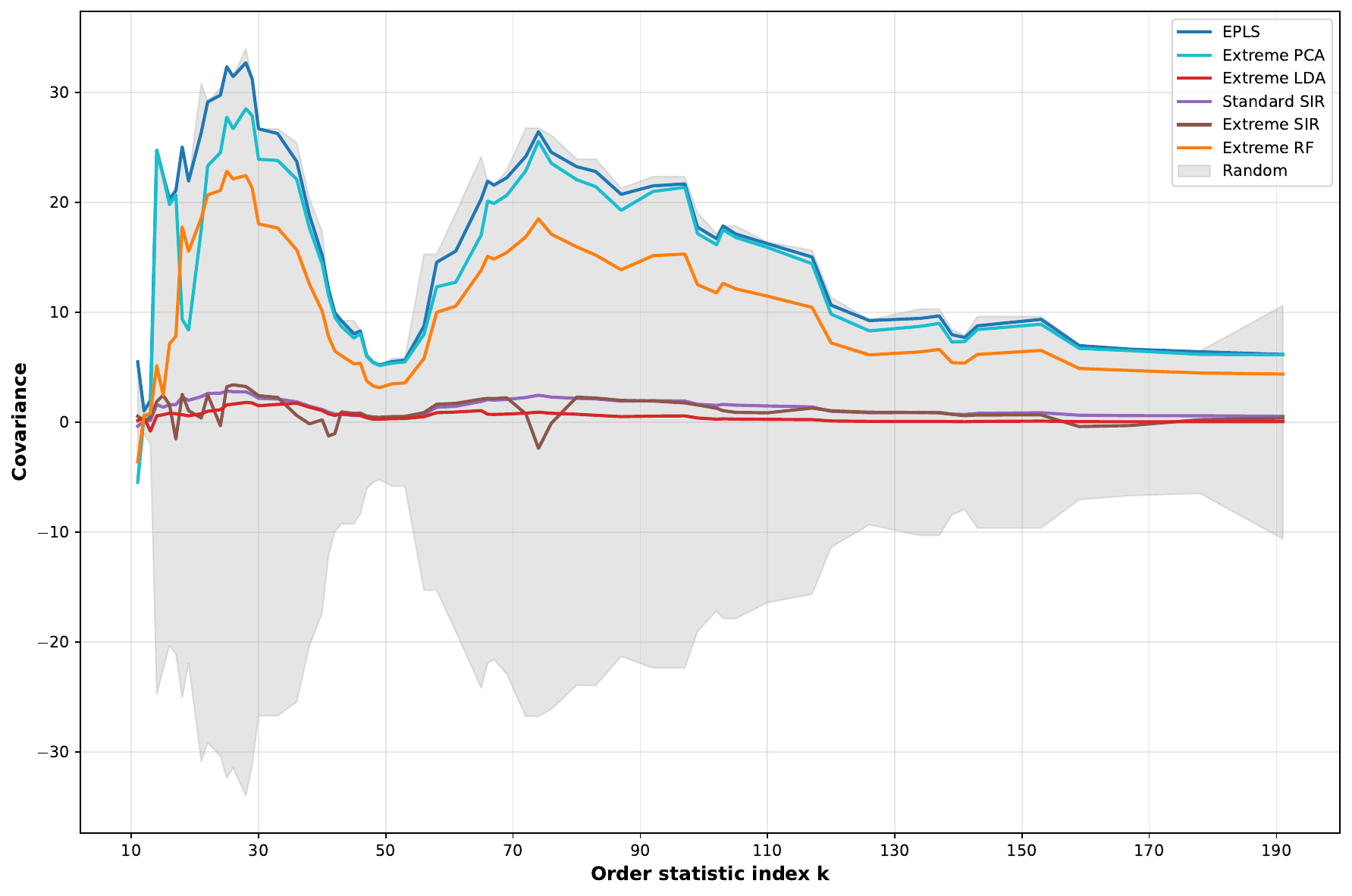} \caption{Tail covariance achieved by EPLS (blue) on the triplet $(\texttt{USC00416892}, \texttt{US1TXMRR003}, \texttt{USC00410058})$ of weather stations compared to the competitors and a random projection baseline. The $x$-axis is the index $k$ associated to the order statistic $Y_{n-k+1,n}$ used as threshold. The grey shaded region indicates the full variability range of the random baseline.}
\label{fig:tail_cov_vs_threshold_value_p1_27178}
\end{figure}

\begin{figure}[htbp]
\centering
\includegraphics[width=0.6\textwidth]{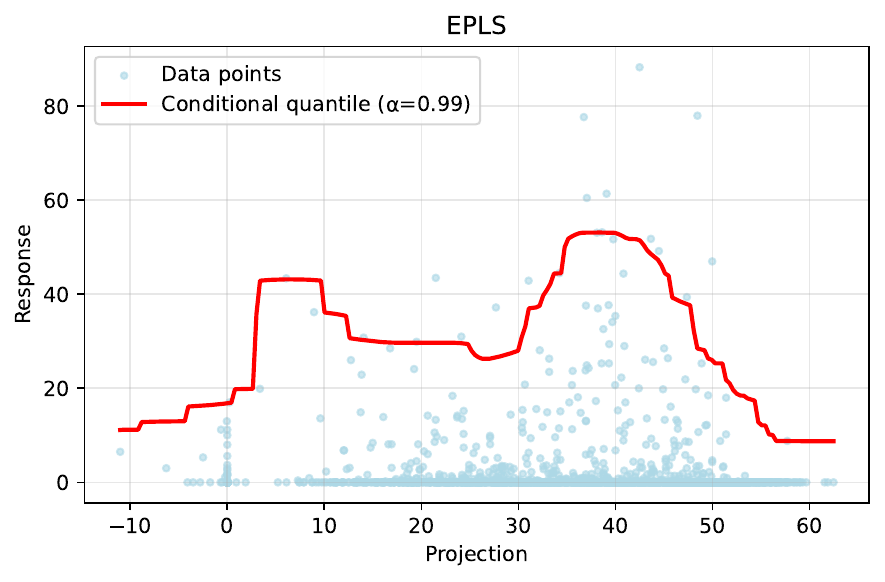}
\includegraphics[width=0.6\textwidth]{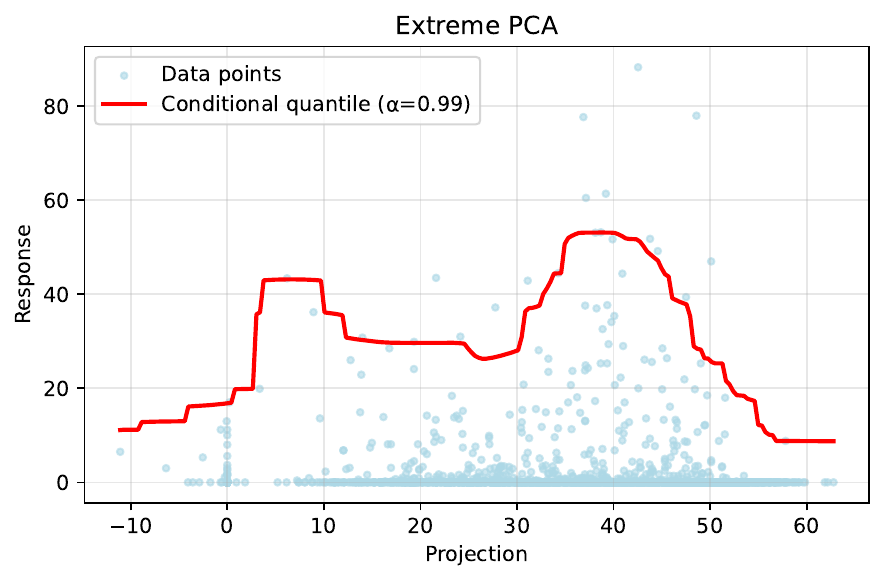}
\includegraphics[width=0.6\textwidth]{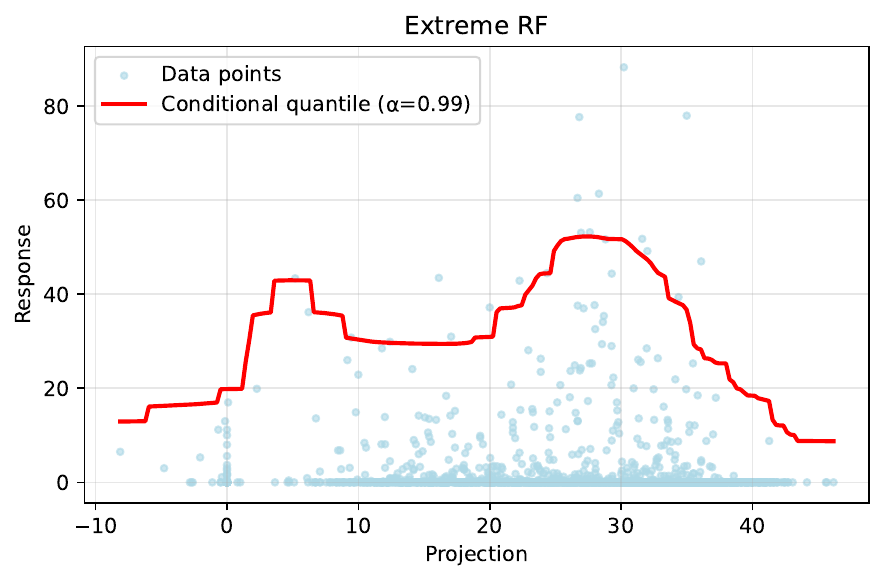}
\caption{Comparison on the triplet $(\texttt{USC00415766}, \texttt{US1TXTN0058}, \texttt{US1TXWA0042})$ of conditional extreme quantile ($\alpha=0.99$) estimation across {three} projection methods. Each panel displays the projected point cloud $(\hat\beta^\top X_i, Y_i)$ with the corresponding estimated conditional quantile curve. 
}
\label{fig:quantile_comparison}
\end{figure}

\section{Appendix: Proofs}

In the sequel we adopt, for any real sequences $A_n$ and $B_n$ such that $B_n\ne 0$ for $n$ large enough, the notation $A_n\lesssim B_n$ if and only if there exists $c\in (0,+\infty)$ independent of $n$ such that $A_n\le c B_n$ for $n$ large enough.
Section~\ref{sec-annexe1} provides three lemmas that will reveal useful to prove the main results in Section~\ref{sec-annexe-main}.

\subsection{Preliminary results}
\label{sec-annexe1}


\noindent The first lemma is an adaptation of~\cite[Proposition~1.5.10]{Bing1989}
to the missing data setting. It provides an asymptotic equivalent of conditional expectations above a large threshold.
\begin{Lem}
\label{denom1}
Let $\rho\in RV_{\mu}$ with $\mu\in \mathbb{R}$ and let $Y$ be a random variable satisfying \Au. 
\begin{itemize}
    \item [(i)]
If {\color{black} $\gamma \mu <1$}, then, as $y\to\infty$,
$$
\mathbb{E}\left[\rho(Y) \mid Y\geq y\right] = \frac{m_{\rho(Y)}(y)}{\bar{F}(y)}\sim \frac{1}{1-\gamma \mu}\rho(y).
$$
\item [(ii)] If \Aq~holds, with {\color{black} $\gamma (\mu+\tau)<1$},
then, as $y\to\infty$, for all $(j\neq\ell)\in\{1,\dots,p\}^2$,
\begin{align*}
\mathbb{E}\left[\Lambda^{(j)}\rho(Y) \ind \right] & \sim  \frac{c_j}{1-\gamma (\mu+\tau)}\rho(y) \lambda(y) \bar F(y).
\end{align*}
\item [(iii)]  If \Aq~holds, with {\color{black} $\gamma (\mu+2\tau)<1$},
then, as $y\to\infty$, for all $(j\neq\ell)\in\{1,\dots,p\}^2$,
\begin{align*}
\mathbb{E}\left[\Lambda^{(j)}\Lambda^{(\ell)}\rho(Y)\ind\right] & \sim  \frac{c_jc_\ell}{1-\gamma (\mu+2\tau)}\rho(y) \lambda^2(y) \bar F(y).
\end{align*}
\item[(iv)] If \At~and \Aq~hold with \textcolor{black}{$\gamma(\frac{q}{q-1}\mu+\tau)<1$}, then, as $y\to\infty$, for all $(j\neq\ell)\in\{1,\dots,p\}^2$, 
\begin{align*}
\mathbb{E}\left[\Lambda^{(j)}\varepsilon^{(j)}\rho(Y) \ind \right] &=O\left(\rho(y) (\lambda\bar F)^{1-1/q}(y)\right).
\end{align*}
\item[(v)] If \At~and \Aq~hold with \textcolor{black}{$\gamma(\frac{q}{q-2-t}\mu+\tau)<1$} for $t\ge 0$, then, as $y\to\infty$, for all $(j\neq\ell)\in\{1,\dots,p\}^2$, 
\begin{align*}
\mathbb{E}\left[\Lambda^{(j)} (\varepsilon^{(j)})^{2+t}\rho(Y) \ind\right] &=O\left(\rho(y) (\lambda\bar F)^{1-(2+t)/q}(y)\right).
\end{align*}
\end{itemize}
\end{Lem}
\begin{proof}[\bf{\textup{Proof}}] (i) is established in \cite[Lemma~2]{Bousebata2023}. \\
(ii) 
Let us consider
\begin{align*}
   \mathbb{E}\left[\Lambda^{(j)}\rho(Y) \ind \right] &= 
 \mathbb{E}\left\{   \mathbb{E}\left[\Lambda^{(j)}\rho(Y) \ind  \mid Y \right] \right \} = \mathbb{E}\left\{ \rho(Y) \ind \mathbb{E}\left[\Lambda^{(j)}   \mid Y \right] \right \}\\
 &=  \mathbb{E}\left\{ \rho(Y) {\lambda_j}(Y) \ind  \right \}=
 c_j \mathbb{E}\left\{ \rho(Y) \lambda(Y) \mid Y\geq y  \right \}\bar F(y).
\end{align*}
Since $\rho(\cdot) \lambda(\cdot) \in RV_{\mu+\tau}$ and {\color{black} $\gamma(\mu+\tau) <1$}, 
part~(i) of the Lemma proves the result. \\
(iii) The proof is similar after noting that 
$\mathbb{E}\left[\Lambda^{(j)}\Lambda^{(\ell)}  \mid Y \right]={\lambda_j}(Y){\lambda_\ell}(Y)$ for $j\neq \ell$.
\\ (iv) 
By Hölder's, for any $q'\ge 1$ such that $1/q+1/q'=1$, one has
$$
\mathbb{E}\left[\Lambda^{(j)}\varepsilon^{(j)}\rho(Y) \ind \right]\le  \mathbb{E}(|\varepsilon^{(j)}|^{q})^{1/{q}} \mathbb{E}( \rho(Y) ^{q'}{\Lambda^{(j)}}\ind)^{1/{q'}}. 
$$
In view of~(ii) together with \textcolor{black}{$\gamma( \frac{q}{q-1}\mu+\tau)<1$} and \At, the result  follows. \\
(v) Again, Hölder inequality implies that for any $q'\ge 1$ such that $(2+t)/q+1/q'=1$,
\begin{eqnarray}
\mathbb{E}\left[\Lambda^{(j)}(\varepsilon^{(j)})^2\rho(Y) \indn \right]\le  \mathbb{E}(|\varepsilon^{(j)}|^{q})^{2/{q}} \mathbb{E}(| \rho(Y) |^{q'}{\Lambda^{(j)}}  \indn)^{1/{q'}}. \nonumber 
\end{eqnarray}
In view of~(ii) together with \textcolor{black}{$\gamma( \frac{q}{q-2-t}\mu+\tau)<1$} and \At, the result  follows. 
\end{proof}

\noindent The next lemma makes use of Lemma~\ref{denom1} to derive asymptotic equivalents of conditional covariances above a large threshold in the missing data framework.

\begin{Lem}
\label{lem-cov}
Suppose $\rho_1\in RV_{\mu_1}$, $\rho_2 \in RV_{\mu_2}$ with $\mu_1,\mu_2\in{\mathbb R}$ and let $Y$ be a random variable satisfying \Au. Assume that \Aq~holds.
\begin{itemize}
    \item [(i)] Assume {\color{black} $\gamma(\mu_1+\mu_2+2\tau)<1$, $\gamma (\mu_s+\tau) <1$ for $s\in\{1,2\}$} holds.
Then, as $y\to\infty$, for all $(j\neq\ell)\in\{1,\dots,p\}^2$,
\begin{align*}
\cov\left[\Lambda^{(j)}\rho_1(Y)\ind,\Lambda^{(\ell)}\rho_2(Y)\ind\right] & \sim  \frac{c_j c_\ell}{1-\gamma (\mu_1+\mu_2+2\tau)}\rho_1(y)\rho_2(y)\lambda^2(y) \bar F(y).
\end{align*}
  \item [(ii)] Assume {\color{black} $\gamma(\mu_1+\mu_2+\tau)<1$, $\gamma (\mu_s+\tau) <1$ for $s\in\{1,2\}$} holds.
Then, as $y\to\infty$, for all $(j\neq\ell)\in\{1,\dots,p\}^2$,
\begin{align*}
\cov\left[\Lambda^{(j)}\rho_1(Y)\ind,\Lambda^{(j)}\rho_2(Y)\ind\right] & \sim \frac{c_j}{1-\gamma (\mu_1+\mu_2+\tau)}\rho_1(y)\rho_2(y) \lambda(y) \bar F(y).
\end{align*}
\item [(iii)] If, moreover, \textcolor{black}{$\gamma(\frac{q}{q-2}(\mu_1+\mu_2) +2\tau) < 1$ and $\gamma(\frac{q}{q-1}\mu_s +\tau) < 1$ for $s\in \{1,2\}$} and \At~hold, then, as $y\to\infty$, for all $(j\neq\ell)\in\{1,\dots,p\}^2$,
\begin{align*}
\cov\left[\Lambda^{(j)}\varepsilon^{(j)}\rho_1(Y)\ind,\Lambda^{(\ell)}\varepsilon^{(\ell)}\rho_2(Y)\ind\right] &=O\left(  \rho_1(y)\rho_2(y)\lambda^{2(1-2/q)}(y)\bar F^{1-2/q}(y) \right).
\end{align*}
\item [(iv)] If \textcolor{black}{$\gamma(\frac{q}{q-1}(\mu_1+\mu_2) +2\tau) < 1$} and \At~hold, then, as $y\to\infty$, for all $(j\neq\ell)\in\{1,\dots,p\}^2$,
\begin{align*}
\cov\left[\Lambda^{(j)}\varepsilon^{(j)}\rho_1(Y)\ind,\Lambda^{(\ell)}\rho_2(Y)\ind\right] &=O\left(  \rho_1(y)\rho_2(y) \lambda^{2(1-1/q)}(y)\bar F^{1-1/q}(y) \right).
\end{align*}
\item [(v)] If \textcolor{black}{$\gamma(\frac{q}{q-1}(\mu_1+\mu_2)\vee \mu_2 +\tau) < 1$, $\gamma(\mu_2 +\tau) < 1$} and \At~hold, then, as $y\to\infty$, for all $j\in\{1,\dots,p\}$,
$$\cov\left[\Lambda^{(j)}\varepsilon^{(j)}\rho_1(Y)\ind,\Lambda^{(j)}\rho_2(Y)\ind\right] =O\left( \rho_1(y)\rho_2(y) (\lambda \bar{F})^{1-2/q}(y) \right).$$

\end{itemize}
\end{Lem}
\begin{proof}[\bf{\textup{Proof}}]
(i)
Taking account of
\begin{align*}
    \cov\left[\Lambda^{(j)}\rho_1(Y)\ind,\Lambda^{(\ell)}\rho_2(Y)\ind\right] &= \mathbb{E}\left\{{\Lambda^{(j)}} {\Lambda^{(\ell)}}  \rho_1(Y)\rho_2(Y) \ind  \right \} \\
    &- \mathbb{E}\left\{ {\Lambda^{(j)}}  \rho_1(Y) \ind  \right \}\mathbb{E}\left\{  {\Lambda^{(\ell)}}\rho_2(Y) \ind  \right \}
\end{align*}
and applying Lemma~\ref{denom1}(ii,iii) show that, for $j\neq \ell$:
$$
    \cov\left[\Lambda^{(j)}\rho_1(Y)\ind,\Lambda^{(\ell)}\rho_2(Y)\ind\right] \sim \mathbb{E}\left\{{\Lambda^{(j)}} {\Lambda^{(\ell)}}  \rho_1(Y)\rho_2(Y) \ind  \right \}
$$
as $y\to\infty$ and the result follows. \\
(ii) Second,
\begin{align*}
    \cov\left[\Lambda^{(j)}\rho_1(Y)\ind,\Lambda^{(j)}\rho_2(Y)\ind\right] &= \mathbb{E}\left\{{\Lambda^{(j)}}   \rho_1(Y)\rho_2(Y) \ind  \right \} \\
    &- \mathbb{E}\left\{ {\Lambda^{(j)}}  \rho_1(Y) \ind  \right \}\mathbb{E}\left\{  {\Lambda^{(j)}}\rho_2(Y) \ind  \right \}
\end{align*}
and Lemma~\ref{denom1}(ii) implies
$$
 \cov\left[\Lambda^{(j)}\rho_1(Y)\ind,\Lambda^{(j)}\rho_2(Y)\ind\right] \sim \mathbb{E}\left\{{\Lambda^{(j)}}  \rho_1(Y)\rho_2(Y) \ind  \right \}
$$
as $y\to\infty$ so that the result is proved. 

\noindent (iii)
The covariance $$C^{\rho_1,\varepsilon,j}_{\rho_2,\varepsilon,\ell}:=\cov\left[\Lambda^{(j)}\varepsilon^{(j)}\rho_1(Y)\ind,\Lambda^{(\ell)}\varepsilon^{(\ell)}\rho_2(Y)\ind\right]$$ may be expanded as:
$$
 \mathbb{E}\left\{{\Lambda^{(j)}} {\Lambda^{(\ell)}}  \varepsilon^{(j)}\varepsilon^{(\ell)}\rho_1(Y)\rho_2(Y) \ind  \right \} - \mathbb{E}\left\{ {\Lambda^{(j)}}\varepsilon^{(j)}  \rho_1(Y) \ind  \right \}\mathbb{E}\left\{  {\Lambda^{(\ell)}}\varepsilon^{(\ell)}\rho_2(Y) \ind  \right \}.
$$ 
The first term is bounded using H\"older inequality, applied for all $q'\geq 1$ such that $2/q + {1}/{q'} =1$:
\begin{eqnarray}
 \mathbb{E}(|\varepsilon^{(j)}|^{q})^{1/{q}}\mathbb{E}(|\varepsilon^{(\ell)}|^{q})^{1/{q}} \mathbb{E}(( \rho_1(Y)\rho_2(Y) )^{q'}{\Lambda^{(j)}} {\Lambda^{(\ell)}}  \ind)^{1/{q'}}. \nonumber 
\end{eqnarray}
Similarly, for all $q''\geq 1$ such that $1/q + {1}/{q''} =1$,
the product of the second and third terms is bounded above by
\begin{eqnarray}
 \mathbb{E}(|\varepsilon^{(j)}|^{q})^{1/{q}}\mathbb{E}(|\varepsilon^{(\ell)}|^{q})^{1/{q}} \mathbb{E}( \rho_1(Y)^{q''}{\Lambda^{(j)}}  \ind)^{1/{q''}}\mathbb{E}( \rho_2(Y)^{q''}{\Lambda^{(\ell)}}  \ind)^{1/{q''}}. \nonumber 
\end{eqnarray}
In view of Lemma~\ref{denom1}(ii,iii) together \textcolor{black}{with $\gamma(\frac{q}{q-2}(\mu_1+\mu_2) +2\tau) < 1$ and $\gamma(\frac{q}{q-1}\mu_s +\tau) < 1$ for $s\in \{1,2\}$} and \At, it follows that for any $q'>2,q''>1$ such that $2/q + {1}/{q'} =1$ and $1/q + {1}/{q''} =1$,
\begin{align*}
C^{\rho_1,\varepsilon,j}_{\rho_2,\varepsilon,\ell}&=O(  ( \rho_1\cdot\rho_2\cdot \lambda^{2/q'}\cdot \bar F^{1/q'}(y)) \vee ( \rho_1\cdot\rho_2 \cdot \lambda^{2/q''}\cdot\bar F(y)^{2/q''}) ) \nonumber
\\
&=O\left((  \rho_1\cdot\rho_2\cdot \lambda^{2/q'}\cdot \bar F^{1/q'})(y) \right).
\end{align*}
The result follows by using $2/q+1/q'=1$.

\noindent (iv)
It readily follows from the same Hölder type argument with $1/p+1/q=1$ and an application of Lemma~\ref{denom1}(iii).

\noindent (v)
The covariance $$C^{\rho_1,\varepsilon,j}_{\rho_2,j}:=\cov\left[\Lambda^{(j)}\varepsilon^{(j)}\rho_1(Y)\ind,\Lambda^{(j)}\rho_2(Y)\ind\right]$$ may be expanded as:
$$
 \mathbb{E}\left\{{\Lambda^{(j)}} \varepsilon^{(j)}\rho_1(Y)\rho_2(Y) \ind  \right \} - \mathbb{E}\left\{ {\Lambda^{(j)}}\varepsilon^{(j)}  \rho_1(Y) \ind  \right \}\mathbb{E}\left\{  {\Lambda^{(j)}}\rho_2(Y) \ind  \right \}.
$$ 
The first term is bounded using H\"older inequality, applied for all $q'\geq 1$ such that $1/q + {1}/{q'} =1$:
\begin{eqnarray}
 | \mathbb{E}\{{\Lambda^{(j)}} \varepsilon^{(j)}\rho_1(Y)\rho_2(Y) \ind   \}|\le \mathbb{E}(|\varepsilon^{(j)}|^{q})^{1/{q}} \mathbb{E}(( \rho_1(Y)\rho_2(Y) )^{q'}{\Lambda^{(j)}}   \ind)^{1/{q'}}. \nonumber 
\end{eqnarray}
Similarly, the second term is bounded above:
\begin{eqnarray}
 |\mathbb{E}\{ {\Lambda^{(j)}}\varepsilon^{(j)}  \rho_1(Y) \ind   \}|\le  \mathbb{E}(|\varepsilon^{(j)}|^{q})^{1/{q}} \mathbb{E}( \rho_1(Y)^{q'}{\Lambda^{(j)}}   \ind)^{1/{q'}}. \nonumber 
\end{eqnarray}
In view of Lemma~\ref{denom1}(ii) together \textcolor{black}{with $\gamma(\frac{q}{q-1}(\mu_1+\mu_2)\vee \mu_2 +\tau) < 1$, $\gamma(\mu_2 +\tau) < 1$} and \At, it follows that for any $q'>1$ such that $1/q + {1}/{q'} =1$,
\begin{align*}
C^{\rho_1,\varepsilon,j}_{\rho_2,j}&=O\left( \rho_1(y)\rho_2(y) (\lambda \bar{F})^{1-2/q}(y) \right).
\end{align*}
The result is thus proved.
\end{proof}

\noindent Finally, Lemma~\ref{denom1} also allows us to derive asymptotic equivalents of conditional tail moments in the missing data framework.

\begin{Lem}
\label{newlemma}
   Let $Y$ be a random variable satisfying \Au. Let $\iota \in \{0,1\}$.
   \begin{itemize}
    \item [(i)] Suppose \Aq~holds. Then, as $y\to\infty$, 
    \begin{align*}
     m_{\lambda^{(j)}}(y) \sim \frac{c_j}{1-\gamma\tau} \lambda(y)\bar{F}(y),
\end{align*}
for all $j\in\{1,\dots,p\}$.

\item [(ii)] Assume \Mu, \Ad~and~\At~hold with \textcolor{black}{$\gamma(\kappa+\iota)<1$}.
 Then, as $y\to\infty$,
 \begin{align*}
m_{Y^\iota X^{(j)}}(y) \sim \frac{\beta_j}{1-\gamma (\kappa+\iota)} y^\iota g(y) \bar F(y),
\end{align*}
for all $j\in\{1,\dots,d\}$.
\item [(iii)] Assume \Mu, \Ad,~\At~and~\Aq~hold with
\textcolor{black}{$\gamma(\kappa+\tau+\iota)<1$}, and \textcolor{black}{$1< \gamma(q\kappa+\tau)$}. Then, as $y\to\infty$,
\begin{align*}
m_{Y^\iota\Lambda^{(j)}X^{(j)}}(y) &\sim 
\frac{\beta_j c_j}{1-\gamma(\kappa+\tau+\iota)}y^\iota g(y)\lambda(y)\bar{F}(y),
\end{align*}
for all $j\in\{1,\dots,d\}$.

\end{itemize}
\end{Lem}
\begin{proof}[\bf{\textup{Proof}}]
(i) One may apply Lemma~\ref{denom1}(i) since $\gamma\tau <1$ in view of $\tau<0$, to get 
for $j\in\{1,\dots,p\}$ as $y\to\infty$:
$$
    m_{\lambda^{(j)}}(y) \sim \frac{c_j}{1-\gamma\tau} \lambda(y)\bar{F}(y).
$$ 
(ii) \cite[Lemma~7(i)]{Bousebata2023} yields for $j\in\{1,\dots,d\}$, as long as \textcolor{black}{$\gamma(\kappa+\iota)<1$},
$$
m_{Y^\iota X^{(j)}}(y) \sim \frac{\beta_j}{1-\gamma (\kappa+s)}y^\iota g(y) \bar F(y),
$$
as $y\to \infty$.
\\
(iii) According to model \Mu, the following decomposition holds
for $j\in\{1,\dots,d\}$:
\begin{align*}
m_{Y\Lambda^{(j)}X^{(j)}}(y) &= \beta_j m_{Y\Lambda^{(j)}g(Y)} (y)+m_{Y\Lambda^{(j)}\varepsilon^{(j)}}(y).
\end{align*} 
The asymptotics of Lemma~\ref{denom1}(ii,iv) together with assumptions \At~and \Aq, readily yield, as $y\to \infty$, under \textcolor{black}{$\gamma(\kappa+\tau+\iota)<1$} and \textcolor{black}{$\gamma( \frac{q}{q-1}+\tau)<1$}, \begin{align*}
m_{Y^\iota \Lambda^{(j)}g(Y)}(y) &\sim \frac{\beta_j c_j}{1-\gamma(\kappa+\tau+\iota)}y^\iota g(y)\lambda(y)\bar{F}(y),
\\
    m_{Y^\iota\Lambda^{(j)}\varepsilon^{(j)}}(y) &=O\left(y^\iota\lambda(y)^{1-1/q} \bar F(y)^{1-1/q}\right)
    =o(m_{Y^\iota\Lambda^{(j)}g(Y)}(y)),
\end{align*} 
 where the last negligibility is induced by the condition \textcolor{black}{$1< \gamma(q\kappa+\tau)$}.
Thus, it follows that, for any $j\in\{1,\dots,d\}$, as $y\to \infty$,
$$
m_{Y^\iota\Lambda^{(j)}X^{(j)}}(y) = 
\frac{\beta_j c_j}{1-\gamma(\kappa+\tau+\iota)} y^\iota g(y)\lambda(y)\bar{F}(y)(1+o(1)) 
$$
and the result is proved.
\end{proof}


\subsection{Proof of main results}
\label{sec-annexe-main}


\begin{proof}[\textbf{\textup{Proof of Proposition~\ref{prop-loi-jointe-alpha-mixing}}}]
Consider $\delta>0$ given in condition \eqref{hyp:mixing} of strong mixing, $\theta_0:=1+\frac{\delta}{2+\delta}$ and $\theta\ge \theta_0$. In virtue of the Cramér–Wold device and to establish the joint asymptotic normality of the $(3d+1)-$ random vector $\Xi_n$, we shall prove that any non-zero linear combination of its components is asymptotically Gaussian.
Consider then
$ \left(a_{1,1},\dots,a_{1,d},a_{2,1},\dots,a_{2,d}, a_{3,1},\dots,a_{3,d},a_4,a_5\right)^\top \in {\mathbb R}^{3d+2}$ deterministic and let us investigate the asymptotic distribution of $\tilde \chi_n$ defined as follows:
\begin{eqnarray*}
 \chi_n &=&  \sqrt{n(\lambda\bar{F})^{\theta}(y_n)}\left\{ \sum_{k=1}^3 \sum_{j=1}^d a_{k,j} \chi^{(k)}_{j,n} + a_4  \chi^{(4)}_n +a_5  \chi^{(5)}_n \right\} := \sum_{i=1}^n \chi_{i,n},
\end{eqnarray*}
where $\chi_{i,n}:= \sqrt{\frac{(\lambda\bar{F})^{\theta}(y_n)}{n}} (\sum_{k=1}^3 \chi^{(k)}_{i,n} +  \chi^{(4)}_{i,n} + \chi^{(5)}_n  )$ with, for any $1\le j \le d$ and for $\iota \in \{0,1\}$,
\begin{align*}
    &\chi^{(1+\iota)}_{i,n} :=  \sum_{j=1}^da_{1,j}\left( \frac{Y_i^\iota\Lambda^{(j)}_iX^{(j)}_i}{m_{Y^\iota\Lambda^{(j)}X^{(j)}}(y_n)}\indin  -1\right)
,\qquad \chi^{(2)}_{i,n} := \sum_{j=1}^da_{2,j}\left(\frac{\Lambda^{(j)}_i}{m_{\lambda^{(j)}}(y_n)} \indin - 1 \right),
        \\
&\chi^{(4+\iota)}_{i,n}:=a_{4+\iota}\left(Y^\iota_i\frac{\indin}{m_{Y^\iota}(y_n)} -1\right).
\end{align*}
By measurable mapping, the  $\alpha$-mixing hypothesis implies that $(\chi_{i,n})_{1\le i \le n}$ is a $\alpha$-mixing triangular, row-strictly stationary array of real random variables. Define for $1\le i<j\le n$, $\mathcal{F}_{n,i,j}=\sigma(\chi_{n,i},\chi_{n,i+1},\ldots,\chi_{n,j})$ and the associated strong mixing coefficient by 
$$ \forall n\ge 2,\quad \forall \ell<n,\quad \alpha_n(\ell):=\sup\{  | \p (A\cap B)-\p(A)\p(B)|, A\in \mathcal{F}_{n,1,k},B\in \mathcal{F}_{n,k+\ell,n},1\le k\le n-\ell \}.$$
Then, it readily follows that $\alpha_n \le \alpha$ and one is in the right setting to apply \cite[Lemma~C.7]{DavisonPadoanStupfler2023} to the sequence $(\chi_{i,n})_{1\le i \le n}$. The four requirements of this result are 
\begin{align}\label{eq:rec1}
    \frac{n}{r_n}\var\left(\sum_{i=1}^{\ell_n} \chi_{i,n}\right) &= \frac{n}{r_n}\ell_n\var( \chi_{1,n})+ 2\frac{n}{r_n}\sum_{1\le i_1<i_2\le \ell_n}\cov( \chi_{i_1,n},\chi_{i_2,n}) \ll 1 , \\
\label{eq:rec2}\frac{n}{r_n}\var\left(\sum_{i=1}^{r_n} \chi_{i,n}\right) &= n\var( \chi_{1,n})+ 2\frac{n}{r_n}\sum_{1\le i_1<i_2\le r_n}\cov( \chi_{i_1,n},\chi_{i_2,n}) \to \sigma^2<\infty,\\
\label{eq:rec3}
\var\left( \sum_{i=1}^{ n-r_n\lfloor n/r_n \rfloor }\chi_{i,n} \right)&=(n-r_n\lfloor n/r_n \rfloor )\var( \chi_{1,n})+ 2\sum_{1\le i_1<i_2\le n-r_n\lfloor n/r_n \rfloor }\cov( \chi_{i_1,n},\chi_{i_2,n}) \nonumber
\\&\le r_n\var( \chi_{1,n})+ 2\sum_{1\le i_1<i_2\le n-r_n\lfloor n/r_n \rfloor }\cov( \chi_{i_1,n},\chi_{i_2,n}) \ll 1,\end{align}
and \begin{align}\label{eq:lyapunov}\frac{n}{r_n}\mathbb{E}\Big(|\sum_{i=1}^{r_n} \chi_{i,n}|^{2+\delta}\Big) \ll 1.
\end{align}
The first step is to study the variance $\var(\chi_{1,n})$ involved in \eqref{eq:rec1}--\eqref{eq:rec3}. An expansion yields:
\begin{align}\label{eq:var_template}
\var(\chi_{1,n})&= \frac{(\lambda \bar{F})^{\theta}(y_n)}{n} \sum_{k=1}^5\var(\chi^{(k)}_{1,n}) + 2\frac{(\lambda \bar{F})^{\theta}(y_n)}{n}  \sum_{k_1<k_2} \cov(\chi^{(k_1)}_{1,n},\chi^{(k_2)}_{1,n}) ,
\end{align}
where each term of variance and covariance will be addressed separately. We begin by collecting the variances, starting with the first case $k=1$. Using the standard variance/covariance properties, as well as the identity $\cov(X-\mathbb{E} X, Y-\mathbb{E} Y)=\cov(X,Y)$, one gets
\begin{align*}
\var(\chi^{(1)}_{i,n}) &= \var\left( \sum_{j=1}^d a_{1,j}\left(\frac{  \Lambda_i^{(j)} X_{i}^{(j)}}{m_{\Lambda^{(j)} X^{(j)}}(y_n)} \indin -1\right)\right)
\\&= \sum_{j=1}^da_{1,j}^2\frac{ \var( \Lambda_i^{(j)} X_{i}^{(j)} \indin )}{m_{\Lambda^{(j)} X^{(j)}}^2(y_n)}
\\&+ 2\sum_{1\le j_1<j_2\le d}a_{1,j_1}a_{1,j_2}\frac{\cov(  \Lambda_i^{(j_1)} X_{i}^{(j_1)}\indin,   \Lambda_i^{(j_2)} X_{i}^{(j_2)}\indin )}{m_{\Lambda^{(j_1)} X^{(j_1)}}(y_n)m_{\Lambda^{(j_2)} X^{(j_2)}}(y_n)} .
\end{align*}
Under \Mu, the first term $\var(\Lambda_i^{(j)} X_{i}^{(j)} \indin )/m_{\Lambda^{(j)} X^{(j)}}^2(y_n)$ may be factorized as\begin{align*}
& \beta_j^2\frac{m_{g^2(Y)\Lambda^{(j)}}(y_n)}{m_{\Lambda^{(j)} X^{(j)}}^2(y_n)}\left\{1+\frac{m_{\Lambda^{(j)}\varepsilon^{(j)}}(y_n) }{\beta_j^2 m_{g^2(Y)\Lambda^{(j)}}(y_n)}
    +2\frac{ m_{g(Y)\Lambda^{(j)}\varepsilon^{(j)}}(y_n)}{\beta_jm_{g^2(Y)\Lambda^{(j)}}(y_n)}-\frac{m^2_{\Lambda^{(j)} X^{(j)}}(y_n)}{\beta_j^2m_{g^2(Y)\Lambda^{(j)}}(y_n)}\right\}.
\end{align*}
Using Lemma~\ref{newlemma}(iii) under \textcolor{black}{$\gamma(\kappa + \tau) <1$} and \textcolor{black}{$1<\gamma(q\kappa+\tau)$}, Lemma~\ref{denom1}(ii) under \textcolor{black}{$\gamma(2\kappa+\tau)<1$} and Lemma~\ref{denom1}(iv) under \textcolor{black}{$\gamma(\frac{q}{q-1}\kappa+\tau)<1$}, it follows:
$$\sum_{j=1}^da_{1,j}^2\frac{ \var(  \Lambda_i^{(j)} X_{i}^{(j)} \indin )}{m_{\Lambda^{(j)} X^{(j)}}^2(y_n)} = \frac{(1+o(1))}{\lambda(y_n)\bar{F}(y_n)} \frac{(1-\gamma(\kappa+\tau))^2}{1-\gamma(2\kappa+\tau)} \sum_{j=1}^d \frac{a^2_{1,j}}{c_j} .
$$
Besides, under \Mu, the covariance $\cov(  \Lambda_i^{(j_1)} X_{i}^{(j_1)}\indin, \Lambda_i^{(j_2)} X_{i}^{(j_2)}\indin )$ is  \begin{align*}
& \beta_{j_1}\beta_{j_2}\cov(g(Y_i)\Lambda_i^{(j_1)} \indin,   g(Y_i)\Lambda_i^{(j_2)} \indin ) 
  \\&  +\beta_{j_1}\cov(g(Y_i)\Lambda_i^{(j_1)} \indin,   \Lambda_i^{(j_2)} \varepsilon_{i}^{(j_2)}\indin )
    \\&+\beta_{j_2}\cov(\Lambda_i^{(j_1)} \varepsilon_{i}^{(j_1)}\indin,   g(Y_i)\Lambda_i^{(j_2)} \indin )
    \\&+\cov(\Lambda_i^{(j_1)} \varepsilon_{i}^{(j_1)}\indin,   \Lambda_i^{(j_2)} \varepsilon_{i}^{(j_2)}\indin ).
\end{align*}
The first term is $\frac{c_{j_1}c_{j_2}\beta_{j_1}\beta_{j_2}}{1-2\gamma(\kappa+\tau)}  g^2(y_n) \lambda^2(y_n)\bar F(y_n)(1+o(1))$ by Lemma~\ref{lem-cov}(i) under \textcolor{black}{$2\gamma(\kappa+\tau)<1$}. The second and third terms are $O\left(g(y_n)\lambda^{2(1-1/q)}(y_n)\bar{F}^{1-1/q}(y_n)\right)$ by Lemma~\ref{lem-cov}(iv) under \textcolor{black}{$\gamma(\frac{q}{q-1}\kappa+2\tau)<1$}. The fourth term is $O\left(\lambda^{2(1-2/q)}(y_n)\bar{F}^{1-2/q}(y_n)\right)$ using Lemma~\ref{lem-cov}(iii) and is negligible compared to the second one. The latter is in turn negligible compared to the first one  under \textcolor{black}{$1<\gamma(2\kappa +2q +\tau)$}. Gathering everything, it follows that
$$
\cov(  \Lambda_i^{(j_1)} X_{i}^{(j_1)}\indin, \Lambda_i^{(j_2)} X_{i}^{(j_2)}\indin )=
\frac{\beta_{j_1}\beta_{j_2}c_{j_1}c_{j_2}}{1-2\gamma(\kappa+\tau)}g^2(y_n) (\lambda\bar{F})(y_n)(1+o(1)) ,
$$
as $n\to\infty$. Then, by Lemma~\ref{newlemma}(iii) under \textcolor{black}{$\gamma(\kappa + \tau) <1$} and \textcolor{black}{$1<\gamma(q\kappa  +\tau)$} for the denominator, it yields
\begin{align}\label{eqvar1}
\var(\chi^{(1)}_{i,n})&= \frac{(1+o(1))}{\lambda(y_n)\bar{F}(y_n)}  \left\{ \frac{(1-\gamma(\kappa+\tau))^2}{1-\gamma(2\kappa+\tau)} \sum_{j=1}^d \frac{a^2_{1,j}}{c_j} +  \frac{(1-\gamma(\kappa+\tau))^2}{1-2\gamma(\kappa+\tau)} \sum_{1\le j_1<j_2\le d} 2 a_{1,j_1}a_{1,j_2} \right\} .
\end{align}
We proceed in the same way for $\chi^{(2)}_{i,n}$. One has to begin:
\begin{align*}
\var(\chi^{(2)}_{i,n}) &= \var\left( \sum_{j=1}^d a_{2,j}\left(\frac{ Y_i \Lambda_i^{(j)} X_{i}^{(j)}}{m_{Y\Lambda^{(j)} X^{(j)}}(y_n)} \indin -1\right)\right)
\\&= \sum_{j=1}^da_{2,j}^2\frac{ \var( Y_i\Lambda_i^{(j)} X_i^{(j)} \indin )}{m_{Y\Lambda^{(j)} X^{(j)}}^2(y_n)}
\\&+ 2\sum_{1\le j_1<j_2\le d}a_{2,j_1}a_{2,j_2}\frac{\cov(  Y_i\Lambda_i^{(j_1)} X_{i}^{(j_1)}\indin,  Y_i \Lambda_i^{(j_2)} X_{i}^{(j_2)}\indin )}{m_{Y\Lambda^{(j_1)} X^{(j_1)}}(y_n)m_{Y\Lambda^{(j_2)} X^{(j_2)}}(y_n)} .
\end{align*}
Under \Mu, the term $\var(Y_i\Lambda_i^{(j)} X_{i}^{(j)} \indin )/m_{Y\Lambda^{(j)} X^{(j)}}^2(y_n)$ may be factorized as\begin{align*}
& \beta_j^2\frac{m_{Y^2g^2(Y)\Lambda^{(j)}}(y_n)}{m_{Y\Lambda^{(j)} X^{(j)}}^2(y_n)} \left\{ 1+\frac{m_{Y\Lambda^{(j)}\varepsilon^{(j)}}(y_n) }{\beta_j^2m_{Y^2g^2(Y)\Lambda^{(j)}}(y_n)}
    +2\frac{ m_{Yg(Y)\Lambda^{(j)}\varepsilon^{(j)}}(y_n)}{\beta_jm_{Y^2g^2(Y)\Lambda^{(j)}}(y_n)}-\frac{m^2_{Y\Lambda^{(j)} X^{(j)}}(y_n)}{\beta_j^2m_{Y^2g^2(Y)\Lambda^{(j)}}(y_n)} \right\}.
\end{align*}
Using Lemma~\ref{newlemma}(iii) under \textcolor{black}{$\gamma(\kappa + \tau + 1) <1$} and \textcolor{black}{$1<\gamma(q(\kappa+1)+\tau)$}, Lemma~\ref{denom1}(ii) under \textcolor{black}{$\gamma(2\kappa+\tau+2)<1$} and Lemma~\ref{denom1}(iv) under \textcolor{black}{$\gamma(\frac{q}{q-1}\kappa+\tau)<1$}, it follows that
$$
\sum_{j=1}^da_{2,j}^2\frac{ \var(  Y_i\Lambda_i^{(j)} X_{i}^{(j)} \indin )}{m_{Y\Lambda^{(j)} X^{(j)}}^2(y_n)} = \frac{(1+o(1))}{\lambda(y_n)\bar{F}(y_n)} \frac{(1-\gamma(\kappa+\tau+1))^2}{1-\gamma(2\kappa+\tau+2)} \sum_{j=1}^d \frac{a^2_{2,j}}{c_j} .
$$
Besides, under \Mu, the covariance $\cov(  Y_i \Lambda_i^{(j_1)} X_{i}^{(j_1)}\indin, Y_i \Lambda_i^{(j_2)} X_{i}^{(j_2)}\indin )$ is  \begin{align*}
& \beta_{j_1}\beta_{j_2}\cov(Y_ig(Y_i)\Lambda_i^{(j_1)} \indin,   Y_ig(Y_i)\Lambda_i^{(j_2)} \indin ) 
  \\&  +\beta_{j_1}\cov(Y_ig(Y_i)\Lambda_i^{(j_1)} \indin,   Y_i\Lambda_i^{(j_2)} \varepsilon_{i}^{(j_2)}\indin )
    \\&+\beta_{j_2}\cov(Y_i\Lambda_i^{(j_1)} \varepsilon_{i}^{(j_1)}\indin,  Y_i g(Y_i)\Lambda_i^{(j_2)} \indin )
    \\&+\cov(Y_i\Lambda_i^{(j_1)} \varepsilon_{i}^{(j_1)}\indin,   Y_i\Lambda_i^{(j_2)} \varepsilon_{i}^{(j_2)}\indin ).
\end{align*}
Factorizing by the first and highest term, using Lemma~\ref{lem-cov}(i) under \textcolor{black}{$2\gamma(\kappa+\tau+1)<1$}, Lemma~\ref{lem-cov}(iv) under \textcolor{black}{$\gamma(\frac{q}{q-1}\kappa+2\tau)<1$} and Lemma~\ref{lem-cov}(iii), one obtains 
$$
\cov(  Y_i \Lambda_i^{(j_1)} X_{i}^{(j_1)}\indin, Y_i \Lambda_i^{(j_2)} X_{i}^{(j_2)}\indin )=\frac{\beta_{j_1}\beta_{j_2}c_{j_1}c_{j_2}}{1-2\gamma(\kappa+\tau+1)}y_n^2g^2(y_n) (\lambda\bar{F})(y_n)(1+o(1)) ,
$$
as $n\to\infty$. Then, by Lemma~\ref{newlemma}(iii) under \textcolor{black}{$\gamma(\kappa + \tau + 1) <1$} and \textcolor{black}{$1<\gamma(q\kappa  +\tau)$} for the denominator, it yields \begin{align}\label{eqvar3}
\var(\chi^{(2)}_{i,n})&= \frac{(1+o(1))}{\lambda(y_n)\bar{F}(y_n)} \left\{ \frac{(1-\gamma(\kappa+\tau+1))^2}{1-\gamma(2\kappa+\tau+2)}\sum_{j=1}^d \frac{a^2_{2,j}}{c_j} +  \frac{(1-\gamma(\kappa+\tau+1))^2}{1-2\gamma(\kappa+\tau+1)} \sum_{1\le j_1<j_2\le d} 2 a_{2,j_1}a_{2,j_2} \right\}.
\end{align}

 Similarly, using Lemma~\ref{newlemma}(i) for the denominator and  Lemma~\ref{lem-cov}(i) for the covariance entails
\begin{align}
\nonumber
\var(\chi^{(3)}_{i,n})&=\sum_{j=1}^d a_{3,j}^2 \frac{ \var( \Lambda_i^{(j)} \indin)}{m_{\lambda^{(j)}}^2(y_n)}
+2 \sum_{1\le j_1<j_2\le d}a_{3,j_1}a_{3,j_2} \frac{\cov(  \Lambda_i^{(j_1)}\indin,   \Lambda_i^{(j_2)} \indin )}{m_{\lambda^{(j_1)} }(y_n)m_{\lambda^{(j_2)}}(y_n)}
\\&=\frac{(1+o(1))}{\lambda(y_n)\bar{F}(y_n)}(1-\gamma\tau)\sum_{j=1}^d 
 \frac{a_{3,j}^2}{c_j} + O\left(\frac{1}{\bar{F}(y_n)}\right)
=  \frac{(1+o(1))}{\lambda(y_n)\bar{F}(y_n)}(1-\gamma\tau)\sum_{j=1}^d 
 \frac{a_{3,j}^2}{c_j} ,
 \label{eqvar2}
\end{align} 
since \textcolor{black}{$\tau < 0$}. Finally, the last sub-variance is easily treated with the variance of a Bernoulli distribution: \begin{align} \label{eqvar4}
\var(\chi^{(4)}_{i,n}) &= \var\left( a_4\left(\frac{  1 }{\bar{F}(y_n)} \indin -1\right)\right) = \frac{a_4^2 (1-\bar{F}(y_n))}{\bar{F}(y_n)}\sim  \frac{a_4^2 }{\bar{F}(y_n)} ,\end{align}
while, using Lemma \ref{denom1}(i) under \textcolor{black}{$\gamma <1/2$},
\begin{align} \label{eqvar5}
\var(\chi^{(5)}_{i,n}) & = a^2_5 \frac{ \var( Y_i \indin)}{m_{Y}^2(y_n)} =  a_5^2\left(\frac{m_{Y^2}(y_n)}{m_{Y}^2(y_n)}-1\right)\sim \frac{(1-\gamma)^2}{1-2\gamma} \ff{\bar{F}(y_n)}.\end{align}
Now, we focus on the sub-covariances in $\var(\chi_{1,n})$. The usual properties of the covariance operator entail
\begin{align*}
\cov(\chi^{(1)}_{i,n},\chi^{(3)}_{i,n})&=\cov\left( \sum_{j=1}^d a_{1,j}\left(\frac{\Lambda_{i}^{(j)} X_{i}^{(j)}}{m_{\Lambda^{(j)} X^{(j)}}(y_n)} \mathds{1}_{\{ Y_{i}\ge y_n \}} -1\right),\sum_{j=1}^d a_{3,j}\left(\frac{  \Lambda_{i}^{(j)}}{m_{\lambda^{(j)}}(y_n)} \mathds{1}_{\{ Y_{i}\ge y_n \}}-1\right)\right)
\\&=\sum_{j_1,j_2=1}^d a_{1,j_1}a_{3,j_2} \frac{\cov( \Lambda_{i}^{(j_1)} X_{i}^{(j_1)}\indin, \Lambda_{i}^{(j_2)}\indin)}{m_{\Lambda^{(j_1)} X^{(j_1)}}(y_n)  m_{\lambda^{(j_2)}}(y_n)}
\\&=\sum_{j_1 \ne j_2} a_{1,j_1}a_{3,j_2}\frac{\cov( \Lambda_i^{(j_1)} X_i^{(j_1)}\indin, \Lambda_i^{(j_2)}\indin)}{m_{\Lambda^{(j_1)} X^{(j_1)}}(y_n)  m_{\lambda^{(j_2)}}(y_n)}
\\&+ \sum_{j=1}^d a_{1,j}a_{3,j} \frac{\cov( \Lambda_i^{(j)} X_i^{(j)}\indin, \Lambda_i^{(j)}\indin)}{m_{\Lambda^{(j)} X^{(j)}}(y_n)  m_{\lambda^{(j)}}(y_n)}.
\end{align*}
Let us start with the double sum in $\cov(\chi^{(1)}_{i,n},\chi^{(3)}_{i,n})$ by writing, under \Mu, that \begin{align*}\cov(\Lambda_i^{(j_1)} X_i^{(j_1)}\indin, \Lambda_i^{(j_2)}\indin)&=\beta_{j_1}\cov( g(Y_i)\Lambda_i^{(j_1)}\indin, \Lambda^{(j_2)}\indin)\\&+\cov(\Lambda_i^{(j_1)} \varepsilon_i^{(j_1)}\indin, \Lambda^{(j_2)}\indin). \end{align*} By Lemma~\ref{lem-cov}(i) under \textcolor{black}{$\gamma(\kappa+2\tau)<1$}, \textcolor{black}{$\gamma\kappa<1$}, the first term is $O((g \lambda^2 \bar F)(y_n))$; while Lemma~\ref{lem-cov}(iv) under \textcolor{black}{$\gamma (\frac{q}{q-1}+2\tau)<1$} shows that the second term is $O((\lambda^{2(1-1/q)}\bar{F}^{1-1/q})(y_n))$. Moreover, $( \lambda^{2(1-1/q)}\bar{F}^{1-1/q})(y_n)\ll (g \lambda^2 \bar F)(y_n)$ under \textcolor{black}{$1<\gamma(q\kappa+2\tau)$}, so that using Lemma~\ref{newlemma}(i), Lemma~\ref{newlemma}(iii) and Lemma~\ref{lem-cov}(iv) under \textcolor{black}{$\gamma(\frac{q}{q-1}+2\tau)<1$} and \textcolor{black}{$1<\gamma(q\kappa+\tau)$}, one has, since $d<+\infty$,
\begin{align*}
\sum_{j_1 \ne j_2} a_{1,j_1}a_{3,j_2}\frac{\cov( \Lambda_i^{(j_1)} X_i^{(j_1)}\indin, \Lambda_i^{(j_2)}\indin)}{m_{\Lambda^{(j_1)} X^{(j_1)}}(y_n)  m_{\lambda^{(j_2)}}(y_n)}&=O\left(\frac{1}{\bar{F}(y_n)}\right) .
\end{align*}
The second double sum in $\cov(\chi^{(1)}_{i,n},\chi^{(3)}_{i,n})$ is treated in the similar fashion, by first decomposing under \Mu~and then using this time Lemma~\ref{lem-cov}(v). All-in-one, it yields: 
\begin{align}
\cov(\chi^{(1)}_{i,n},\chi^{(3)}_{i,n})&=  \frac{1}{\lambda(y_n)\bar{F}(y_n)}\sum_{j=1}^d \frac{a_{1,j}a_{3,j}}{c_j}.
\label{eqcovar}
\end{align}
Concerning $\cov(\chi^{(2)}_{i,n},\chi^{(3)}_{i,n})$, we begin by writing:\begin{align*}
\cov(\chi^{(2)}_{i,n},\chi^{(3)}_{i,n})&=\cov\left( \sum_{j=1}^d a_{2,j}\left(\frac{Y_i\Lambda_{i}^{(j)} X_{i}^{(j)}}{m_{Y\Lambda^{(j)} X^{(j)}}(y_n)} \mathds{1}_{\{ Y_{i}\ge y_n \}} -1\right),\sum_{j=1}^d a_{3,j}\left(\frac{  \Lambda_{i}^{(j)}}{m_{\lambda^{(j)}}(y_n)} \mathds{1}_{\{ Y_{i}\ge y_n \}}-1\right)\right)
\\&=\sum_{j_1,j_2=1}^d a_{2,j_1}a_{3,j_2} \frac{\cov( Y_i\Lambda_{i}^{(j_1)} X_{i}^{(j_1)}\indin, \Lambda_{i}^{(j_2)}\indin)}{m_{Y\Lambda^{(j_1)} X^{(j_1)}}(y_n)  m_{\lambda^{(j_2)}}(y_n)}
\\&=\sum_{j_1 \ne j_2} a_{2,j_1}a_{3,j_2}\frac{\cov( Y_i\Lambda_i^{(j_1)} X_i^{(j_1)}\indin, \Lambda_i^{(j_2)}\indin)}{m_{Y\Lambda^{(j_1)} X^{(j_1)}}(y_n)  m_{\lambda^{(j_2)}}(y_n)}
\\&+ \sum_{j=1}^d a_{2,j}a_{3,j} \frac{\cov(Y_i \Lambda_i^{(j)} X_i^{(j)}\indin, \Lambda_i^{(j)}\indin)}{m_{Y\Lambda^{(j)} X^{(j)}}(y_n)  m_{\lambda^{(j)}}(y_n)}.
\end{align*}
Let us start with the double sum in $\cov(\chi^{(2)}_{i,n},\chi^{(3)}_{i,n})$ by remarking that, under \Mu, \begin{align*}\cov(Y_i\Lambda_i^{(j_1)} X_i^{(j_1)}\indin, \Lambda_i^{(j_2)}\indin)&=\beta_{j_1}\cov( Y_ig(Y_i)\Lambda_i^{(j_1)}\indin, \Lambda^{(j_2)}\indin)\\&+\cov(Y_i\Lambda_i^{(j_1)} \varepsilon_i^{(j_1)}\indin, \Lambda^{(j_2)}\indin). \end{align*} By Lemma~\ref{lem-cov}(i) under \textcolor{black}{$\gamma(\kappa+2\tau+1)<1$}, \textcolor{black}{$\gamma(\kappa+1)<1$}, the first term is $O(y_n(g \lambda^2 \bar F)(y_n))$; while Lemma~\ref{lem-cov}(iv) under \textcolor{black}{$\gamma (\frac{q}{q-1}+2\tau)<1$} shows that the second term is $O(y_n(\lambda^{2(1-1/q)}\bar{F}^{1-1/q})(y_n))$. Moreover, $( \lambda^{2(1-1/q)}\bar{F}^{1-1/q})(y_n)\ll (g \lambda^2 \bar F)(y_n)$ under \textcolor{black}{$1<\gamma(q\kappa+2\tau)$}, so that using Lemma~\ref{newlemma}(i), Lemma~\ref{newlemma}(iii) and Lemma~\ref{lem-cov}(iv) under \textcolor{black}{$\gamma(\frac{q}{q-1}+2\tau)<1$} and \textcolor{black}{$1<\gamma(q\kappa+\tau)$}, the leading term is the first one and by comparison, since $d<+\infty$,
\begin{align*}
\sum_{j_1 \ne j_2} a_{2,j_1}a_{3,j_2}\frac{\cov( Y_i\Lambda_i^{(j_1)} X_i^{(j_1)}\indin, \Lambda_i^{(j_2)}\indin)}{m_{\Lambda^{(j_1)} X^{(j_1)}}(y_n)  m_{\lambda^{(j_2)}}(y_n)}&=O\left(\frac{1}{\bar{F}(y_n)}\right) .
\end{align*}
The second double sum in $\cov(\chi^{(2)}_{i,n},\chi^{(3)}_{i,n})$ is treated in the similar fashion, by first decomposing under \Mu~and then using this time Lemma~\ref{lem-cov}(v) under \textcolor{black}{$\gamma(\frac{q}{q-1}+\tau)<1$}. All-in-one, it yields: 
\begin{align}
\cov(\chi^{(2)}_{i,n},\chi^{(3)}_{i,n})&=  \frac{1}{\lambda(y_n)\bar{F}(y_n)}\sum_{j=1}^d \frac{a_{2,j}a_{3,j}}{c_j}.
\label{eqcovar2}
\end{align}
Regarding $\cov(\chi^{(2)}_{i,n},\chi^{(1)}_{i,n})$, we begin by writing:\begin{align*}
\cov(\chi^{(2)}_{i,n},\chi^{(1)}_{i,n})&=\cov\left( \sum_{j=1}^d a_{2,j}\left(\frac{Y_i\Lambda_{i}^{(j)} X_{i}^{(j)}}{m_{Y\Lambda^{(j)} X^{(j)}}(y_n)} \mathds{1}_{\{ Y_{i}\ge y_n \}} -1\right),\sum_{j=1}^d a_{1,j}\left(\frac{\Lambda_{i}^{(j)} X_{i}^{(j)}}{m_{\Lambda^{(j)} X^{(j)}}(y_n)} \mathds{1}_{\{ Y_{i}\ge y_n \}} -1\right)\right)
\\&=\sum_{j_1,j_2=1}^d a_{2,j_1}a_{1,j_2} \frac{\cov( Y_i\Lambda_{i}^{(j_1)} X_{i}^{(j_1)}\indin, \Lambda_{i}^{(j_2)} X_{i}^{(j_2)}\indin)}{m_{Y\Lambda^{(j_1)} X^{(j_1)}}(y_n)  m_{\Lambda^{(j_2)} X^{(j_2)}}(y_n) }
\\&=\sum_{j_1 \ne j_2} a_{2,j_1}a_{1,j_2}\frac{\cov( Y_i\Lambda_i^{(j_1)} X_i^{(j_1)}\indin, \Lambda_i^{(j_2)}X_{i}^{(j_2)}\indin)}{m_{Y\Lambda^{(j_1)} X^{(j_1)}}(y_n)  m_{\Lambda^{(j_2)} X^{(j_2)}}(y_n) }
\\&+ \sum_{j=1}^d a_{1,j}a_{2,j} \frac{\cov(Y_i \Lambda_i^{(j)} X_i^{(j)}\indin, \Lambda_i^{(j)}X_{i}^{(j)}\indin)}{m_{Y\Lambda^{(j)} X^{(j)}}(y_n)  m_{\lambda^{(j)}X^{(j)}}(y_n)}.
\end{align*}
We start with the double sum in $\cov(\chi^{(2)}_{i,n},\chi^{(1)}_{i,n})$ by writing, under \Mu, that \begin{align*}\cov(Y_i\Lambda_i^{(j_1)} X_i^{(j_1)}\indin, \Lambda_i^{(j_2)}X_i^{(j_2)}\indin)&=\beta_{j_1}\beta_{j_2}\cov( Y_ig(Y_i)\Lambda_i^{(j_1)}\indin, g(Y_i)\Lambda^{(j_2)}\indin)\\&+\beta_{j_1}\cov(Y_ig(Y_i)\Lambda_i^{(j_1)} \indin, \Lambda^{(j_2)}\varepsilon_i^{(j_2)}\indin)
\\&+\beta_{j_2}\cov(\Lambda_i^{(j_1)}\varepsilon_i^{(j_1)} \indin, Y_ig(Y_i)\Lambda^{(j_2)}\indin)
\\&+\cov(Y_i\Lambda_i^{(j_1)}\varepsilon_i^{(j_1)} \indin, \Lambda^{(j_2)}\varepsilon_i^{(j_2)}\indin)
. \end{align*} By Lemma~\ref{lem-cov}(i) under \textcolor{black}{$\gamma(2\kappa+2\tau+1)<1$}, the first term is $O(y_n(g^2 \lambda^2 \bar F)(y_n))$; while Lemma~\ref{lem-cov}(iv) under \textcolor{black}{$\gamma(\frac{q}{q-1}(\kappa + 1) +2\tau) < 1$} shows that the second and third term are $O\left( y_ng(y_n) \lambda^{2(1-1/q)}(y_n)\bar F^{1-1/q}(y_n)\right)$. The fourth term is $O\left( y_n\lambda^{2(1-2/q)}(y_n)\bar F^{1-2/q}(y_n)\right) $ thanks to Lemma~\ref{lem-cov}(iii) under \textcolor{black}{$\gamma(\frac{q}{q-2} +2\tau) < 1$} and \textcolor{black}{$\gamma(\frac{q}{q-1} +\tau) < 1$}. Doing again a comparison, it holds that $( \lambda^{2(1-1/q)}\bar{F}^{1-1/q})(y_n)\ll (g \lambda^2 \bar F)(y_n)$ under \textcolor{black}{$1<\gamma(q\kappa+2\tau)$}, so that using Lemma~\ref{newlemma}(i), Lemma~\ref{newlemma}(iii) and Lemma~\ref{lem-cov}(iv) under \textcolor{black}{$\gamma(\frac{q}{q-1}+2\tau)<1$} and \textcolor{black}{$1<\gamma(q\kappa+\tau)$}, the leading term is the first one and, since $d<+\infty$,
\begin{align*}
\sum_{j_1 \ne j_2} a_{2,j_1}a_{1,j_2}\frac{\cov( Y_i\Lambda_i^{(j_1)} X_i^{(j_1)}\indin, \Lambda_i^{(j_2)}X_{i}^{(j_2)}\indin)}{m_{Y\Lambda^{(j_1)} X^{(j_1)}}(y_n)  m_{\Lambda^{(j_2)} X^{(j_2)}}(y_n) }&=O\left(\frac{1}{\bar{F}(y_n)}\right) .
\end{align*}
The case $\cov(\chi^{(k_1)}_{i,n},\chi^{(k_2)}_{i,n})$ with $k_1\in \{1,2\}, k_2\in \{4,5\}$ is treated as follows. Let $\iota,s\in \{0,1\}$ and using the definition of the tail-moments, one has, since $d<+\infty$,
\begin{align*}
\cov(\chi^{(1+\iota)}_{i,n},\chi^{(4+s)}_{i,n})&=\cov\left( \sum_{j=1}^d a_{1+\iota,j}\left(\frac{Y^\iota_i\Lambda_{i}^{(j)} X_{i}^{(j)}}{m_{Y^\iota\Lambda^{(j)} X^{(j)}}(y_n)} \mathds{1}_{\{ Y_{i}\ge y_n \}} -1\right),a_{4+s}\left(\frac{Y^s_i\indin}{m_{Y^s}(y_n)} -1\right)\right)
\\&=a_{4+s}\sum_{j=1}^d a_{1+\iota,j} \frac{\cov( Y^\iota_i\Lambda_{i}^{(j)} X_{i}^{(j)}\indin,Y^s_i \indin)}{m_{Y^\iota\Lambda^{(j)} X^{(j)}}(y_n)  m_{Y^s}(y_n) } 
\\&=a_{4+s} \sum_{j=1}^d a_{1+\iota,j} \frac{\E( Y^{\iota + s}\Lambda^{(j)} X^{(j)}\ind)-\E( Y^{\iota}\Lambda^{(j)} X^{(j)}\ind)\E( Y^s\ind)}{m_{Y^\iota\Lambda^{(j)} X^{(j)}}(y_n)  m_{Y^s}(y_n) }  
\\&= a_{4+s} \sum_{j=1}^d a_{1+\iota,j} \frac{\E( Y^{\iota + s}\Lambda^{(j)} X^{(j)}\ind)}{m_{Y^\iota\Lambda^{(j)} X^{(j)}}(y_n)  m_{Y^s}(y_n) } (1+o(1)) = O\left(\frac{1}{\bar{F}(y_n)}\right).
\end{align*}
Besides, for the case $k_1=3$ and $k_2\in \{4,5\}$, since $d<+\infty$,
\begin{align*}
\cov(\chi^{(3)}_{i,n},\chi^{(4+s)}_{i,n})&=\cov\left( \sum_{j=1}^d a_{3,j}\left(\frac{\Lambda_{i}^{(j)} }{m_{\lambda^{(j)}}(y_n)} \mathds{1}_{\{ Y_{i}\ge y_n \}} -1\right),a_{4+s}\left(\frac{Y^s_i\indin}{m_{Y^s}(y_n)} -1\right)\right)
\\&=a_{4+s}\sum_{j=1}^d a_{3,j} \frac{\cov( \Lambda_{i}^{(j)} \indin, Y^s_i\indin)}{m_{\lambda^{(j)}}(y_n) m_{Y^s}(y_n) } 
\\&=a_{4+s} \sum_{j=1}^d a_{3,j} \frac{m_{Y^s \Lambda^{(j)}}(y_n)- m_{\lambda^{(j)}}(y_n)m_{Y^s}(y_n)}{m_{\lambda^{(j)}}(y_n)  m_{Y^s}(y_n) }  = O\left(\frac{1}{\bar{F}(y_n)}\right).
\end{align*}
Finally,
\begin{align*}
\cov(\chi^{(5)}_{i,n},\chi^{(4)}_{i,n})&=\cov\left(a_5\left(\frac{Y_{i} }{m_{Y}(y_n)} \mathds{1}_{\{ Y_{i}\ge y_n \}} -1\right),a_4\left(\frac{\indin}{\bar{F}(y_n)} -1\right)\right)
\\&=a_4a_5 \frac{\cov( Y_i \indin, \indin)}{m_{Y}(y_n)  \bar{F}(y_n) } 
= a_4a_5 \frac{\E( Y\ind)(1-\bar{F}(y_n))}{m_{Y}(y_n)  \bar{F}(y_n) }  = O\left(\frac{1}{\bar{F}(y_n)}\right).
\end{align*}
Combining \eqref{eqvar1}--\eqref{eqcovar2} into~\eqref{eq:var_template}, it follows that, for some $\vfi_1$ depending on $d$, $\gamma,\kappa,\tau$ and $a_{k,j}$ for $1\le k \le 5$ and $1\le j \le d$, since $d<+\infty$,
\begin{align}\label{eq:var}
    \var( \tilde \chi_{1,n}) &= \frac{(\lambda \bar{F})^{\theta-1}(y_n)}{n}\vfi_1(1+o(1)).
\end{align}
The second step is the control of the covariances in \eqref{eq:rec1}--\eqref{eq:rec3} and, to this aim, let us consider the expansion, for $i_1\ne i_2$,
\begin{align*}
\cov(\chi_{i_1,n},\chi_{i_2,n}) &= \frac{(\lambda \bar{F})^{\theta}(y_n)}{n} \sum_{k=1}^4 \cov(\chi^{(k)}_{i_1,n},\chi^{(k)}_{i_2,n})
+
2\frac{(\lambda \bar{F})^{\theta}(y_n)}{n} \sum_{k_1<k_2} \cov(\chi^{(k_1)}_{i_1,n},\chi^{(k_2)}_{i_2,n}).
\end{align*}
Our goal now is to prove that, for any $k_1,k_2\in \{1,2,3,4,5\}$ and $i_1,i_2\le n$, it holds that
$$ 
\cov(\chi^{(k_1)}_{i_1,n},\chi^{(k_2)}_{i_2,n}) \lesssim \alpha^{\frac{\delta}{2+\delta}}(|i_1-i_2|) (\lambda \bar{F})^{\frac{2}{2+\delta}-2}(y_n).
$$
Invoking \cite[Theorem~3]{Doukhan1994} with herein $p=q=2+\delta$, one may write, for $\iota \in \{0,1\}$,
\begin{align*}
\cov(\chi^{(1+\iota)}_{i_1,n},\chi^{(1+\iota)}_{i_2,n})&= 
 \sum_{j_1,j_2=1}^d a_{1+\iota,j_1} a_{1+\iota,j_2} \frac{\cov( Y^\iota_{i_1}\Lambda_{i_1}^{(j_1)} X_{i_1}^{(j_1)}  \mathds{1}_{\{ Y_{i_1}\ge y_n \}} ,Y^\iota_{i_2}\Lambda_{i_2}^{(j_2)} X_{i_2}^{(j_2)}  \mathds{1}_{\{ Y_{i_2}\ge y_n \}})}{m_{Y^\iota\Lambda^{(j_1)} X^{(j_1)}(y_n)}m_{Y^\iota\Lambda^{(j_2)} X^{(j_2)}}(y_n)}
 \\&\le 
8\alpha^{\frac{\delta}{2+\delta}}(|i_1-i_2|) \sum_{j_1,j_2=1}^d a_{1
+\iota,j_1}  a_{1+\iota,j_2} \frac{\prod_{s=1}^2\E( Y^\iota \Lambda^{(j_s)} (X^{(j_s)})^{2+\delta}  \mathds{1}_{\{ Y\ge y_n \}})^{\frac{1}{2+\delta}} }{m_{Y^\iota\Lambda^{(j_1)} X^{(j_1)}(y_n)}m_{Y^\iota\Lambda^{(j_2)} X^{(j_2)}}(y_n)}.
 \end{align*}
Then, under model \Mu, and by a Jensen-type argument, it follows that
\begin{align*}
\cov(\chi^{(1+\iota)}_{i_1,n},\chi^{(1+\iota)}_{i_2,n}) &\lesssim \alpha^{\frac{\delta}{2+\delta}}(|i_1-i_2|)\sum_{j_1,j_2=1}^d a_{1+\iota,j_1} a_{1+\iota,j_2}\frac{\max_{j\le d}\E(Y^\iota g^{2+\delta}(Y)\Lambda^{(j)} \mathds{1}_{\{ Y\ge y_n \}} )^{\frac{2}{2+\delta}}}{m_{Y^\iota\Lambda^{(j_1)} X^{(j_1)}(y_n)}m_{Y^\iota\Lambda^{(j_2)} X^{(j_2)}}(y_n)} 
\\&+\alpha^{\frac{\delta}{2+\delta}}(|i_1-i_2|)\sum_{j_1,j_2=1}^d a_{1+\iota,j_1} a_{1+\iota,j_2}\frac{\max_{j\le d}\E(Y^\iota \Lambda^{(j)} (\varepsilon^{(j)})^{2+\delta} \mathds{1}_{\{ Y\ge y_n \}} )^{\frac{2}{2+\delta}}}{m_{Y^\iota\Lambda^{(j_1)} X^{(j_1)}(y_n)}m_{Y^\iota\Lambda^{(j_2)} X^{(j_2)}}(y_n)} .
\end{align*}
Next, using Lemma~\ref{newlemma}(iii) under \textcolor{black}{$1<\gamma(q\kappa+\tau)$} for the denominator and Lemma~\ref{denom1}(ii) under \textcolor{black}{$\gamma((2+\delta)\kappa+\tau+1)<1$} together with Lemma~\ref{denom1}(v) under \textcolor{black}{$\gamma(\frac{q}{q-2-\delta} +\tau)<1$} for respectively the first and second numerator terms, one may write under \At~for the negligibility in the last line, for $n$ large enough, for any $i_2\ne i_2$, since $d<+\infty$ and \textcolor{black}{$1<\gamma(q\kappa+\tau)$}, 
\begin{align}
\nonumber
\cov(\chi^{(1+\iota)}_{i_1,n},\chi^{(1+\iota)}_{i_2,n}) &\lesssim\alpha^{\frac{\delta}{2+\delta}}(|i_1-i_2|)\left\{ \frac{y_n^{\iota\frac{2}{2+\delta}}g^2(y_n)(\lambda\bar{F})^{\frac{2}{2+\delta}}(y_n)}{y_n^2g^2(y_n)(\lambda \bar{F})^2(y_n)} + \frac{y_n^{\iota\frac{2}{2+\delta}}(\lambda\bar{F})^{\frac{2}{2+\delta}-\frac{2}{q}}(y_n)}{y_n^2g^2(y_n)(\lambda \bar{F})^2(y_n)} \right\}
  \\ \nonumber&= \alpha^{\frac{\delta}{2+\delta}}(|i_1-i_2|)y_n^{\iota\frac{2}{2+\delta}-2}( \lambda\bar{F})^{\frac{2}{2+\delta}-2}(y_n)\left\{1 + g^{-2}(y_n)(\lambda\bar{F})^{-\frac{2}{q}}(y_n)\right\}
    \\&\lesssim \alpha^{\frac{\delta}{2+\delta}}(|i_1-i_2|)(\lambda \bar{F})^{\frac{2}{2+\delta}-2}(y_n).
    \label{eqC33}
\end{align}
Lemma~\ref{newlemma}(i) and \cite[Theorem~3]{Doukhan1994} give that, for $n$ large enough, for any $i_2\ne i_2$, since $d<+\infty$,
\begin{align}
\nonumber
\cov(\chi^{(3)}_{i_1,n},\chi^{(3)}_{i_2,n}) &=\cov\left( \sum_{j=1}^d a_{3,j}\left(\frac{  \Lambda_{i_1}^{(j)}}{m_{\lambda^{(j)}}(y_n)} \mathds{1}_{\{ Y_{i_1}\ge y_n \}}-1\right),\sum_{j=1}^d a_{3,j}\left(\frac{  \Lambda_{i_2}^{(j)}}{m_{\lambda^{(j)}}(y_n)} \mathds{1}_{\{ Y_{i_2}\ge y_n \}}-1\right)\right)
\\ \nonumber&=\sum_{j_1,j_2=1}^d a_{3,j_1}a_{3,j_2} \frac{\cov(\Lambda_{i_1}^{(j_1)}\mathds{1}_{\{ Y_{i_1}\ge y_n \}} ,  \Lambda_{i_2}^{(j_2)}\mathds{1}_{\{ Y_{i_2}\ge y_n \}})}{m_{\lambda^{(j_1)}}(y_n)  m_{\lambda^{(j_2)}}(y_n)}
\\ \nonumber&\lesssim \alpha^{\frac{\delta}{2+\delta}}(|i_1-i_2|)\sum_{j_1,j_2=1}^d a_{3,j_1}a_{3,j_2}\frac{\max_{j\le d} \{m_{\lambda^{(j)}}(y_n)\}^{\frac{2}{2+\delta}}}{m_{\lambda^{(j_1)}}(y_n)  m_{\lambda^{(j_2)}}(y_n)} 
\\&\lesssim \alpha^{\frac{\delta}{2+\delta}}(|i_1-i_2|) (\lambda \bar{F})^{\frac{2}{2+\delta}-2}(y_n).
\label{eqC22}
\end{align}
Once more, \cite[Theorem~3]{Doukhan1994} herein with $p=q=2+\delta$ provides, for $\iota \in \{0,1\}$,
\begin{align}\label{eqC44}
\cov(\chi^{(4+\iota)}_{i_1,n},\chi^{(4+\iota)}_{i_2,n})&= 
 a_{4+\iota}^2 \frac{\cov(  Y^\iota_{i_1}\mathds{1}_{\{ Y_{i_1}\ge y_n \}} , Y^\iota_{i_2} \mathds{1}_{\{ Y_{i_2}\ge y_n \}})}{m^2_{Y^\iota}(y_n)}\le 
8a_{4+\iota}^2 \alpha^{\frac{\delta}{2+\delta}}(|i_1-i_2|) \bar{F}^{\frac{2}{2+\delta}-2}(y_n).
 \end{align}
Concerning the cross-covariances, one may write in virtue of \cite[Theorem~3]{Doukhan1994},
\begin{align*}
\cov(\chi^{(1)}_{i_1,n},\chi^{(3)}_{i_2,n})&=\cov\left( \sum_{j=1}^d a_{1,j}\left(\frac{  \Lambda_{i_1}^{(j)} X_{i_1}^{(j)}}{m_{\Lambda^{(j)} X^{(j)}}(y_n)} \mathds{1}_{\{ Y_{i_1}\ge y_n \}} -1\right),\sum_{j=1}^d a_{3,j}\left(\frac{  \Lambda_{i_2}^{(j)}}{m_{\lambda^{(j)}}(y_n)} \mathds{1}_{\{ Y_{i_2}\ge y_n \}}-1\right)\right)
\\&=\sum_{j_1,j_2=1}^d \alpha_{1,j_1}\alpha_{2,j_2}\frac{\cov( \Lambda_{i_1}^{(j_1)} X_{i_1}^{(j_1)}\mathds{1}_{\{ Y_{i_1}\ge y_n \}}, \Lambda_{i_2}^{(j_2)}\mathds{1}_{\{ Y_{i_2}\ge y_n \}})}{m_{\Lambda^{(j_1)} X^{(j_1)}}(y_n)  m_{\lambda^{(j_2)}}(y_n)}
\\&\le 8 \alpha^{\frac{\delta}{2+\delta}}(|i_1-i_2|) \sum_{j_1,j_2=1}^d a_{1,j_1} a_{1,j_2}\frac{\E(\Lambda^{(j_1)} (X^{(j_1)})^{2+\delta}  \mathds{1}_{\{ Y\ge y_n \}})^{\frac{1}{2+\delta}} \E(\Lambda^{(j_2)}  \mathds{1}_{\{ Y\ge y_n \}})^{\frac{1}{2+\delta}} }{m_{\Lambda^{(j_1)} X^{(j_1)}(y_n)}m_{\lambda^{(j_2)}}(y_n)}.
\end{align*}
Now, in view of \Mu, one can decompose the first expectation to get
\begin{align*}
\cov(\chi^{(1)}_{i_1,n},\chi^{(3)}_{i_2,n}) &\lesssim \alpha^{\frac{\delta}{2+\delta}}(|i_1-i_2|)\sum_{j_1,j_2=1}^d a_{1,j_1} a_{3,j_2}\frac{\max_{j\le d}\{\E(g^{2+\delta}(Y)\Lambda^{(j)}  \mathds{1}_{\{ Y\ge y_n \}} )\cdot m_{\lambda^{(j)}}(y_n)\}\}^{\frac{1}{2+\delta}}}{m_{\Lambda^{(j_1)} X^{(j_1)}(y_n)}m_{\lambda^{(j_2)}}(y_n)}
\\& +  \alpha^{\frac{\delta}{2+\delta}}(|i_1-i_2|)\sum_{j_1,j_2=1}^d a_{1,j_1}a_{3,j_2} \frac{\max_{j\le d}\{\E(\Lambda^{(j)} (\varepsilon^{(j)})^{2+\delta} \mathds{1}_{\{ Y\ge y_n \}} ) \cdot m_{\lambda^{(j)}}(y_n)\} \}^{\frac{1}{2+\delta}}}{m_{\Lambda^{(j_1)} X^{(j_1)}(y_n)}m_{\lambda^{(j_2)}}(y_n)}.
\end{align*}
By Lemma~\ref{newlemma}(i,iii) since $\tau<0$ and \textcolor{black}{$1<\gamma(q\kappa+\tau)$} for the denominator, by Lemma~\ref{denom1}(ii) under \textcolor{black}{$\gamma((2+\delta)\kappa+\tau)<1$} for the first numerator and Lemma~\ref{denom1}(v) for the second numerator, one may write, since $d<+\infty$,
\begin{align}
\nonumber
\cov(\chi^{(1)}_{i_1,n},\chi^{(3)}_{i_2,n}) &\lesssim \alpha^{\frac{\delta}{2+\delta}}(|i_1-i_2|) \left\{\frac{g(y_n)(\lambda \bar{F})^{\frac{2}{2+\delta}}(y_n)}{g(y_n)(\lambda \bar{F})^2(y_n)} + \frac{(\lambda\bar{F})^{\frac{2}{2+\delta}-\frac{2}{q}}(y_n)}{g(y_n)(\lambda \bar{F})^2(y_n)} \right\}
\\ \nonumber&= \alpha^{\frac{\delta}{2+\delta}}(|i_1-i_2|) \{(\lambda \bar{F})^{\frac{2}{2+\delta}-2}(y_n)+ g^{-1}(y_n)(\lambda\bar{F})^{\frac{2}{2+\delta}-\frac{2}{q}-2}(y_n) \}
\\&\lesssim \alpha^{\frac{\delta}{2+\delta}}(|i_1-i_2|) (\lambda \bar{F})^{\frac{2}{2+\delta}-2}(y_n),
\label{eqC12}
\end{align}
where we used \textcolor{black}{$1<\gamma(q\kappa/2+\tau)$}  for negligibility in the last line.
By \cite[Theorem~3]{Doukhan1994},
\begin{align*}
\cov(\chi^{(1)}_{i_1,n},\chi^{(2)}_{i_2,n})&=\cov\left( \sum_{j=1}^d a_{1,j}\left(\frac{  \Lambda_{i_1}^{(j)} X_{i_1}^{(j)}}{m_{\Lambda^{(j)} X^{(j)}}(y_n)} \mathds{1}_{\{ Y_{i_1}\ge y_n \}} -1\right),\sum_{j=1}^d a_{2,j}\left(\frac{  Y_{i_2} \Lambda_{i_2}^{(j)} X_{i_2}^{(j)}}{m_{Y\Lambda^{(j)}X^{(j)}}(y_n)} \mathds{1}_{\{ Y_{i_2}\ge y_n \}}-1\right)\right)
\\&=\sum_{j_1,j_2=1}^d \alpha_{1,j_1}\alpha_{2,j_2}\frac{\cov( \Lambda_{i_1}^{(j_1)} X_{i_1}^{(j_1)}\mathds{1}_{\{ Y_{i_1}\ge y_n \}}, Y_{i_2}\Lambda_{i_2}^{(j_2)}X^{(j_2)}_{i_2}\mathds{1}_{\{ Y_{i_2}\ge y_n \}})}{m_{\Lambda^{(j_1)} X^{(j_1)}}(y_n)  m_{Y\Lambda^{(j_2)}X^{(j_2)}}(y_n)}
\\&\le 8 \alpha^{\frac{\delta}{2+\delta}}(|i_1-i_2|) \sum_{j_1,j_2=1}^d a_{1,j_1} a_{3,j_2}\prod_{s\in \{0,1\}}\frac{\E(Y^{s}\Lambda^{(j_{s+1})} (X^{(j_{s+1})})^{2+\delta}  \mathds{1}_{\{ Y\ge y_n \}})^{\frac{1}{2+\delta}}}{m_{Y^{s}\Lambda^{(j_{s+1})} X^{(j_{s+1})}}(y_n)}
\\&=:  8 \alpha^{\frac{\delta}{2+\delta}}(|i_1-i_2|) \sum_{j_1,j_2=1}^d a_{1,j_1} a_{3,j_2}\prod_{s\in \{0,1\}}E_{j_s,\delta,n}.
\end{align*}
Now, in view of \Mu~and a Jensen-type argument, one can decompose:
\begin{align*}
\prod_{s\in \{0,1\}}E_{j_s,\delta,n} &\lesssim \frac{\{\E(g^{2+\delta}(Y)\Lambda^{(j_1)}  \mathds{1}_{\{ Y\ge y_n \}} )\cdot \E(Yg^{2+\delta}(Y)\Lambda^{(j_2)}  \mathds{1}_{\{ Y\ge y_n \}} )\}^{\frac{1}{2+\delta}}}{m_{\Lambda^{(j_1)} X^{(j_1)}(y_n)}m_{Y\Lambda^{(j_2)}X^{(j_2)} }(y_n)}
\\& +  \frac{\{\E(\Lambda^{(j_1)} (\varepsilon^{(j_1)})^{2+\delta} \mathds{1}_{\{ Y\ge y_n \}} ) \cdot \E(Y\Lambda^{(j_2)}  (\varepsilon^{(j_2)})^{2+\delta}\mathds{1}_{\{ Y\ge y_n \}} ) \}^{\frac{1}{2+\delta}}}{m_{\Lambda^{(j_1)} X^{(j_1)}(y_n)}m_{Y\Lambda^{(j_2)}X^{(j_2)} }(y_n)}
\\& + \frac{\{\E(\Lambda^{(j_1)} (\varepsilon^{(j_1)})^{2+\delta} \mathds{1}_{\{ Y\ge y_n \}} ) \cdot \E(Yg^{2+\delta}(Y)\Lambda^{(j_2)}\mathds{1}_{\{ Y\ge y_n \}} ) \}^{\frac{1}{2+\delta}}}{m_{\Lambda^{(j_1)} X^{(j_1)}(y_n)}m_{Y\Lambda^{(j_2)}X^{(j_2)} }(y_n)}
\\& + \frac{\{\E(g^{2+\delta}(Y)\Lambda^{(j_1)}\mathds{1}_{\{ Y\ge y_n \}} )\cdot \E(Y\Lambda^{(j_2)} (\varepsilon^{(j_2)})^{2+\delta} \mathds{1}_{\{ Y\ge y_n \}} )  \}^{\frac{1}{2+\delta}}}{m_{\Lambda^{(j_1)} X^{(j_1)}(y_n)}m_{Y\Lambda^{(j_2)}X^{(j_2)} }(y_n)}.
\end{align*}
By Lemma \ref{newlemma}(iii) under \textcolor{black}{$\gamma(\kappa+\tau+1)<1$}, the common denominator is, up to a multiplicative constant, asymptotically equal to $y_ng^2(y_n)(\lambda \bar F)^2(y_n)$. The first term is then $O\left(y_n^{\frac{1}{2+\delta}-1}(\lambda\bar{F})^{\frac{2}{2+\delta}-2}(y_n)\right)$ by Lemma \ref{denom1}(ii) under \textcolor{black}{$\gamma((2+\delta)\kappa + \tau + 1) < 1$}. The second term is $O\left( y_n^{\ff{2+\delta}-1}g^{-2}(y_n) (\lambda \bar F)^{\f{2}{2+\delta}-\f{2}{q}-2}(y_n)\right)$ thanks to Lemma \ref{denom1}(v) under \textcolor{black}{$\gamma( \frac{q}{q-2-\delta}+\tau) <1$}. Using Lemma \ref{denom1}(ii,v), the third and fourth terms are both $O\left( y_n^{\ff{2+\delta}-1}g^{-1}(y_n) (\lambda \bar F)^{\f{2}{2+\delta}-\f{1}{q}-2}(y_n)\right)$. Notice that the leading term is the first one under \textcolor{black}{$1<\gamma(q\kappa+\tau)$} while being bounded by $(\lambda\bar{F})^{\frac{2}{2+\delta}-2}(y_n)$. Therefore, it follows that, since $d<+\infty$,
\begin{align}
\cov(\chi^{(1)}_{i_1,n},\chi^{(2)}_{i_2,n}) &\lesssim \alpha^{\frac{\delta}{2+\delta}}(|i_1-i_2|) (\lambda \bar{F})^{\frac{2}{2+\delta}-2}(y_n).
\label{eqC13}
\end{align}
Once more, by \cite[Theorem~3]{Doukhan1994}, using  \Mu~and a Jensen-type argument, it yields:
\begin{align*}
\cov(\chi^{(3)}_{i_1,n},\chi^{(2)}_{i_2,n})&=\cov\left( \sum_{j=1}^d a_{3,j}\left(\frac{  \Lambda_{i_1}^{(j)} }{m_{\lambda^{(j)} }(y_n)} \mathds{1}_{\{ Y_{i_1}\ge y_n \}} -1\right),\sum_{j=1}^d a_{2,j}\left(\frac{  Y_{i_2} \Lambda_{i_2}^{(j)} X_{i_2}^{(j)}}{m_{Y\Lambda^{(j)}X^{(j)}}(y_n)} \mathds{1}_{\{ Y_{i_2}\ge y_n \}}-1\right)\right)
\\&=\sum_{j_1,j_2=1}^d a_{3,j_1}a_{2,j_2}\frac{\cov( \Lambda_{i_1}^{(j_1)} \mathds{1}_{\{ Y_{i_1}\ge y_n \}}, Y_{i_2}\Lambda_{i_2}^{(j_2)}X^{(j_2)}_{i_2}\mathds{1}_{\{ Y_{i_2}\ge y_n \}})}{m_{\lambda^{(j_1)} }(y_n)  m_{Y\Lambda^{(j_2)}X^{(j_2)}}(y_n)}
\\&\le 8 \alpha^{\frac{\delta}{2+\delta}}(|i_1-i_2|) \sum_{j_1,j_2=1}^d a_{3,j_1} a_{2,j_2}
\frac{m^{\ff{2+\delta}-1}_{\lambda^{(j_1)} }(y_n)}{ m_{Y\Lambda^{(j_2)} X^{(j_2)}(y_n)}}m_{Y\Lambda^{(j_2)} (X^{(j_2)})^{2+\delta}}^{\frac{1}{2+\delta}}(y_n)
\\&\lesssim  \alpha^{\frac{\delta}{2+\delta}}(|i_1-i_2|) \sum_{j_1,j_2=1}^d a_{3,j_1} a_{2,j_2}
\frac{m^{\ff{2+\delta}-1}_{\lambda^{(j_1)} }(y_n)}{ m_{Y\Lambda^{(j_2)} X^{(j_2)}(y_n)}}(m_{Y\Lambda^{(j_2)} g^{2+\delta}(Y)}^{\frac{1}{2+\delta}}(y_n)+m_{Y\Lambda^{(j_2)} (\varepsilon^{(j_2)})^{2+\delta}}^{\frac{1}{2+\delta}}(y_n))
\end{align*}
By Lemma \ref{newlemma}(i) and Lemma \ref{denom1}(ii,v), since $d<+\infty$, 
\begin{align*}
\cov(\chi^{(3)}_{i_1,n},\chi^{(2)}_{i_2,n})&\lesssim  \alpha^{\frac{\delta}{2+\delta}}(|i_1-i_2|)y_n^{\ff{2+\delta}-1} (\lambda \bar F)^{\f{2}{2+\delta}-2}(y_n) \big\{ 1 + g^{-1}(y_n) (\lambda \bar F)^{-\ff{q}}(y_n)  \big\}.
\end{align*}
Since \textcolor{black}{$1<\gamma(\kappa q+\tau)$}, we conclude that:
\begin{align}
\cov(\chi^{(3)}_{i_1,n},\chi^{(2)}_{i_2,n}) &\lesssim \alpha^{\frac{\delta}{2+\delta}}(|i_1-i_2|) (\lambda \bar{F})^{\frac{2}{2+\delta}-2}(y_n).
\label{eqC23}
\end{align}
It only remains to investigate the cases with $\chi^{(k)}_{i,n}$, $k\in \{4,5\}$. This time, we employ \cite[Corollary A.1]{hallheyde} with herein $p=2+\delta$ to get, using \Mu~and Jensen's inequality,
\begin{align*}
\cov(\chi^{(1)}_{i_1,n},\chi^{(4)}_{i_2,n})&=\cov\left( \sum_{j=1}^d a_{1,j}\left(\frac{  \Lambda_{i_1}^{(j)} X_{i_1}^{(j)}}{m_{\Lambda^{(j)} X^{(j)}}(y_n)} \mathds{1}_{\{ Y_{i_1}\ge y_n \}} -1\right),a_4\left(\frac{\mathds{1}_{\{ Y_{i_2}\ge y_n \}}}{\bar{F}(y_n)} -1\right)\right)
\\&=a_4\sum_{j=1}^d a_{1,j}\frac{\cov( \Lambda_{i_1}^{(j_1)} X_{i_1}^{(j_1)}\mathds{1}_{\{ Y_{i_1}\ge y_n \}}, \mathds{1}_{\{ Y_{i_2}\ge y_n \}})}{m_{\Lambda^{(j_1)} X^{(j_1)}}(y_n)  \bar{F}(y_n)}
\\&\le 6a_4 \alpha^{\frac{\delta}{2+\delta}}(|i_1-i_2|) \sum_{j=1}^d a_{1,j} \frac{\E(\Lambda^{(j)} (X^{(j)})^{2+\delta}  \mathds{1}_{\{ Y\ge y_n \}})^{\frac{1}{2+\delta}} }{m_{\Lambda^{(j_1)} X^{(j_1)}(y_n)}}\bar{F}^{-1}(y_n) 
\\&\lesssim \alpha^{\frac{\delta}{2+\delta}}(|i_1-i_2|) \bar{F}^{-1}(y_n) \frac{m_{\Lambda^{(j)} g^{2+\delta}(Y)}^{\frac{1}{2+\delta}}(y_n)+m_{\Lambda^{(j)} (\varepsilon^{(j)})^{2+\delta}}^{\frac{1}{2+\delta}}(y_n) }{m_{\Lambda^{(j_1)} X^{(j_1)}(y_n)}}.
\end{align*}
By Lemma \ref{denom1}(ii,v), Lemma \ref{newlemma}(iii) and the fact that \textcolor{black}{$1<\gamma(\kappa q + \tau)$}, it follows that,
\begin{align*}
\cov(\chi^{(1)}_{i_1,n},\chi^{(4)}_{i_2,n})&\lesssim \alpha^{\frac{\delta}{2+\delta}}(|i_1-i_2|) \bar{F}^{-1}(y_n) \frac{g(y_n) (\lambda \bar F)^{\ff{2+\delta}}(y_n)+ (\lambda \bar{F})^{\ff{2+\delta}-\ff{q}}(y_n)  }{g(y_n)(\lambda \bar F)(y_n)}
\\&\lesssim \alpha^{\frac{\delta}{2+\delta}}(|i_1-i_2|) \bar{F}^{-1}(y_n) (\lambda \bar F)^{\ff{2+\delta}-1}(y_n) \big\{ 1+ g^{-1}(\lambda \bar F)^{-\ff{q}}(y_n)\big\}
\\&\lesssim \alpha^{\frac{\delta}{2+\delta}}(|i_1-i_2|) \bar{F}^{-1}(y_n) (\lambda \bar F)^{\ff{2+\delta}-1}(y_n) .
\end{align*}
Again, by Lemma \ref{denom1}(i,ii,v) under \textcolor{black}{$(2+\delta)\gamma<1$}, \textcolor{black}{$\gamma((2+\delta)\kappa + \tau + 1)<1$} and \textcolor{black}{$\gamma(\frac{q}{q-2-\delta}+\tau)<1$}, for $\iota \in \{0,1\}$, we have that, using \Mu~and Jensen's inequality:
\begin{align*}
\cov(\chi^{(1+\iota)}_{i_1,n},\chi^{(5)}_{i_2,n})&=\cov\left( \sum_{j=1}^d a_{1+\iota,j}\left(\frac{ Y^\iota_{i_1}  \Lambda_{i_1}^{(j)} X_{i_1}^{(j)}}{m_{Y^\iota\Lambda^{(j)} X^{(j)}}(y_n)} \mathds{1}_{\{ Y_{i_1}\ge y_n \}} -1\right),a_5\left(\frac{Y_{i_2}\mathds{1}_{\{ Y_{i_2}\ge y_n \}}}{m_Y(y_n)} -1\right)\right)
\\&=a_5\sum_{j=1}^d a_{1+\iota,j}\frac{\cov( Y^\iota_{i_1}\Lambda_{i_1}^{(j_1)} X_{i_1}^{(j_1)}\mathds{1}_{\{ Y_{i_1}\ge y_n \}}, Y_{i_2}\mathds{1}_{\{ Y_{i_2}\ge y_n \}})}{m_{Y^\iota\Lambda^{(j_1)} X^{(j_1)}}(y_n)  m_Y(y_n)}
\\&\lesssim \alpha^{\frac{\delta}{2+\delta}}(|i_1-i_2|) \max_{1\le j \le d} \frac{\E(Y^\iota\Lambda^{(j)} (X^{(j)})^{2+\delta}  \mathds{1}_{\{ Y\ge y_n \}})^{\frac{1}{2+\delta}} m_{Y^{2+\delta}}^{\ff{2+\delta}}(y_n)}{m_{Y^\iota\Lambda^{(j_1)} X^{(j_1)}(y_n)}m_Y(y_n)}
\\&\lesssim \alpha^{\frac{\delta}{2+\delta}}(|i_1-i_2|)  \frac{m_{Y^{2+\delta}}^{\ff{2+\delta}}(y_n)}{m_Y(y_n)} \frac{m_{Y^\iota\Lambda^{(j)} g^{2+\delta}(Y)}^{\frac{1}{2+\delta}}(y_n)+m_{Y^\iota\Lambda^{(j)} (\varepsilon^{(j)})^{2+\delta}}^{\frac{1}{2+\delta}}(y_n) }{m_{Y^\iota\Lambda^{(j_1)} X^{(j_1)}(y_n)}}
\\&\lesssim \alpha^{\frac{\delta}{2+\delta}}(|i_1-i_2|) \bar{F}^{\frac{1}{2+\delta}-1}(y_n)y_n^{\iota(\ff{2+\delta}-1)} \frac{g(y_n) (\lambda \bar F)^{\ff{2+\delta}}(y_n)+ (\lambda \bar{F})^{\ff{2+\delta}-\ff{q}}(y_n)  }{g(y_n)(\lambda \bar F)(y_n)}
\\&\lesssim \alpha^{\frac{\delta}{2+\delta}}(|i_1-i_2|) y_n^{\iota(\ff{2+\delta}-1)} (\lambda \bar{F}^2)^{\ff{2+\delta}-1}(y_n)
.
\end{align*}
Using in addition Lemma \ref{newlemma}(i), we get:
\begin{align*}
\cov(\chi^{(4)}_{i_1,n},\chi^{(5)}_{i_2,n})&=\cov\left( a_4\left(\frac{  \Lambda_{i_1}^{(j)} }{m_{\lambda^{(j)}}(y_n)} \mathds{1}_{\{ Y_{i_1}\ge y_n \}} -1\right),a_5\left(\frac{ Y_{i_2}\mathds{1}_{\{ Y_{i_2}\ge y_n \}}}{m_{Y}(y_n)} -1\right)\right)
\\&=a_4a_5\frac{\cov( \Lambda_{i_1}^{(j_1)} \mathds{1}_{\{ Y_{i_1}\ge y_n \}},  Y_{i_2}\mathds{1}_{\{ Y_{i_2}\ge y_n \}})}{m_{\lambda^{(j_1)} }(y_n)  m_{Y}(y_n)}
\lesssim \alpha^{\frac{\delta}{2+\delta}}(|i_1-i_2|)  \frac{ m^{\frac{1}{2+\delta}}_{Y^{2+\delta}}(y_n)}{m_{Y}(y_n)}m^{\frac{1}{2+\delta}-1}_{\lambda^{(j)}}(y_n)
\\&\lesssim \alpha^{\frac{\delta}{2+\delta}}(|i_1-i_2|) \bar{F}^{\frac{1}{2+\delta}-1}(y_n) (\lambda \bar{F})^{\frac{1}{2+\delta}-1}(y_n)\lesssim \alpha^{\frac{\delta}{2+\delta}}(|i_1-i_2|)  (\lambda \bar{F}^2)^{\frac{1}{2+\delta}-1}(y_n).
\end{align*}
We may conclude that, since $\tau<0$,
\begin{align}\label{eqC14}
\cov(\chi^{(1)}_{i_1,n},\chi^{(4)}_{i_2,n})\vee\cov(\chi^{(1)}_{i_1,n},\chi^{(5)}_{i_2,n})\vee \cov(\chi^{(4)}_{i_1,n},\chi^{(5)}_{i_2,n})&\lesssim \alpha^{\frac{\delta}{2+\delta}}(|i_1-i_2|)  (\lambda \bar F)^{\f{2}{2+\delta}-2}(y_n) .
\end{align}
Samewise, using \Mu~and Jensen's inequality,
\begin{align*}
\cov(\chi^{(2)}_{i_1,n},\chi^{(4)}_{i_2,n})&=\cov\left( \sum_{j=1}^d a_{2,j}\left(\frac{ Y_{i_1}  \Lambda_{i_1}^{(j)} X_{i_1}^{(j)}}{m_{Y\Lambda^{(j)} X^{(j)}}(y_n)} \mathds{1}_{\{ Y_{i_1}\ge y_n \}} -1\right),a_4\left(\frac{\mathds{1}_{\{ Y_{i_2}\ge y_n \}}}{\bar{F}(y_n)} -1\right)\right)
\\&=a_4\sum_{j=1}^d a_{2,j}\frac{\cov( Y_{i_1}\Lambda_{i_1}^{(j_1)} X_{i_1}^{(j_1)}\mathds{1}_{\{ Y_{i_1}\ge y_n \}}, \mathds{1}_{\{ Y_{i_2}\ge y_n \}})}{m_{Y\Lambda^{(j_1)} X^{(j_1)}}(y_n)  \bar{F}(y_n)}
\\&\le 6a_4 \alpha^{\frac{\delta}{2+\delta}}(|i_1-i_2|) \sum_{j=1}^d a_{2,j} \frac{\E(Y\Lambda^{(j)} (X^{(j)})^{2+\delta}  \mathds{1}_{\{ Y\ge y_n \}})^{\frac{1}{2+\delta}} }{m_{Y\Lambda^{(j_1)} X^{(j_1)}(y_n)}}\bar{F}^{-1}(y_n) 
\\&\lesssim \alpha^{\frac{\delta}{2+\delta}}(|i_1-i_2|) \bar{F}^{-1}(y_n) \frac{m_{Y\Lambda^{(j)} g^{2+\delta}(Y)}^{\frac{1}{2+\delta}}(y_n)+m_{Y\Lambda^{(j)} (\varepsilon^{(j)})^{2+\delta}}^{\frac{1}{2+\delta}}(y_n) }{m_{Y\Lambda^{(j_1)} X^{(j_1)}(y_n)}}
\\&\lesssim \alpha^{\frac{\delta}{2+\delta}}(|i_1-i_2|) \bar{F}^{-1}(y_n) y_n^{\ff{2+\delta}-1} \frac{g(y_n) (\lambda \bar F)^{\ff{2+\delta}}(y_n)+ (\lambda \bar{F})^{\ff{2+\delta}-\ff{q}}(y_n)  }{g(y_n)(\lambda \bar F)(y_n)}
\\&\lesssim \alpha^{\frac{\delta}{2+\delta}}(|i_1-i_2|) \bar{F}^{-1}(y_n) y_n^{\ff{2+\delta}-1}(\lambda \bar F)^{\ff{2+\delta}-1}(y_n) \big\{ 1+ g^{-1}(\lambda \bar F)^{-\ff{q}}(y_n)\big\}
\\&\lesssim \alpha^{\frac{\delta}{2+\delta}}(|i_1-i_2|) \bar{F}^{-1}(y_n)y_n^{\ff{2+\delta}-1} (\lambda \bar F)^{\ff{2+\delta}-1}(y_n) .
\end{align*}
On the other hand, we have according to \cite[Theorem 3]{Doukhan1994}, 
\begin{align*}
\cov(\chi^{(3)}_{i_1,n},\chi^{(5)}_{i_2,n})&=\cov\left( \sum_{j=1}^d a_{3,j}\left(\frac{  \Lambda_{i_1}^{(j)} }{m_{\lambda^{(j)} }(y_n)} \mathds{1}_{\{ Y_{i_1}\ge y_n \}} -1\right),a_5\left(\frac{Y_{i_2}\mathds{1}_{\{ Y_{i_2}\ge y_n \}}}{m_Y(y_n)} -1\right)\right)
\\&=a_5\sum_{j=1}^d a_{3,j}\frac{\cov( \Lambda_{i_1}^{(j)} \mathds{1}_{\{ Y_{i_1}\ge y_n \}}, Y_{i_2}\mathds{1}_{\{ Y_{i_2}\ge y_n \}})}{m_{\lambda^{(j)} }(y_n)  m_Y(y_n)}
\\&\lesssim \alpha^{\frac{\delta}{2+\delta}}(|i_1-i_2|) \frac{m_{Y^{2+\delta}}^{\ff{2+\delta}}(y_n)}{m_Y(y_n)}  m_{\lambda^{(j)}}^{\frac{1}{2+\delta}-1}(y_n) 
\lesssim  \alpha^{\frac{\delta}{2+\delta}}(|i_1-i_2|)  (\lambda \bar {F}^2)^{\ff{2+\delta}-1}(y_n).
\end{align*}
Combining the previous results, it follows that, since $\tau<0$, \begin{align}\label{eqC34}
\cov(\chi^{(2)}_{i_1,n},\chi^{(4)}_{i_2,n})\vee \cov(\chi^{(3)}_{i_1,n},\chi^{(5)}_{i_2,n}) \vee\cov(\chi^{(2)}_{i_1,n},\chi^{(5)}_{i_2,n})&\lesssim \alpha^{\frac{\delta}{2+\delta}}(|i_1-i_2|)  (\lambda \bar F)^{\f{2}{2+\delta}-2}(y_n) .
\end{align}
In the same fashion as before, 
\cite[Corollary~A.1]{hallheyde} with herein $p=2+\delta$ yields: 
\begin{align*}
\cov(\chi^{(3)}_{i_1,n},\chi^{(4)}_{i_2,n})&=\cov\left( \sum_{j=1}^d a_{3,j}\left(\frac{  \Lambda_{i_1}^{(j)} }{m_{\lambda^{(j)} }(y_n)} \mathds{1}_{\{ Y_{i_1}\ge y_n \}} -1\right),a_4\left(\frac{\mathds{1}_{\{ Y_{i_2}\ge y_n \}}}{\bar{F}(y_n)} -1\right)\right)
\\&=a_4\sum_{j=1}^d a_{3,j}\frac{\cov( \Lambda_{i_1}^{(j)} \mathds{1}_{\{ Y_{i_1}\ge y_n \}}, \mathds{1}_{\{ Y_{i_2}\ge y_n \}})}{m_{\lambda^{(j)} }(y_n)  \bar{F}(y_n)}
\\&\le 6a_4 \alpha^{\frac{\delta}{2+\delta}}(|i_1-i_2|) \bar{F}^{-1}(y_n)\sum_{j=1}^d a_{3,j} m_{\lambda^{(j)}}^{\frac{1}{2+\delta}-1}(y_n) 
\\&\lesssim a_4 \alpha^{\frac{\delta}{2+\delta}}(|i_1-i_2|)  \bar{F}^{-1}(y_n) (\lambda \bar F)^{\ff{2+\delta}-1}(y_n).
\end{align*}
Therefore, it yields since $\tau<0$,
\begin{align}\label{eqC24}
\cov(\chi^{(3)}_{i_1,n},\chi^{(4)}_{i_2,n})&\lesssim \alpha^{\frac{\delta}{2+\delta}}(|i_1-i_2|)  (\lambda \bar F)^{\f{2}{2+\delta}-2}(y_n) .
\end{align}
Combining the bounds~\eqref{eqC33}--\eqref{eqC24}, it follows that, when $n$ is large enough and for $u_n\le n$,
\begin{align}\label{eq:cov_un}
\sum_{1\le i_1<i_2\le u_n} \cov(\chi_{i_1,n},\chi_{i_2,n}) &\lesssim \frac{1}{n}(\lambda \bar{F})^{\frac{2}{2+\delta}-2+\theta}(y_n) \sum_{1\le i_1<i_2\le u_n} \alpha^{\frac{\delta}{2+\delta}}(|i_1-i_2|) .
\end{align}
Furthermore, the following inequality holds under~\eqref{hyp:mixing},
\begin{align}\label{eq:sum_alpha}
     \sum_{1\le i_1<i_2\le u_n} \alpha^{\frac{\delta}{2+\delta}}(i_2-i_1)&\lesssim u_n.
\end{align} 
 Indeed, a derivation writes as,
\begin{align*}
\sum_{1\le i_1<i_2\le u_n} \alpha^{\frac{\delta}{2+\delta}}(i_2-i_1)&= \sum_{i=1}^{u_n}\sum_{j=1}^{u_n-i}\alpha^{\frac{\delta}{2+\delta}}(j)
= \sum_{i=1}^{u_n}\sum_{j=1}^{u_n-i}\alpha^{\frac{\delta}{2+\delta}}(j) = \sum_{i,j} \alpha^{\frac{\delta}{2+\delta}}(j)\mathds{1}_{\{ 1\le i \le u_n, 1\le j \le u_n-i \}}
\\&= \sum_{j=1}^{u_n-1}\sum_{i=1}^{u_n-j}\alpha^{\frac{\delta}{2+\delta}}(j) = \sum_{i,j} \alpha^{\frac{\delta}{2+\delta}}(j)\mathds{1}_{\{ 1\le i \le u_n - j, 1\le j \le u_n-1 \}}
\\&= \sum_{j=1}^{u_n-1}(u_n-j)\alpha^{\frac{\delta}{2+\delta}}(j) \le  u_n\sum_{j=1}^{u_n-1}\alpha^{\frac{\delta}{2+\delta}}(j).
\end{align*}
Once that the control of the variance~\eqref{eq:var} and of the covariance~\eqref{eq:cov_un} are established, we may go back to the conditions~\eqref{eq:rec1}--\eqref{eq:lyapunov} required for applying \cite[Lemma~C.7]{DavisonPadoanStupfler2023}. Firstly, in the case $\theta=\theta_0$, the asymptotic~\eqref{eq:var} implies that $ n\var( \chi_{1,n})\ll 1$. Moreover, the case $\delta=0$, namely $\theta_0=1$, corresponds to $n\var( \chi_{1,n})$ being convergent. 

Secondly,~\eqref{hyp:mixing},~\eqref{eq:var} and~\eqref{eq:cov_un} readily yield the first statement~\eqref{eq:rec1} and the third one~\eqref{eq:rec3} since $n-r_n\lfloor n/r_n \rfloor\le r_n$. The second statement~\eqref{eq:rec2} also follows from~\eqref{hyp:mixing}, ~\eqref{eq:var},~\eqref{eq:cov_un} and \eqref{eq:sum_alpha} with $u_n=r_n$ since we have shown that:
\begin{align*}
   \frac{n}{r_n}\sum_{1\le i_1<i_2\le r_n}\cov( \chi_{i_1,n},\chi_{i_2,n}) \lesssim (\lambda \bar{F})^{\frac{2}{2+\delta}-2+\theta}(y_n)= (\lambda \bar{F})^{\theta-\theta_0}(y_n) & <\infty.
\end{align*}
When $\theta>\theta_0 :=1+\frac{\delta}{2+\delta}$, the latest bound in the previous display is converging to zero, since $\tau<0$, which in turn implies that the limiting variance is zero. Furthermore, it is not mandatory to verify the Lyapunov condition~\eqref{eq:lyapunov} according to the statement of \cite[Lemma~C.7]{DavisonPadoanStupfler2023}. As such, the proof in the case $\theta>\theta_0$ is finished.

It only remains to fulfill the Lyapunov condition~\eqref{eq:lyapunov} in the case $\theta=\theta_0$. Let us start with the triangular inequality for the $L^{2+\delta}(\Omega)$-norm and Jensen's inequality to write 
\begin{align*}
\mathbb{E}\Big(|\sum_{i=1}^{r_n} \chi_{i,n}|^{2+\delta}\Big) &=  \| \sum_{i=1}^{r_n} \chi_{i,n} \|^{2+\delta}_{L^{2+\delta}(\Omega)} \le \sum_{i=1}^{r_n} \|  \chi_{i,n} \|^{2+\delta}_{L^{2+\delta}(\Omega)}
= \frac{(\lambda\bar{F})^{\theta}(y_n)}{n} \sum_{i=1}^{r_n}  \sum_{k=1}^5 \|  \chi^{(k)}_{i,n} \|^{2+\delta}_{L^{2+\delta}(\Omega)}.
\end{align*}
Denote in the sequel $\rho_n=(\lambda\bar{F})^{\theta}(y_n)/n$.
Therefore, the homogeneity of the norm and the triangular inequality entail, for $n$ large enough,
\begin{align*}
&\mathbb{E}\Big(|\sum_{i=1}^{r_n} \chi_{i,n}|^{2+\delta}\Big)\\
&\le   \rho_n^{1+\delta/2} \sum_{i=1}^{r_n} \sum_{j=1}^d \left(a_{1,j}^{2+\delta} \left \|\frac{  \Lambda_i^{(j)} X_{i}^{(j)}}{m_{\Lambda^{(j)} X^{(j)}}(y_n)} \indin  \right\|^{2+\delta}_{L^{2+\delta}(\Omega)}+a_{2,j}^{2+\delta}\left \|\frac{  Y_i\Lambda_i^{(j)} X_{i}^{(j)}}{m_{Y\Lambda^{(j)} X^{(j)}}(y_n)} \indin \right \|^{2+\delta}_{L^{2+\delta}(\Omega)} \right) 
\\&+ \rho_n^{1+\delta/2} \sum_{i=1}^{r_n} 
a_5^{2+\delta} \left \|\frac{Y_i\indin}{m_{Y}(y_n)}  \right\|^{2+\delta}_{L^{2+\delta}(\Omega)}+  \rho_n^{1+\delta/2} \sum_{i=1}^{r_n}\left(2a_4^{2+\delta}+2a_5^{2+\delta}+ \sum_{j=1}^d (a_{1,j}^{2+\delta}+2a_{3,j}^{2+\delta}+a_{2,j}^{2+\delta})\right)
\\&\lesssim   \rho_n^{1+\delta/2} \sum_{i=1}^{r_n} \sum_{j=1}^d \left(a_{1,j}^{2+\delta}\frac{m_{\Lambda^{(j)} (X^{(j)})^{2+\delta}  }(y_n)}{m^{2+\delta}_{\Lambda^{(j)} X^{(j)}}(y_n)}+a_{2,j}^{2+\delta}\frac{m_{\Lambda^{(j)} (YX^{(j)})^{2+\delta}  }(y_n)}{m^{2+\delta}_{Y\Lambda^{(j)} X^{(j)}}(y_n)}  \right)+\rho_n^{1+\delta/2} \sum_{i=1}^{r_n} a_5^{2+\delta} \frac{m_{Y^{2+\delta}}(y_n)}{m^{2+\delta}_Y(y_n)} + r_n\rho_n^{1+\delta/2} 
\\&\lesssim   \rho_n^{1+\delta/2} \sum_{i=1}^{r_n} \sum_{j=1}^d a_{1,j}^{2+\delta}\frac{(m_{ g^{2+\delta}(Y_i)\Lambda^{(j)}  }+m_{\Lambda^{(j)} (\varepsilon^{(j)})^{2+\delta}  })(y_n)}{m^{2+\delta}_{\Lambda^{(j)} X^{(j)}}(y_n)}
\\&+ \rho_n^{1+\delta/2} \sum_{i=1}^{r_n} \sum_{j=1}^d a_{2,j}^{2+\delta}\frac{(m_{ Y^{2+\delta}g^{2+\delta}(Y)\Lambda^{(j)}  }+m_{\Lambda^{(j)} (Y\varepsilon^{(j)})^{2+\delta}  })(y_n)}{m^{2+\delta}_{Y\Lambda^{(j)} X^{(j)}}(y_n)}   + r_n\rho_n^{1+\delta/2} \bar{F}^{-(1+\delta)}(y_n)  +r_n\rho_n^{1+\delta/2} .
\end{align*}
By Lemma~\ref{denom1}(ii) under \textcolor{black}{$\gamma((2+\delta)(\kappa+1)+\tau)<1$} and Lemma~\ref{denom1}(v) under \textcolor{black}{$\gamma(\frac{q}{q-2-\delta}(2+\delta)+\tau)<1$} for the numerators, and by Lemma \ref{newlemma}(i,iii) for the denominators, it follows that
\begin{align*}
\mathbb{E}\Big(|\sum_{i=1}^{r_n} \chi_{i,n}|^{2+\delta}\Big)&\lesssim  r_n \rho_n^{1+\delta/2}  \left\{(\lambda \bar{F})^{-(1+\delta)}(y_n) + g^{-(2+\delta)}(y_n) (\lambda \bar F)^{-(1+\delta)-(2+\delta)/q}(y_n)\right\} 
\\&\lesssim  r_n \rho_n^{1+\delta/2}  (\lambda \bar{F})^{-(1+\delta)}(y_n) = n^{-\delta/2} \ll 1 ,\end{align*} 
where we used in the last line the fact that \textcolor{black}{$1<\gamma(\kappa q + \tau)$} and $\theta_0(1+\delta/2)=1+\delta$. This concludes~\eqref{eq:lyapunov} and thus the proof.


\end{proof}


\begin{proof}[\textbf{\textup{Proof of Theorem~\ref{theo-princ}}}]
We start by considering the expansion $\hat{\beta}_{\Lambda}(y_n)  -  \beta = (\hat{\beta}_{\Lambda}(y_n)  -  w(y_n) ) +  ( w(y_n) - \beta )$ where $w$ is given in \eqref{solution}. Then, by triangle inequality and in view of $ g(y_n)\bar{F}^{1/q}(y_n) \| w(y_n) - \beta\| = O(1)$ provided by \cite[Proposition~2]{Bousebata2023} under \textcolor{black}{$\gamma(\kappa+1)<1$}, it is sufficient to show that, for any $\theta > 1+\frac{\delta}{2+\delta}$, for some Gaussian random vector $G$,
\begin{align}\label{eq:cv_aux}
    n^{1/2}(\lambda\bar{F})^{\theta/2}(y_n) \left(\hat \beta_{\Lambda}(y_n) -  w(y_n)\right)& \xrightarrow[n\to +\infty]{\mathbb{P}} \beta \odot G,\quad \text{in $\mathbb{R}^{p}$.}
\end{align}
Let us introduce
$$\psi_1:= \frac{1-\gamma\tau}{1-\gamma(\kappa+\tau+1)}- \frac{1-\gamma\tau}{(1-\gamma)(1-\gamma(\kappa+\tau))}\mbox{ and }\psi_2 :=  \frac{\gamma^2 \kappa}{(1-\gamma(\kappa+1))(1-\gamma\kappa)(1-\gamma)}.
$$
Let $j\in\{1,\dots,d\}$. Applying the Delta-method together with Proposition~\ref{prop-loi-jointe-alpha-mixing} and the identity given in \eqref{esti}, there exists a sequence of real random variables $(\vartheta_{j,n})_j$ such that $n^{1/2}(\lambda\bar{F})^{\theta/2}(y_n)(\vartheta_{j,n})_j$ is asymptotically a centered Gaussian vector $G\in \R^d$ when $\theta=1+\frac{\delta}{2+\delta}$; in the case $\theta>1+\frac{\delta}{2+\delta}$, the variance of $G$ is zero, \emph{i.e.}, $n^{1/2}(\lambda\bar{F})^{\theta/2}(y_n)\vartheta_{j,n}\ll 1$ in probability. Additionally,
\begin{align*}
     \hat v^{(j)}_{\Lambda}(y_n)&= \frac{\bar{F}(y_n)}{{m}_{\lambda^{(j)}}(y_n)}\left({\bar{F}}(y_n) {m}_{Y \Lambda^{(j)}X^{(j)}}(y_n) - {m}_{Y}(y_n){m}_{\Lambda^{(j)} X^{(j)}}(y_n)\right) (1+\vartheta_{j,n}) .
     \end{align*}
Next, using Lemma~\ref{denom1}(i), Lemma~\ref{newlemma}(i,iii) and Lemma~\ref{newlemma}(ii) under {\color{black} $\gamma(\kappa+1)<1$}, one obtains, since $\beta_j\ne 0$ when $j\in\{1,\dots,d\}$:
\begin{align*}
     \hat v^{(j)}_{\Lambda}(y_n)&= \psi_1 \beta_j  y_ng(y_n)\bar{F}(y_n)(1+\vartheta_{j,n}) .
\end{align*}
Furthermore, going back to~\eqref{esti} and expanding $\| \hat v_{\Lambda}(y_n) \|^2 = \hat v_{\Lambda}(y_n)^\top \hat v_{\Lambda}(y_n)$ thanks to Proposition~\ref{prop-loi-jointe-alpha-mixing}, Lemma~\ref{denom1}(i) and Lemma~\ref{newlemma}(i,iii) yield
\begin{align}\label{eq:aux_norm}
    \| \hat v_{\Lambda}(y_n) \| &= \psi_1 y_ng(y_n)\bar{F}(y_n)(1+o_{\p}(1)), 
\end{align}
since $ \beta  $ is of unit norm, 
Besides, by Lemma~\ref{denom1}(i) and Lemma~\ref{newlemma}(ii) under {\color{black} $\gamma(\kappa+1)<1$}, \begin{align*}
    v^{(j)}(y_n) &=  \psi_2  \beta_jy_ng(y_n)\bar{F}^2(y_n)(1+o(1)) \mbox{ and } 
 \|  v(y_n) \| = \psi_2 y_ng(y_n)\bar{F}^2(y_n)(1+o(1)) .
 \end{align*}
Combining the last three displays, it yields:
\begin{align*}
    n^{1/2}(\lambda\bar{F})^{\theta/2}(y_n)\left( \frac{\hat v^{(j)}_{\Lambda}(y_n)}{\| \hat v_{\Lambda}(y_n) \|} -  \frac{v^{(j)}(y_n)}{\| v(y_n) \|}\right)_{1\le j \le d}&=   n^{1/2}(\lambda\bar{F})^{\theta/2}(y_n)(\beta_j \vartheta_{j,n})_{1\le j \le d}\,  (1+o_\p(1)) .
\end{align*}
Now, let $d+1\le j \le p$, meaning that $\beta_j=0$. It remains to show that for any such $j$ and any $\theta\ge 1+\f{\delta}{2+\delta}$, \begin{align}\label{eq:cv_aux2}
    n^{1/2}(\lambda\bar{F})^{\theta/2}(y_n) \frac{\hat{v}^{(j)}_{\Lambda}(y_n)}{\| \hat{v}_{\Lambda}(y_n)\| }& \xrightarrow[n\to +\infty]{\mathbb{P}} 0.
\end{align}
When $\beta_j=0$, using Proposition \ref{prop-loi-jointe-alpha-mixing} with \eqref{esti}, one may write:
\begin{align*}
     \hat v^{(j)}_{\Lambda}(y_n)&=\frac{1}{{m}_{\lambda^{(j)}}(y_n)}\left({\bar{F}}(y_n) \hat{m}_{Y \Lambda^{(j)}\varepsilon^{(j)}}(y_n) - \hat{m}_{Y}(y_n)\hat{m}_{\Lambda^{(j)} \varepsilon^{(j)}}(y_n)\right) (1+O_{\p}(\{n(\lambda \bar F)^\theta(y_n)\}^{1/2})) .
     \end{align*}
The tool to control the random terms in the latter display will be Tchebychev's
inequality. As such, we begin by studying the variance of each quantities. First, $\var(\hat{m}_{Y^\iota \Lambda^{(j)}\varepsilon^{(j)}}(y_n)) = \var( \ff{n}\sum_{i=1}^n  Y^\iota_i  \Lambda^{(j)}\varepsilon^{(j)}\indin)$ may be decomposed as
\begin{align*}
  &\ff{n} \var( Y^\iota  \Lambda^{(j)}\varepsilon^{(j)} \indn) + \f{2}{n^2}\sum_{i_1<i_2} \cov(Y^\iota_{i_1} \Lambda^{(j)}_{i_1}\varepsilon^{(j)}_{i_1} \mathds{1}_{ \{  Y_{i_1}\geq y_n\} },Y^\iota_{i_2} \Lambda^{(j)}_{i_2}\varepsilon^{(j)}_{i_2} \mathds{1}_{ \{  Y_{i_2}\geq y_n\} }).
\end{align*}
By \cite[Theorem~3]{Doukhan1994} and Lemma \ref{denom1}(v) under \textcolor{black}{$\gamma(\frac{q}{q-2-\delta}(2+\delta)+\tau)<1$}, we may bound the double sum of covariances and the variance as follows, using \eqref{eq:sum_alpha},
\begin{align*}
&
\sum_{i_1<i_2}\alpha^{\frac{\delta}{2+\delta}}(|i_1-i_2|)  \E(  \Lambda^{(j)} (Y^\iota \varepsilon^{(j)})^{2+\delta}  \mathds{1}_{\{ Y\ge y_n \}})^{\frac{2}{2+\delta}}
\lesssim n y_n^{2\iota}(\lambda\bar F)^{\frac{2}{2+\delta}-\f{2}{q}}(y_n),
\\
     &\var( Y^\iota  \Lambda^{(j)}\varepsilon^{(j)} \ind)\le  m_{ \Lambda^{(j)}(Y^\iota \varepsilon^{(j)})^2}(y_n)
          \lesssim y_n^{2\iota}(\lambda\bar F)^{1-2/q}(y_n).
 \end{align*} 
 Thus, by comparison, one gets: \begin{align*}
    \var(\hat{m}_{Y^\iota \Lambda^{(j)}\varepsilon^{(j)}}(y_n)) &\lesssim \ff{n} y_n^{2\iota}(\lambda\bar F)^{\frac{2}{2+\delta}-\f{2}{q}}(y_n).
\end{align*}
Finally, using Tchebychev's
inequality together with $\E(\hat{m}_{Y^\iota \Lambda^{(j)}\varepsilon^{(j)}}(y_n)) =\E(Y^\iota \Lambda^{(j)}\varepsilon^{(j)} \indn) $ upon which we apply Lemma~\ref{denom1}(iv), and $n(\lambda\bar{F})^{\theta}(y_n)\gg 1$ for the comparison, it follows
\begin{align*}
  \hat{m}_{Y^\iota \Lambda^{(j)}\varepsilon^{(j)}}(y_n)  &= \E(Y^\iota \Lambda^{(j)}\varepsilon^{(j)} \indn)+ O\left( \ff{\sqrt{n}} y_n^{\iota}(\lambda\bar F)^{\frac{1}{2+\delta}-\f{1}{q}}(y_n)\right)
  \\&=O\left( y_n^\iota (\lambda \bar F)^{1-1/q}(y_n)\right) .
\end{align*}
Besides, another application of Tchebychev's
inequality, together with a similar argument as above leading to $\var(\hat m_Y(y_n))\lesssim y_n^{2+\delta} \bar{F}^{\frac{2}{2+\delta}}(y_n) $, gives after comparison,
\begin{align*}
  \hat{m}_{Y}(y_n)  &= O\left( y_n \bar{F}(y_n)\right) + O\left( \ff{\sqrt{n}} y_n^{1+\f{\delta}{2}}\bar{F}^{\frac{1}{2+\delta}}(y_n)\right)= O\left( y_n \bar{F}(y_n)\right) .
\end{align*}
Therefore, under \Mu, plugging the above asymptotics in~\eqref{esti} entails, using also~\eqref{eq:aux_norm} and Lemma~\ref{newlemma}(iii),
\begin{align*}
  n^{1/2}(\lambda\bar{F})^{\theta/2}(y_n)  \f{\hat v^{(j)}_{\Lambda}(y_n)}{ \| \hat v_{\Lambda}(y_n) \|}&= n^{1/2}(\lambda\bar{F})^{\theta/2}(y_n) O_{\mathbb{P}}\left( g^{-1}(y_n)(\lambda\bar{F})^{-1/q}(y_n)\right).    
    \end{align*}
    Note that the quantity in the brackets converges to zero under \textcolor{black}{$1<\gamma(\kappa q +\tau)$}.
    The assumption that \textcolor{black}{$n g^{-2}(y_n)(\lambda \bar F)(y_n)^{\theta-2/q}\ll 1$} concludes the proof in the case $j\in \{d+1,\dots,p\}$, regardless of $\theta=1+\f{\delta}{2+\delta}$ or not.
\end{proof}

\bibliographystyle{chicago}
\bibliography{mybib}

\vspace*{1cm}
\paragraph{\bf Acknowledgments} 
This work is supported by the French National Research Agency (ANR), grant EXSTA - ANR-23-CE40-0009. S.~Girard also acknowledges the support of the Chair Stress Test, Risk Management and Financial Steering, led by the French Ecole Polytechnique and its Foundation and sponsored by BNP Paribas.

\end{document}